\documentclass{article}              %  final version
\usepackage{arxiv}

\input{./settings}

\newcommand{\EE}[1]{\mathbb{E}\left[ #1 \right]}
\newcommand{\Var}[1]{\mathbb{V}\left[ #1 \right]}

\DeclareMathOperator*{\argmin}{arg\,min}

\newcommand{\approxsym}[1]{\widehat{#1}}
\newcommand{\biasedtxt}{\text{biased}}

\newcommand{\Cov}{\mathbb{C}\text{ov}}

\newcommand{\mone}{\mu_1}
\newcommand{\moneapprox}{\approxsym{\mu}_1}
\newcommand{\cmtwo}{\mu_2}
\newcommand{\cmtwoapprox}{\approxsym{\mu}_2}

\newcommand{\cmfour}{\mu_4}
\newcommand{\cmfourapprox}{\approxsym{\mu}_4}
\newcommand{\cmfourapproxbiased}{\approxsym{\mu}_{4 ,{\biasedtxt}}} %doesn't follow the chain definition due to double subscript
\newcommand{\sigmaapprox}{\approxsym{\sigma}_{\biasedtxt}}

\newcommand{\textml}{\text{\tiny ML}}
\newcommand{\Ql}{{Q_{\ell}}}
\newcommand{\Qlmone}{{Q_{\ell - 1}}}
\newcommand{\Qli}{Q_{\ell}^{(i)}}
\newcommand{\Qlmonei}{Q_{\ell - 1}^{(i)}}
\newcommand{\Qlj}{Q_{\ell}^{(j)}}
\newcommand{\Qlmonej}{Q_{\ell - 1}^{(j)}}
\newcommand{\Qlmonek}{Q_{\ell - 1}^{(k)}}
\newcommand{\Qlmoneh}{Q_{\ell - 1}^{(h)}}
\newcommand{\Qlone}{Q_{\ell_1}}
\newcommand{\Qltwo}{Q_{\ell_2}}
\newcommand{\Qlthree}{Q_{\ell_3}}

\newcommand{\Qlonei}{Q_{\ell_1}^{(i)}}
\newcommand{\Qltwoi}{Q_{\ell_2}^{(i)}}
\newcommand{\Qlthreei}{Q_{\ell_3}^{(i)}}

\newcommand{\Qltwoj}{Q_{\ell_2}^{(j)}}

\newcommand{\Qlthreek}{Q_{\ell_3}^{(k)}}
\newcommand{\mlsum}{\sum_{\ell = 1}^L}

\newcommand{\Cl}{C_\ell}
\newcommand{\Nl}{N_\ell}
\newcommand{\Nlmone}{N_{\ell - 1}}
\newcommand{\Nlsum}{\sum_{i=1}^{\Nl}}
\newcommand{\Nlsumj}{\sum_{j=1}^{\Nl}}
\newcommand{\Nlsumk}{\sum_{k=1}^{\Nl}}
\newcommand{\Nlsumh}{\sum_{h=1}^{\Nl}}
\newcommand{\Nlsumjneqi}{\sum_{j=1, j \neq i }^{\Nl}}

\newcommand{\Zl}{Z_{\ell}}
\newcommand{\Zlmone}{Z_{\ell - 1}}
\newcommand{\Zli}{Z_{\ell}^{(i)}}
\newcommand{\Zlj}{Z_{\ell}^{(j)}}
\newcommand{\Zlk}{Z_{\ell}^{(k)}}
\newcommand{\Zlmonei}{Z_{\ell - 1}^{(i)}}
\newcommand{\Zlmonej}{Z_{\ell - 1}^{(j)}}
\newcommand{\Zlmonek}{Z_{\ell - 1}^{(k)}}

\newcommand{\monel}{\mu_{1, \ell}}

\newcommand{\monelmone}{\mu_{1, \ell - 1}}
\newcommand{\moneapproxml}{\approxsym{\mu}_{1, \textml}}
\newcommand{\moneapproxl}{\approxsym{\mu}_{1, \ell}}
\newcommand{\moneapproxli}{\approxsym{\mu}_{1, \ell_i}}
\newcommand{\moneapproxbootl}{\approxsym{\mu}_{1, \ell}^{(b)}}
\newcommand{\moneapproxlmone}{\approxsym{\mu}_{1, \ell - 1}}

\newcommand{\monelone}{\mu_{1, \ell_1}}
\newcommand{\moneltwo}{\mu_{1, \ell_2}}
\newcommand{\monelthree}{\mu_{1, \ell_3}}
\newcommand{\moneapproxlone}{\approxsym{\mu}_{1, \ell_1}}
\newcommand{\moneapproxltwo}{\approxsym{\mu}_{1, \ell_2}}
\newcommand{\moneapproxlthree}{\approxsym{\mu}_{1, \ell_3}}

\newcommand{\cmtwol}{\mu_{2, \ell}}
\newcommand{\cmtwolmone}{\mu_{2, \ell - 1}}
\newcommand{\cmtwoapproxl}{\approxsym{\mu}_{2, \ell}}
\newcommand{\cmtwoapproxlj}{\approxsym{\mu}_{2, \ell_j}}
\newcommand{\cmtwoapproxlmone}{\approxsym{\mu}_{2, \ell - 1}}
\newcommand{\cmtwoapproxml}{\approxsym{\mu}_{2, \textml}}

\newcommand{\cmthreel}{\mu_{3, \ell}}

\newcommand{\sigmaapproxml}{\approxsym{\sigma}_{\textml, \biasedtxt}}
\newcommand{\sigmaapproxlj}{\approxsym{\sigma}_{\ell_j}}
\newcommand{\sigmaapproxbootl}{\approxsym{\sigma}_{\ell, \biasedtxt}^{(b)}}
\newcommand{\sigmaapproxl}{\approxsym{\sigma}_{\ell, \biasedtxt}}
\newcommand{\sigmaapproxlmone}{\approxsym{\sigma}_{\biasedtxt, \ell - 1}}

\newcommand{\Rmu}[1]{\mathcal{R}^{#1}_\mu}
\newcommand{\Rsig}[1]{\mathcal{R}^{#1}_\sigma}
\newcommand{\Rscal}[1]{\mathcal{R}^{#1}_{\mu+\alpha\sigma}}

\title{Multilevel Monte Carlo estimators for derivative-free optimization under Uncertainty}

\author{ %\href{https://orcid.org/0000-0000-0000-0000}{\includegraphics[scale=0.06]{orcid.pdf}\hspace{1mm} }
	Friedrich Menhorn \\
	Department of Computer Science\\
	Technical University of Munich\\
	85748 Garching, Germany \\
	\texttt{menhorn@in.tum.de} \\
	\And
	Gianluca Geraci \\
	Optimization and Uncertainty Quantification\\
	Sandia National Laboratories\\
	Albuquerque, New Mexico, 87185, USA \\
	\texttt{ggeraci@sandia.gov} \\
	\And
	D. Thomas Seidl \\
	Optimization and Uncertainty Quantification\\
	Sandia National Laboratories\\
	Albuquerque, New Mexico, 87185, USA \\
	\texttt{dtseidl@sandia.gov} \\
	\And
	Youssef M. Marzouk \\
	Department of Aeronautics and Astronautics\\
	Massachusetts Institute of Technology\\
	Cambridge, MA 02139, USA \\
	\texttt{ymarz@mit.edu} \\
	\And
	Michael S. Eldred \\
	Optimization and Uncertainty Quantification\\
	Sandia National Laboratories\\
	Albuquerque, New Mexico, 87185, USA \\
	\texttt{mseldre@sandia.gov} \\
	\And
	Hans-Joachim Bungartz \\
	Department of Informatics\\
	Technical University of Munich\\
	85748 Garching, Germany \\
	\texttt{bungartz@in.tum.de} \\
}

\begin{document}

\maketitle

\begin{abstract}
Optimization is a key tool for scientific and engineering applications, however, in the presence of models affected by uncertainty, the optimization formulation needs to be extended to consider statistics of the quantity of interest. Optimization under uncertainty (OUU) deals with this endeavor and requires uncertainty quantification analyses at several design locations. The cost of OUU is proportional to the cost of performing a forward uncertainty analysis at each design location visited, which makes the computational burden too high for high-fidelity simulations with significant computational cost. From a high-level standpoint, an OUU workflow typically has two main components: an inner loop strategy for the computation of statistics of the quantity of interest, and an outer loop optimization strategy tasked with finding the optimal design, given a merit function based on the inner loop statistics. 
In this work, we propose to alleviate the cost of the inner loop uncertainty analysis by leveraging the so-called Multilevel Monte Carlo (MLMC) method. MLMC has the potential of drastically reducing the computational cost by allocating resources over multiple models with varying accuracy and cost. The resource allocation problem in MLMC is formulated by minimizing the computational cost given a target variance for the estimator. We consider MLMC estimators for statistics usually employed in OUU workflows and solve the corresponding allocation problem. For the outer loop, we consider a derivative-free optimization strategy implemented in the SNOWPAC library; our novel strategy is implemented and released in the Dakota software toolkit. We discuss several numerical test cases to showcase the features and performance of our novel approach with respect to the single fidelity counterpart, based on standard Monte Carlo evaluation of statistics.
\end{abstract}

\keywords{Optimization Under Uncertainty, Multilevel Monte Carlo, Uncertainty Quantification}

\maketitle

% \tableofcontents

%%%%%%%%%%%%%%%%%%%%%%%%%%%%%%%%%
\section{Introduction}
\label{sec:introduction}
Complex systems are required to maintain their design performance under various operational conditions, which may not be fully controlled. Optimization under uncertainty (OUU) deals with the task of designing systems that are robust with respect to the variability in operative conditions, e.g., a wind turbine operating in the atmosphere where its properties cannot be fully controlled, but are only characterized in a stochastic sense. In literature, such problems are also often described as stochastic optimization problems or stochastic programming, which we will use as synonymous~\cite{Ben-Tal1999,Beyer2007}. 

One of the main challenges associated with OUU is the high computational cost. The need for evaluating statistics corresponds to requiring a forward Uncertainty Quantification (UQ) step for each design condition. From this standpoint, leaving aside the underlying details of the optimization algorithm, OUU can be seen as the evaluation of system's statistics at several design conditions which may themselves involve $\mathcal{O}(100)$ or more evaluations. For this reason, lowering the computational cost of each of the required forward UQ steps has the potential to drastically reduce the cost of OUU, especially in the presence of high-fidelity and computationally expensive computer codes. 

In the last two decades, the problem of reducing the computational burden of UQ studies for computationally expensive codes
has gained the attention of the UQ community. The seminal work introduced in~\cite{Heinrich2001,Giles2008} illustrated that realizations from computational codes with varying accuracy and cost can be optimally fused to obtain statistics at a much reduced cost, when compared with single fidelity estimators with the same accuracy. The so-called multilevel Monte Carlo (MLMC)~\cite{Heinrich2001,Giles2008} paved the way for subsequent research and has been demonstrated to be able to work efficiently, whenever it is possible to obtain a family of computational approximations based on varying spatial and temporal resolutions. More recently, the need to apply similar approaches to computational systems in which the use of different physics models can lead to additional computational savings has generated interest in the so-called multifidelity UQ approaches, both in sampling-based~\cite{Pasupathy2012,ng2014,Nobile2015,Haji_NT_NM_2016,Peherstorfer2016b,geraci_multifidelity_2017,Fairbanks2017,Peherstorfer18,Gorodetsky_GEJ_JCP_2020,Schaden_U_SISC_2020} or surrogate-based approaches~\cite{LeGratiet_G_IJUQ_2014,Kennedy_O_B_2000,Jakeman_EGG_2019,Rumpfkeil_B_AIAAJ_2020,MFNets_IJUQ_2020,Gorodetsky_2021}.

Despite the differences in algorithms, all multilevel and multifidelity UQ approaches provide a way to optimally allocate computational resources, i.e., computational realizations, by taking into account their cost and accuracy. In order to allocate computational resources, all these methods require solving an optimization problem in which the overall computational cost is minimized, while targeting the variance of the chosen estimator. This optimization is affected by the desired statistics; most of the literature has focused on estimating the expected value~\cite{Giles2008}, or central moments~\cite{Bierig2016,Quian2018,Krumscheid2020}, however, very little has been investigated regarding the statistics that are usually important for OUU workflows.

In this work, we focus our attention on MLMC approaches, which, despite being more limited in their application compared to multifidelity UQ approaches, can provide an easier starting point for investigating the coupling between multilevel/multifidelity UQ and OUU. Moreover, we explicitly target statistics which are useful in the OUU context. The first work regarding MLMC estimators, for higher order moments known to us, was published in~\cite{Bierig2016} and introduced a multilevel variance estimator. More recently, \cite{Krumscheid2020} leveraged h-statistics and symbolic computations to find unbiased closed form solutions for the higher-order moments; the authors approximated the underlying optimization problem for the sample allocation, thereby solving an approximate analytic problem. We are unaware of any previous efforts to develop MLMC estimators for the standard deviation and its linear combination with the mean. We have developed these estimators, which are presented in this paper. OUU formulation are routinely used either in a robust sense, i.e., maximizing the performance of a system while minimizing its sensitivity to perturbations, see~\cite{Sandgren2002, Zang2005, Yao2011}, or in a reliability sense, i.e., ensuring that the system's performance are met with a certain probability; see~\cite{Bichon2007, Paiva2014, DakotaTheory}. 
In this contribution we provide MLMC formulation that are optimal for statistics usable in both the robustness and reliability context; however, our numerical OUU experiments only focus on single objective optimization problems for reliability design, which is the motivation of this work. As it will be clear later, arguably the most important reliability measure depend on standard deviation, which required us to also focus on the variance. As a consequence, the provided formulation for this moment could be embedded in a robustness OUU formulation, but this is not done in the present manuscript.  

The main contributions of this paper, which support the use of MLMC for OUU reliability formulations, are:
\begin{enumerate}
\item We derived an allocation strategy for an MLMC estimator for the variance where, unlike the work in~\cite{Krumscheid2020}, we do not use h-statistics (which lead to approximate analytical solutions), but rather rely on numerical optimization. Moreover, we provide numerical comparisons between the two approaches;
\item We derived an allocation strategy for a new MLMC estimator for the standard deviation;
\item We derived an allocation strategy for a new MLMC estimator for a linear combination of the mean and standard deviation, which is a common measure of reliability in OUU. 
\end{enumerate}

The remainder of the paper is organized as follows. In Section~\ref{sec:mathback}, we introduce the mathematical and algorithmic background, while, in Section~\ref{ssec:measuresofrobustness}, we present the measures of robustness and risk that we consider. Section~\ref{ssec:snowpac} briefly introduces the derivative-free stochastic constrained optimization method, available in the library SNOWPAC, which we will use as our solver for the OUU. Afterwards, we introduce the sampling estimators for statistics used for either robustness or reliability design, in Section~\ref{sec:samplingforouu}. The single fidelity case is presented first, in Section~\ref{ssec:single_fidelity}. Then, the multilevel case is introduced in Section~\ref{ssec:mlmcformean}, where we update the notation and summarize the classical results for the expected values, from~\cite{Heinrich2001,Giles2008}. Our first contribution, the MLMC estimator for the variance, is presented in Section~\ref{ssec:mlmcforvariance}. Section~\ref{ssec:mlmcforstddev} and Section~\ref{ssec:mlmcforscalarization} encompass the main contribution of this work in which the new MLMC estimators, for the standard deviation and linear combination of mean and standard deviation, are introduced. In Section~\ref{ssec:implementationdetails}, we describe the implementation details and our algorithm for the adaptive allocation of the samples over the different levels, for the different estimators. We show the benefit of these contributions by applying the new estimators on a simple 1-D toy problem, as well as a more challenging problem, namely a modified Rosenbrock function, in Section~\ref{sec:numericalresults}. For these numerical results, both the UQ only case and OUU are considered. Finally, conclusions are presented in Section~\ref{sec:conclusion}.

%%%%%%%%%%%%%%%%%%%%%%%%%%%%%%%%%
\section{Mathematical and algorithmic background}
\label{sec:mathback}
In this work, we are concerned with optimization problems of the following formulation 
\begin{equation}\label{eq:general_stoch_opt_problem}
\begin{split}
&\;\;\;\accentset{\ast}{f} = \min f(x, \theta)\\
&\mbox{s.t.} \quad c_i(x, \theta) \leq 0, i = 1,..., M,
\end{split}
\end{equation}
where $f(x, \theta): \mathbb{R}^d \times \Theta \rightarrow \mathbb{R}$ is the objective function subject to $M$ constraints $c_i(x, \theta): \mathbb{R}^d \times \Theta \rightarrow \mathbb{R}, i = 1, ..., M$. The vector $x \in \mathbb{R}^d$ is our design variable, while $\theta \in \Theta$ is the vector of random variables with the complete probability space $(\Theta, \mathcal{F}, P)$; as usual, $\Theta$ is the set of all possible outcomes, the Borel $\sigma$-algebra $\mathcal{F}$ is the event space and $P$ is the probability function. 
The functions $f$ and $c_i$ are derived from models affected by uncertainty, therefore they are both random, and we will consider nonlinear and black box models; see~\cite{Bertsimas2011, Pflug1996}. Since the probabilistic nature of the problem also introduces challenges for gradient estimation, we will rely on a derivative-free approach to circumvent this issue (see Section~\ref{ssec:snowpac}).

In order to solve and find an optimum solution we consider a class of problems denoted as \say{Here and Now} \cite{Diwekar2003}. Here and now problems require that the objective function and constraints be expressed in terms of some probabilistic representation (e.g., expected value, variance or quantiles). Furthermore, the decision variables and uncertain parameters are separated from each other. We achieve this separation by first integrating over the stochastic space, at the current design, and then using a stochastic model for the optimization. 

Formally, we write the problem as 
\begin{equation}\label{eq:general_robust_opt_problem}
\begin{split}
&\;\;\;\accentset{\ast}{\mathcal{R}}^{f}(x)\ = \min \mathcal{R}^f(x)\\
&\mbox{s.t.} \quad \mathcal{R}^{c_i}(x) \leq 0, i = 1,..., M.
\end{split}
\end{equation}
where $\mathcal{R}^b, b \in \{f, c_1, ..., c_M\}$ are arbitrary statistics of the QoI, e.g. the aforementioned expected value or standard deviation. We will discuss the measures $\mathcal{R}^b$ that we use in this work in more detail in Section~\ref{ssec:measuresofrobustness}, but we also refer to the rich literature on robustness, reliability, risk and deviation measures \cite{Acerbi2002, Artzner1999, Krokhomal2011, Rackwitz2001, Rockafellar2000, Rockafellar2002a, Rockafellar2002, Uryasev2000, Zhang, Beyer2007}. The method SNOWPAC that we use to solve Eq.~\eqref{eq:general_robust_opt_problem} is afterwards presented in Section~\ref{ssec:snowpac}.

\subsection{Measures of robustness and reliability}
\label{ssec:measuresofrobustness}
The following sections introduce the sampling estimators for the robustness and reliability measures that we consider in this work.  They are, e.g., given in \cite{Beyer2007} and we refer to \cite{Rockafellar2002, Szego2002} for a detailed discussion of risk assessment strategies and an introduction to a wider class of measures.

The classical first measure is the expected value
\begin{equation}
\Rmu{b}(x) := \mathbb{E}\left[ b(x; \theta) \right] = \int\limits_{\Theta} b(x; \theta) dP.
\end{equation}
It is a widely applied measure to handle uncertain parameters in optimization problems although despite measuring robustness with respect to variations in $\theta$, it does not inform about the spread of $b$. 

In order to also account for the spread of realizations of $b$ around its mean $\Rmu{b}$, we consider the standard deviation
\begin{equation}
\mathcal{R}^b_\sigma(x) := \mathbb{V}^{\frac{1}{2}}\left[ b(x; \theta) \right] = \left( \int\limits_{\Theta} \left( b(x; \theta) - \mathbb{E}\left[ b(x; \theta) \right] \right)^2 dP \right)^{1/2}.
% \sqrt{\mathbb{E}\left[ b(x; \theta)^2 \right] - \mathcal{R}^b_0(x)^2},
\end{equation}
Finally, the linear combination of $\Rmu{b}(x)$ and $\Rsig{b}(x)$ given by
\begin{equation}
\Rscal{b}(x) := \Rmu{b}(x) + \alpha \Rsig{b}(x), \alpha \in \mathbb{R},
\end{equation}
is a common measure for reliability since it provides a trade-off between two possibly contradicting goals: the minimization of the expected outcome and the minimization of the variability. Here, $\alpha$ is called the reliability index \cite{Haldar2000}. Hereinafter, we refer to this latter measure as \textit{scalarization} since we collapse two measures in a single scalar quantity. We note that, when used as a constraint, $\Rscal{b}(x)$ can also be interpreted as a chance constraint if we assume the quantity of interest to be normally distributed. For example, by using $\alpha = 3$ and under the aforementioned distribution assumption, $Q$ lies in the range of $\Rmu{b}(x) \pm \alpha \Rsig{b}(x)$ with probability $p \approx 0.998$.

In practice, all the measures presented above cannot be computed exactly; rather, an approximation needs to be introduced. In this work, we rely on MC sampling~\cite{Augustin2} and its multilevel extension to estimate these measures at a reduced computational cost. In both cases, the finite number of samples used to approximate these measures introduces an error. For a generic measure, $\mathcal{R}^b(x)$, we introduce a sampling estimator that produces an approximation, $R^b(x)$, which leads to the error 
$ \varepsilon_R^b = \mathcal{R}^b(x) - R^b(x)$. It is well known that for sampling methods, this error decreases with $\mathcal{O}(N^{\frac{1}{2}})$, where $N$ is the number of realizations; however, controlling this error with an acceptable computational burden is difficult in practical applications.

\subsection{Derivative-free stochastic optimization method: the SNOWPAC algorithm}
\label{ssec:snowpac}
In the previous section, we defined the measures for solving Eq.~\eqref{eq:general_robust_opt_problem}. There is an abundance of literature on approaches on how to solve such kind of problems, e.g., surrogate-based approaches \cite{Bortz1998, Conn1993, Kelley1999, March2012, Regis2011, Regis2014, Sampaio2015, Carter1991, Choi2000, Heinkenschloss2002, Larson2016, Chen2018} or the re-popularized stochastic approximation method \cite{Robbins1951, Kiefer1952, Bottou2018} due to the rise of machine learning. In the following, without specifying the sampling estimator adopted, we just assume to have a suitable sampling approximation $R^b(x)$ for the measure $\mathcal{R}^b(x)$, which results in an error ${\varepsilon}_R^b(x)$. 

In our work we expect to have as little knowledge about the problem \eqref{eq:general_robust_opt_problem} as possible. We assume the underlying model of $f$ and $\{c_i\}_{i=1}^M$ to be black box which, e.g., means we might not have access to derivatives. To avoid the need for dealing with gradients, we rely on a derivative-free optimization method; in particular, we use the SNOWPAC algorithm introduced in~\cite{Augustin2}. SNOWPAC is a stochastic derivative-free optimization method which uses a trust region approach to tackle problems of the form of Eq.~\eqref{eq:general_robust_opt_problem}. It extends its deterministic predecessor NOWPAC~\cite{Augustin1} that uses
fully-linear surrogate models and quadratic optimization on the surrogate in the trust region. 

For the stochastic application SNOWPAC employs MC sampling to evaluate the measures $R^b \approx \mathcal{R}^b$, which we will extend to MLMC in the following sections. Regardless of the sampling approach used, the error $\varepsilon_R^b$ deteriorates the approximation quality of the surrogates of objective function and/or  constraint in the trust region. 
This restricts the possible trust region size and the minimal possible trust region radius $\rho_k$ at a given optimization step $k$ is restricted by the maximal noise $\varepsilon^k_{\text{max}} = \max_{b \in \{f, c\}} {\varepsilon^k_{R^b}}$ as,
\begin{equation}\label{eq:lower_bound_on_trust_region}
{\varepsilon}_{max}^k \rho_k^{-2} \leq \lambda_t^{-2}, \quad \mbox{resp.} \qquad \rho_k \geq \lambda_t \sqrt{{\varepsilon}_{max}^k} = \max_i \lambda_t \sqrt{\varepsilon_{i}^k},
\end{equation}
where $\lambda_t$ is a safety parameter. The size of the trust region is directly linked to the convergence of the algorithm; thus, progress of the algorithm is only achieved if the noise can be reduced. 

To counteract the lower bound on the trust region introduced by ${\varepsilon}^b_{max}$, Gaussian Process (GP) surrogates \cite{Rasmussen2006} are used to bias the evaluations and to reduce the noise. This is comparable to a biased control variate approach. 
For this, smoothed evaluations $\tilde{R}^{b}$ and noise estimates $\tilde{\varepsilon}_R^b$ are constructed
\begin{align}
\label{eq:adjusted_black_box_evaluations}
\tilde{R}^{b} &= \;\gamma \mathcal{G}^b({x}) +(1-\gamma) {R}^b(x),\\
\tilde{\varepsilon}_R^b\; &=\; \gamma 2 \sigma^b({x})+(1-\gamma){\varepsilon}_R^b(x),
\nonumber
\end{align}
where $\mathcal{G}^b({x})$ is the GP mean estimate and $\sigma^b({x})$ is the standard deviation of the GP at ${x}$. Note that $\gamma := e^{-\sigma^b\left({x}\right)}$ is chosen to approach 1 following the approximation quality of the Gaussian process. The corrected evaluations $\tilde{R}^{b}$ at the local interpolation points are then used to build local surrogates and the associated reduced noise level $\tilde{\varepsilon}_R^b$ allows a reduction in the trust region radius $\rho$.

While the Gaussian process surrogate is built over a larger domain and, therefore, holds more global information, the minimum Frobenius norm surrogate models are built locally in the trust region. Through the combination following equation \eqref{eq:adjusted_black_box_evaluations}, we balance the error in the surrogate model via the lower bound on the trust region with the error in the Gaussian process model represented by its standard deviation estimate; increasing the number of evaluations achieves a higher confidence in the Gaussian process model, a decrease in the noise and, consequently, a decrease of the trust region.

By combining those two surrogate models SNOWPAC balances two sources of approximation errors. On the one hand, there is the structural error in the approximation of the local surrogate models, which is controlled by the size of the trust region radius. On the other hand, we have the inaccuracy in the GP surrogate itself, which is reflected by the standard deviation of the GP. Noteworthy is that SNOWPAC relates these two sources of errors by coupling the size of the trust region radii to the size of the credible interval, only allowing the trust region radius to decrease if the GP posterior estimator becomes small.

Biasing the evaluations following \eqref{eq:adjusted_black_box_evaluations} may, however, result in infeasible steps. Thus, thirdly, SNOWPAC also offers a feasible restoration mode which switches the OUU formulation. While a point is infeasible, the constraints are first minimized to return to the feasible region before we proceed with the actual optimization.

The method is available in the optimization and uncertainty quantification framework DAKOTA (from version 6.7 \cite{Dakota}) where it can be used a stand-alone solver or an approximate subproblem solver. With its derivative-free approach it offers the flexibility and applicability to a wide range of problems. For a more elaborate introduction to the method we refer the interested reader to~\cite{Augustin2}.

%%%%%%%%%%%%%%%%%%%%%%%%%%%%%%%%
\section{Sampling estimators for robustness and reliability measures}
\label{sec:samplingforouu}

In this section we present sampling estimators that we use to approximate the measures
described in Section~\ref{ssec:measuresofrobustness}, which ultimately leads to the error $\varepsilon^b_R$ used in the 
construction of the trust-region surrogate in SNOWPAC (Section~\ref{ssec:snowpac}). We start by presenting the single fidelity MC estimator
in Section~\ref{ssec:single_fidelity}, and the multilevel MC (MLMC) for the mean in Section~\ref{ssec:mlmcformean} to introduce the notation and illustrate the \textit{resource allocation} problem, which is the problem of allocating the realizations over a set of models with varying accuracy and cost. In the remaining sections, we discuss the extension of MLMC to variance, standard deviation and scalarization along with the introduction of strategies for solving their resource allocation problems.

\subsection{Single fidelity MC estimators}
\label{ssec:single_fidelity}
Let start by introducing the notation used throughout the paper. Given a QoI $Q(x,\theta) : \mathbb{R} \times \Theta \rightarrow \mathbb{R}$ where $\theta \in \Theta$ is the vector of random variables, we use the shorthand $Q := Q(x, \theta)$. We call $Q$ our quantity of interest. In the optimization scenario of \eqref{eq:general_robust_opt_problem}, $Q$ can be the objective $f$ or a constraint $c_i$. A realization (or sample) is then written as $Q^{(i)} := Q(x, \theta_{i})$ 
where $N$ samples are used such that: $\{ Q^{(1)} , \hdots, Q^{(N)} \} = \{ Q (x, \theta_{1}), \hdots, Q (x, \theta_{N}) \}$. We employ $\mone[Q] := \mathbb{E}[Q]$ for the expected value while $\mu_i := \mathbb{E}[ (Q - \mone[Q])^i ], i > 1,$ is used for the $i$-th central moment. If obvious from context, we omit the integrand, e.g. $\mu_i := \mu_i[Q]$. Finally, the hat symbol stands for a sampling approximation of the quantity, e.g., $\moneapprox \approx \mone$.

Using this notation, we define the MC estimator for the expected value (or mean) as
\begin{equation}
\label{eq:unbiasedmone}
\moneapprox = \frac{1}{N} \sum_{i=1}^N Q^{(i)}.
\end{equation}
This estimator is unbiased and its variance can be obtained as
\begin{equation}
\label{eq:var_mean}
 \Var{ \moneapprox } = \frac{\Var{Q}}{N}.
\end{equation}
The knowledge of the MC estimator variance, Eq.~\eqref{eq:var_mean}, allows for a straightforward allocation of resources, whenever a target variance for $\Var{ \moneapprox }$ is desired, \textit{i.e.}, $N = \Var{Q} / \Var{ \moneapprox }$.

MC can be also used for the unbiased estimator for the variance, i.e. the second centered moment, as
\begin{equation}
\label{eq:unbiasedcmtwo}
\cmtwoapprox = \frac{1}{N-1} \sum_{i=1}^{N} (Q^{(i)} - \moneapprox)^2.
\end{equation}
This estimator, thanks to the use of the Bessel correction, is also unbiased and its variance has a closed-form expression~\cite{Mood1974}
\begin{equation}
\label{eq:var_var}
\begin{split}
\mathbb{V}[\cmtwoapprox]  &= \frac{1}{N} ( \cmfour - \frac{N-3}{N-1}  \cmtwo^2). 
\end{split}
\end{equation}
which, however, depends on both the exact statistics of the second and fourth, $\cmtwo$ and $\cmfour$ respectively, central moments of the QoI. We derive an unbiased estimator for Eq.~\eqref{eq:var_var} when we use estimators for $\cmfour$ and $\cmtwo$ in the following Lemma. This result is also given in~\cite{Krumscheid2020} derived using h-statistics.

\begin{lemma}
Let $\cmtwoapprox$ and $\cmfourapprox$ be unbiased estimators for the second and fourth central moment. The unbiased estimator of the variance of the second central moment is then given as
\begin{equation}
\label{eq:unbiasedvarianceofvariance} 
\mathbb{V}[\cmtwoapprox] \approx \frac{(N-1)}{N^2 - 2N +3} \left( \cmfourapprox -   \frac{N-3}{N-1}  \cmtwoapprox^2 \right).
\end{equation}
\begin{proof}
 See \ref{prf:unbiasedestimatorforvarianceofvariance} for the proof. 
\end{proof}

\end{lemma}
\begin{Appendix}
\section{Proof: Unbiased estimator for variance of variance}
\label{prf:unbiasedestimatorforvarianceofvariance}
\begin{proof}
\begin{equation}
\begin{split}
\mathbb{E}\left[\frac{(N-1)}{N^2 - 2N +3} \left( \cmfourapprox -   \frac{N-3}{N-1}  \cmtwoapprox^2 \right)\right] &=  \frac{(N-1)}{N^2 - 2N +3}  \left( \mathbb{E}[\cmfourapprox] - \frac{(N-3)}{(N-1)} \mathbb{E}[\cmtwoapprox^2] \right) \\
&=  \frac{(N-1)}{N^2 - 2N +3} \left( \cmfour- \frac{N-3}{N-1} \mathbb{E}[\cmtwoapprox^2]\right) \\
&=  \frac{(N-1)}{N^2 - 2N +3} \left( \cmfour- \frac{N-3}{N-1} \bigg[\frac{1}{N} \left(\cmfour - \frac{N-3}{N-1} \cmtwo^2\right) + \cmtwo^2\bigg]\right) \\
&=  \frac{(N-1)}{N^2 - 2N +3} \left( \cmfour- \frac{N-3}{N(N-1)} \cmfour + \frac{(N-3)^2}{N(N-1)^2} \cmtwo^2 - \frac{N-3}{N-1} \cmtwo^2 \right) \\
&=  \frac{(N-1)}{N^2 - 2N +3} \left[ \left(1 - \frac{N-3}{N(N-1)} \right) \cmfour - \left(1 -  \frac{(N-3)}{N(N-1)}\right) \frac{(N-3)}{(N-1)} \cmtwo^2 \right] \\
&=  \frac{(N-1)}{N^2 - 2N +3} \left(1 - \frac{N-3}{N(N-1)} \right) \frac{1}{N} \left( \cmfour - \frac{(N-3)}{(N-1)} \cmtwo^2 \right) \\
&=  \frac{(N-1)}{N^2 - 2N +3}  \left(\frac{N^2 - 2N +3}{N(N-1)} \right) \left( \cmfour - \frac{(N-3)}{(N-1)} \cmtwo^2 \right) \\
&=  \frac{1}{N}  \left( \cmfour - \frac{(N-3)}{(N-1)} \cmtwo^2 \right)
\end{split}
\end{equation}
\end{proof}
\end{Appendix}

Indeed, in Eq.~\eqref{eq:var_var}, we can estimate the variance $\mathbb{V}[\cmtwoapprox]$ by relying on sample estimators for both the fourth $\cmfour$ and second $\cmtwo$ central moments. Since an unbiased estimator for the variance is already available (see Eq.~\eqref{eq:unbiasedcmtwo}), we only need to obtain an unbiased estimator for the fourth central moment $\cmfour$. Obtaining this unbiased estimator, from its biased counterpart, is discussed in the following proposition.

\begin{lemma}
Let $\cmfourapproxbiased = \frac{1}{N} \sum_{i=1}^N (Q^{(i)} - \moneapprox)^4$ be a biased estimator for the fourth central moment and let $\cmtwoapprox$ be an unbiased estimator for the second central moment as given in~\eqref{eq:unbiasedcmtwo}. Then, an unbiased estimator for the fourth central moment is given as
\begin{equation}
\label{eq:unbiasedcmfour}
\cmfourapprox = \frac{1}{(N^2-3N+3) - \frac{(6N-9)(N^2-N)}{N(N^2-2N+3)}} \left(\frac{N^3}{N-1} \cmfourapproxbiased - \frac{(6N-9)(N^2-N)}{N^2-2N+3}  \cmtwoapprox^2 \right).
\end{equation}
\end{lemma}
\begin{proof}
See \ref{prf:unbiasedestimatorforfourthcentralmoment} for the proof.
\end{proof}

\begin{Appendix}
\section{Proof: Unbiased estimator for fourth central moment}
\label{prf:unbiasedestimatorforfourthcentralmoment}
\begin{proof}
A biased estimator for the fourth central moment is given in \cite[p.268, after eq. (6)]{Dodge1999} as
\begin{equation}
\label{eq:unbiasedcmfour1}
\begin{split}
&\mathbb{E}[\cmfourapproxbiased] = \frac{(N-1)(N^2-3N+3)}{N^3} \cmfour + \frac{3(2N-3)(N-1)}{N^3} \cmtwo^2 \\
\Leftrightarrow &\cmfour = \frac{1}{N^2-3N+3} \left(\frac{N^3}{N-1} \mathbb{E}[\cmfourapproxbiased] - (6N-9) \cmtwo^2 \right).
\end{split}
\end{equation}
Note that this is an unbiased estimator only if we use the exact value for $\cmtwo$ since $\cmtwo^2$ is unbiased while $\cmtwoapprox^2$ is not. Therefore, we need an unbiased estimator for $\cmtwoapprox^2$ and we know that
\begin{equation}
\label{eq:unbiasedcmfour2}
\begin{split}
\mathbb{E}[\cmtwoapprox^2] &= \mathbb{V}[\cmtwoapprox] + \mathbb{E}[\cmtwoapprox]^2 \\
													  &= \frac{1}{N} \left(\cmfour - \frac{N-3}{N-1} \cmtwo^2\right) + \cmtwo^2.
\end{split}
\end{equation}
Using both~\eqref{eq:unbiasedcmfour1} and~\eqref{eq:unbiasedcmfour2} we get the result
\begin{equation}
\label{eq:unbiasedcmfour3}
\begin{split}
&\mathbb{E}\left[ \underbrace{\frac{1}{(N^2-3N+3) - \frac{(6N-9)(N^2-N)}{N(N^2-2N+3)}}}_{(*)} \left(\frac{N^3}{N-1} \cmfourapproxbiased - \frac{(6N-9)(N^2-N)}{N^2-2N+3}  \cmtwoapprox^2 \right)\right] \\
&= (*) \left(\frac{N^3}{N-1} \mathbb{E}[\cmfourapproxbiased] - \frac{(6N-9)(N^2-N)}{N^2-2N+3}  \mathbb{E}[\cmtwoapprox^2] \right) \\
&= (*) \bigg(\frac{N^3}{N-1} \mathbb{E}[\cmfourapproxbiased] - \frac{(6N-9)(N^2-N)}{N^2-2N+3}  \bigg[ \frac{1}{N} (\cmfour - \frac{N-3}{N-1} \cmtwo^2) + \cmtwo^2 \bigg] \bigg) \\
&= (*) \bigg(\frac{N^3}{N-1} \bigg[\frac{(N-1)(N^2-3N+3)}{N^3} \cmfour + \frac{3(2N-3)(N-1)}{N^3} \cmtwo^2 \bigg] -\\
& \quad \quad \quad \frac{(6N-9)(N^2-N)}{N^2-2N+3}  \bigg[ \frac{1}{N} (\cmfour - \frac{N-3}{N-1} \cmtwo^2) + \cmtwo^2 \bigg] \bigg) \\
&= \cmfour.
\end{split}
\end{equation}
\end{proof}
\end{Appendix}

Finally, the last single fidelity estimator we need to discuss is the standard deviation, which can be approximated, directly from the variance estimator, as
\begin{equation}
\label{eq:unbiasedstddev}
\sigmaapprox = \sqrt{\cmtwoapprox}.
\end{equation}

This latter case introduces a number of challenges. First, an unbiased version of the estimator cannot be easily obtained (even if we rely on the unbiased variance Eq.~\eqref{eq:unbiasedcmtwo}). Second, not surprisingly, the variance of this estimator cannot be obtained in closed-form. To overcome this difficulty and obtain an expression to use for resource allocation purposes, we can rely on the delta method~\cite{DeltaMethod2012}. It employs a Taylor series expansion to find the approximate probability distribution for a function of an asymptotically normal estimator, which, in our case, will be the square root function and the variance estimator respectively: 

\begin{lemma}
Let assume that $\sigmaapprox^2$ is asymptotically normal distributed, 
and that mean and variance exist. The variance of $\sigmaapprox$ can be approximated by using the delta method~\cite{DeltaMethod2012} as
\begin{equation}
\label{eq:mcvarianceofstddev}
\begin{split}
\mathbb{V}[\sigmaapprox] &\approx   
\frac{1}{4 \cmtwoapprox}   \mathbb{V}[ \cmtwoapprox ].
\end{split}
\end{equation}
\end{lemma}
\begin{proof}
See \ref{prf:deltamethod} for the proof.
\end{proof}
\begin{Appendix}
\section{Proof: delta method}
\label{prf:deltamethod}
\begin{proof}
We can find an approximation by using a Taylor expansion of $g(X)$ around $\mone$, s.t. 
\begin{equation}
g(X) = g(\mone) + g'(\mone)(X-\mone)+g''(\mone)\frac{(X-\mone)^2}{2!} + ...,
\end{equation}
where we drop the higher terms to get
\begin{equation}
g(X) \approx g(\mone) + g'(\mone)(X-\mone).
\end{equation}
Taking the variance on both sides yields
\begin{equation}
\label{eq:deltageneral}
\mathbb{V}[g(X)] \approx \mathbb{V}[g(\mone)] + \mathbb{V}[g'(\mone)(X-\mone)] = g'(X)^2  \mathbb{V}[X].
\end{equation}
Denoting $X := \sigmaapprox^2$ and $g(X):= \sqrt{X}$, we get
\begin{equation}
\begin{split}
\mathbb{V}\left[\sqrt{\sigmaapprox^2}\right] &\approx   \left(\frac{1}{2 \sqrt{\sigmaapprox^2}}\right)^2    \mathbb{V}[ \sigmaapprox^2 ] =
\frac{1}{4 \cmtwoapprox}   \mathbb{V}[ \cmtwoapprox ].
\end{split}
\end{equation}
\end{proof}
\end{Appendix}

We again point out that the variance is not necessarily normal distributed, which results in an approximation. Of course, the absence of simple expression for $\mathbb{V}[\sigmaapprox]$ also mandates the use of a numerical optimization to determine the number of samples $N$, which provides a prescribed estimator variance. Since the use of the numerical optimization is also required in the MLMC context, we defer its presentation to the next section.

\subsection{Multilevel sample allocation for the mean estimator}
\label{ssec:mlmcformean}
In this section we discuss how the estimation of the quantities introduced in Section~\ref{ssec:single_fidelity} can be accelerated by employing sampling on multiple approximations or levels for a QoI. In particular, we illustrate how, by relying on different approximations, the target variance of the estimator can be reached more efficiently by only using a limited number of samples at the highest resolution level. We start by considering the mean estimation case already available in literature~\cite{Giles2008,Giles2015}, which allows us to introduce the notation and the basic concepts. Afterwards, we move to more complex estimator targets---variance, standard deviation, and scalarization---as the main contribution of this work.

We consider samples for $Q$ obtained on different levels $\ell$, and, in particular, each multilevel estimator at level $\ell$ will include evaluations of $Q$ on two consecutive levels $\ell$ and $\ell-1$. When multiple levels of $\ell \in [1, L]$ are available, we extend notation from the previous section to $Q_\ell:= Q_\ell(x, \theta_\ell)$. A realization (or sample) is then written as $Q^{(i)}_\ell:= Q_\ell (x, \theta_\ell^{(i)})$, where $\Nl$ samples are used for level $\ell$, as follows: $Q^{(1)}_\ell, \hdots, Q^{(\Nl)}_\ell = Q_\ell (x, \theta_\ell^{(1)}), \hdots, Q_\ell (x, \theta_\ell^{(\Nl)})$.

The MLMC estimator for a QoI $Q$ can be expressed by expanding the expected value over levels, with $\ell=1$ being the coarsest level and $L$ being the finest level (which we omit for brevity, i.e., $Q_L=Q$), as 
\begin{equation}
\label{eq:mlmc_mean}
\begin{split}
\EE{Q} 
       \approx \moneapproxml[Q] &:= \mlsum \moneapproxl =  \mlsum \moneapprox[\Ql - \Qlmone] \\
&= \mlsum \frac{1}{\Nl} \sum_{i=1}^{\Nl} (\Qli - \Qlmonei), \quad Q^{(i)}_{0} := 0,
\end{split}
\end{equation}
Because they use the same samples $\theta_\ell^{(i)}$ for the evaluation on different levels, we introduce a dependence between the two quantities $\Qli:= Q_\ell[x, \theta_\ell^{(i)}]$ and $\Qlmonei:= Q_{\ell-1}[x, \theta_\ell^{(i)}]$. When we consider single-level estimators in the multilevel case, such as $\moneapproxl$, this dependence is implicitly assumed for the rest of this work.

The variance of the MLMC estimator for the mean is given as
\begin{equation}
\label{eq:mlmcvarianceofmean}
\mathbb{V}[\moneapproxml] = \mlsum \mathbb{V}[\moneapproxl - \moneapproxlmone] = \mlsum \frac{1}{\Nl} \mathbb{V}[Q_\ell - Q_{\ell-1}].
\end{equation}

To achieve the desired accuracy, the computational load can be redistributed toward the coarser level for a sequence of levels where $\mathbb{V}[Q_\ell - Q_{\ell-1}] \rightarrow 0$ with $\ell \rightarrow L$. To accomplish this, we define an associated computational cost for each level, such that a single $Q_\ell$ evaluation has a computational cost of $C_\ell$, and $C_1 < C_2 < \cdots < C_L$. The estimator in Eq.~\eqref{eq:mlmc_mean} represents a special case in which the closed-form solution for its variance allows for a closed-form solution for the allocation strategy. We introduce here the following notation that we will adopt throughout the paper: if a target variance of $\epsilon^2$ is provided for a MLMC estimator of a certain statistic, e.g., the expected value, we write the corresponding optimization problem as
\begin{equation}
\label{eq:mlmcmean}
\begin{split}
\accentset{\ast}{N}_\ell^{\mathbb{E}} = \argmin_{\Nl^{\mathbb{E}}} C_T^{\mathbb{E}} := \mlsum \Cl \Nl^{\mathbb{E}}, \\
\text{s.t. } \mathbb{V}[\moneapproxml] = \epsilon_{\mathbb{E}}^2,
\end{split}
\end{equation}
where $C_T^{\mathbb{E}}$ describes the total computational cost targeting the expectation.
By denoting the target statistic as a superscript, we can distinguish between resource allocations optimized for a specific statistic, here $\Nl^{\mathbb{E}}$. For brevity, we define this approach as \emph{targeting the mean}. 

In the case of the mean estimator, we only need to estimate the quantity $\mathbb{V}[Q_\ell - Q_{\ell-1}]$, whereas the optimization solution has been provided by~\cite{Heinrich2001, Giles2008} using the Lagrangian multiplier technique
\begin{equation}
\lambda = \epsilon_{\mathbb{E}}^{-2} \mlsum \sqrt{\mathbb{V}[Q_\ell - Q_{\ell-1}] \Cl},
\end{equation}
which can be used to express the optimal resource allocation, for a generic level $\ell$, as
\begin{equation}
\accentset{\ast}{N}_\ell^{\mathbb{E}} = \left\lceil \lambda \sqrt{\frac{\mathbb{V}[Q_\ell - Q_{\ell-1}]}{\Cl}} \right\rceil.
\end{equation}

\subsection{Multilevel sample allocation for the variance estimator}
\label{ssec:mlmcforvariance}
Following what we presented for the mean, we now move to the resource allocation problem for a MLMC estimator targeting the variance of the QoI. 
We first define the MLMC estimator for the variance of the QoI as
\begin{equation}
\label{eq:mlmcvarianceestimator}
\begin{split}
\mathbb{V}[Q_L] \approx \cmtwoapproxml[Q_L] &:= \mlsum \cmtwoapprox[\Ql] - \cmtwoapprox[\Qlmone] \\
																						&= \mlsum (\cmtwoapproxl - \cmtwoapproxlmone) \\
																						&= \mlsum \frac{1}{\Nl - 1} \bigg( \sum_{i=1}^{\Nl} (\Qli - \moneapproxl)^2 - (\Qlmonei - \moneapproxlmone)^2 \bigg).
\end{split}
\end{equation}

Our optimization problem reads as
\begin{equation}
\label{eq:mlmcvariance}
\begin{split}
\accentset{\ast}{N}_\ell^{\mathbb{V}} = \argmin_{\Nl^{\mathbb{V}}} C_T^{\mathbb{V}} := \mlsum \Cl \Nl^{\mathbb{V}}, \\
\text{s.t. } \mathbb{V}[\cmtwoapproxml] = \epsilon_{\mathbb{V}}^2,
\end{split}
\end{equation}
where we now \emph{target the variance} and in which the variance of the estimator still has to be determined.

The variance of the estimator $\cmtwoapproxml[Q_L]$ can be written as 
\begin{equation}
\label{eq:mlmcvarianceofvariance}
\begin{split}
\mathbb{V}[\cmtwoapproxml] = \mathbb{V}\left[\mlsum (\cmtwoapproxl - \cmtwoapproxlmone)\right] &= \mlsum \mathbb{V}\left[\cmtwoapproxl - \cmtwoapproxlmone\right]\\
&= \mlsum \mathbb{V}[\cmtwoapproxl] + \mathbb{V}[\cmtwoapproxlmone] - 2 \Cov[\cmtwoapproxl, \cmtwoapproxlmone]
\end{split}
\end{equation}
where we use independence of the samples over the different levels. Note again the short-hand notation $\cmtwoapproxlmone = \cmtwoapproxlmone\left[Q_{\ell-1}[x, \theta_\ell]\right]$ where we evaluate the moment of interest on level $\ell - 1$ but use the same samples as on level $\ell$. This results in the dependence expressed by the covariance term where we also employ this short-hand notation: $\Cov[\cmtwoapproxl, \cmtwoapproxlmone] = \Cov\left[\cmtwoapproxl\left[Q_{\ell-1}[x, \theta_\ell]\right], \cmtwoapproxlmone\left[Q_{\ell-1}[x, \theta_\ell]\right]\right]$. In Section~\ref{ssec:single_fidelity}, we discussed the estimation of the single fidelity variance expressions for the terms $\mathbb{V}[\cmtwoapproxl]$ and $\mathbb{V}[\cmtwoapproxlmone]$ (see Eq.~\eqref{eq:var_var} and its unbiased estimator in Eq.~\eqref{eq:unbiasedvarianceofvariance}). However, in Eq.~\eqref{eq:mlmcvarianceofvariance}, an additional term, $\Cov[\cmtwoapproxl, \cmtwoapproxlmone]$, needs to be evaluated, for which we provide the expression in the following lemma. 

\begin{lemma}
\label{lem:covariance}
Let $\cmtwoapproxl$ and $\cmtwoapproxlmone$ be unbiased single level estimators for the respective level $\ell$ and $\ell - 1$ as described in Eq.~\eqref{eq:unbiasedcmtwo}. Then, the covariance term in Eq.~\eqref{eq:mlmcvarianceofvariance} is given as
\begin{equation}
\label{eq:covariance}
\begin{split}
\Cov[\cmtwoapproxl, \cmtwoapproxlmone] &= \frac{1}{N_\ell} \mathbb{E}[\cmtwoapproxl \cmtwoapproxlmone] \\
																			  &+ \frac{1}{\Nl (\Nl - 1)} \bigg(\mathbb{E}[\Ql \Qlmone]^2 - 2 \mathbb{E}[\Ql \Qlmone] \mathbb{E}[\Ql]  \mathbb{E}[\Qlmone] + (\mathbb{E}[\Ql]  \mathbb{E}[\Qlmone])^2 \bigg)
\end{split}
\end{equation}
where 
\begin{equation}
\label{eq:covcostterm}
\begin{split}
\mathbb{E}[\cmtwoapproxl \cmtwoapproxlmone] &=\mathbb{E}[\Ql^2 \Qlmone^2] - \mathbb{E}[\Ql^2]\mathbb{E}[ \Qlmone^2] \\
																			  &- 2 \mathbb{E}[\Qlmone] \mathbb{E}[\Ql^2 \Qlmone] + 2 \mathbb{E}[\Qlmone]^2 \mathbb{E}[\Ql^2] \\
																			  &- 2 \mathbb{E}[\Ql] \mathbb{E}[\Ql \Qlmone^2] + 2 \mathbb{E}[\Ql]^2\mathbb{E}[\Qlmone^2] \\
																			  &+ 4 \mathbb{E}[\Ql] \mathbb{E}[\Qlmone] \mathbb{E}[\Ql \Qlmone] - 4 \mathbb{E}[\Ql]^2\mathbb{E}[\Qlmone]^2 .
\end{split}
\end{equation}
\begin{proof}
See \ref{prf:covarianceofvariance} for the proof.
\end{proof}

\end{lemma}
\begin{Appendix}
\section{Proof: Covariance of variance}
\label{prf:covarianceofvariance}
\begin{proof}
Change to centered moments using $\Zli = \Qli - \moneapproxl$ and $\Zlmonei = \Qlmonei - \moneapproxlmone$:
\begin{eqnarray}
\cmtwoapproxl = \frac{1}{\Nl - 1} \sum_{i=1}^{\Nl} {\Zli}^2 - \frac{1}{\Nl (\Nl - 1) } \left(\sum_{i=1}^{\Nl} \Zli \right)^2 \\
\cmtwoapproxlmone = \frac{1}{\Nl - 1} \sum_{i=1}^{\Nl} {\Zlmonei}^2 - \frac{1}{\Nl (\Nl - 1) } \left(\sum_{i=1}^{\Nl} \Zlmonei \right)^2
\end{eqnarray}
Plug back in, and split up covariance
\begin{align}
\Cov[\cmtwoapproxl, \cmtwoapproxlmone] &= \frac{1}{(\Nlmone)^2} \Cov\left[ \Nlsum {\Zli}^2, \Nlsum {\Zlmonei}^2 \right] \tag{c1} \\
																			  &-  \frac{1}{\Nl (\Nlmone)^2} \Cov\left[ \Nlsum {\Zli}^2, \left(\Nlsum {\Zlmonei} \right)^2 \right] \tag{c2}\\
																			  &-  \frac{1}{\Nl (\Nlmone)^2} \Cov\left[ \left(\Nlsum {\Zli} \right)^2, \Nlsum {\Zlmonei}^2 \right] \tag{c3}\\
																			  &+  \frac{1}{\Nl^2 (\Nlmone)^2} \Cov\left[ \left(\Nlsum {\Zli} \right)^2, \left(\Nlsum {\Zlmonei} \right)^2 \right] \tag{c4}
\end{align}
Solve the four different terms independently:
\begin{equation}
\begin{split}
\text{(c1) } &\Cov\left[ \Nlsum {\Zli}^2, \Nlsum {\Zlmonei}^2 \right] = \Nl \Cov[{\Zl}^2, {\Zlmone}^2] \\
					&= \Nl \Cov[\Ql^2 - 2 \Ql \monel + {\monel}^2, \Qlmone^2 - 2 \Qlmone \monelmone + {\monelmone}^2] \\
					&= \Nl \bigg( \Cov[\Ql^2,  \Qlmone^2]		 - 2 \monelmone \Cov[\Ql^2, \Qlmone] - 2 \monel \Cov[\Ql, \Qlmone^2] \\
					&+ 4 \monel \monelmone \Cov[\Ql, \Qlmone]\bigg)
\end{split}
\end{equation}
\begin{equation}
\begin{split}
\text{(c2) } &\Cov \left[ \Nlsum {\Zli}^2, \left(\Nlsum {\Zlmonei}\right)^2\right] = \Cov\left[ \Nlsum {\Zli}^2, \Nlsum \Nlsumj {\Zlmonei}{\Zlmonej}\right] \\
					&= \Cov\left[ \Nlsum {\Zli}^2, \Nlsum {\Zlmonei}^2\right] + \Cov\left[ \Nlsum {\Zli}^2, \Nlsum \Nlsumjneqi {\Zlmonei}{\Zlmonej}\right] \\
					&= \Cov\left[ \Nlsum {\Zli}^2, \Nlsum {\Zlmonei}^2\right] = (c1)\\
					&\text{since (making use of centered Z and independence of $\Zlmonei$ to $\Zlmonej$)} \\
					&\Cov\left[ \Nlsum {\Zli}^2, \Nlsum \Nlsumjneqi {\Zlmonei}{\Zlmonej}\right]  \\
					&= \mathbb{E}\left[ \left(\Nlsum {\Zli}^2\right) \Nlsum \Nlsumjneqi {\Zlmonei}{\Zlmonej}\right] - \mathbb{E}\left[\Nlsum {\Zli}^2\right] \mathbb{E}\left[\Nlsum \Nlsumjneqi {\Zlmonei}{\Zlmonej}\right] \\
					&= \mathbb{E}\left[ \left(\Nlsum {\Zli}^2\right) \Nlsum {\Zlmonei} \Nlsumjneqi {\Zlmonej}\right] - \mathbb{E}\left[\Nlsum {\Zli}^2\right] \mathbb{E}\left[\Nlsum {\Zlmonei} \Nlsumjneqi {\Zlmonej}\right] \\
					&=  \mathbb{E}\left[ \left(\Nlsum {\Zli}^2 \right) \Nlsum {\Zlmonei}\right] \mathbb{E}\left[\Nlsumjneqi {\Zlmonej}\right] - \mathbb{E}\left[\Nlsum {\Zli}^2\right] \mathbb{E}\left[\Nlsum {\Zlmonei}\right] \mathbb{E}\left[\Nlsumjneqi {\Zlmonej}\right] \\
					&=  \mathbb{E}\left[ \left(\Nlsum {\Zli}^2\right) \Nlsum {\Zlmonei}\right] \Nlsumjneqi \mathbb{E}[{\Zlmonej}] - \mathbb{E}\left[\Nlsum {\Zli}^2\right] \mathbb{E}\left[\Nlsum {\Zlmonei}\right] \Nlsumjneqi \mathbb{E}[{\Zlmonej}] \\
					&= 0, {\text{since }} \mathbb{E}[{\Zlmonei}] = \mathbb{E}[\Qlmonei - \monelmone] = 0
\end{split}
\end{equation}
\begin{equation}
\begin{split}
\text{(c3) } &\Cov\left[ \left(\Nlsum {\Zli}\right)^2, \Nlsum {\Zlmonei}^2\right] = \Cov\left[ \Nlsum \Nlsumj \Zli \Zlj, \Nlsum {\Zlmonei}^2\right] \\
			        &=\Cov\left[ \Nlsum {\Zli}^2 + \Nlsum \Nlsumjneqi \Zli \Zlj, \Nlsum {\Zlmonei}^2\right] \\
			        &= \Cov\left[ \Nlsum {\Zli}^2, \Nlsum {\Zlmonei}^2\right] + \Cov\left[\Nlsum \Nlsumjneqi \Zli \Zlj, \Nlsum {\Zlmonei}^2\right] \\
			        &=  \Cov\left[ \Nlsum {\Zli}^2, \Nlsum {\Zlmonei}^2\right] + \Cov\left[\Nlsum \Zli  \Nlsumjneqi\Zlj, \Nlsum {\Zlmonei}^2\right] \\
			        &=  \Cov\left[ \Nlsum {\Zli}^2, \Nlsum {\Zlmonei}^2\right] \\
			        &+ \mathbb{E}\left[\Nlsum \Zli  \Nlsumjneqi\Zlj  \Nlsum {\Zlmonei}^2\right] - \mathbb{E}\left[\Nlsum \Zli  \Nlsumjneqi\Zlj \right]\mathbb{E}\left[\Nlsum {\Zlmonei}^2\right]\\
			        &= \Cov\left[ \Nlsum {\Zli}^2, \Nlsum {\Zlmonei}^2 \right] \\
			        &+ \mathbb{E}\left[\Nlsum \Zli  \Nlsum {\Zlmonei}^2 \right] \mathbb{E}\left[\Nlsumjneqi\Zlj \right] - \mathbb{E}\left[\Nlsum \Zli \right] \mathbb{E}\left[\Nlsumjneqi\Zlj \right] \mathbb{E}\left[\Nlsum {\Zlmonei}^2 \right] \\
			        &= \Cov \left[ \Nlsum {\Zli}^2, \Nlsum {\Zlmonei}^2 \right] = (c1) , {\text{since }} \mathbb{E}[{\Zli}] = 0
\end{split}
\end{equation}
\begin{equation}
\begin{split}
\text{(c4) } &\Cov\left[ \left(\Nlsum {\Zli}\right)^2, \left(\Nlsum {\Zlmonei} \right)^2\right] = \Cov\left[ \Nlsum \Nlsumj {\Zli} \Zlj, \Nlsum \Nlsumj {\Zlmonei} \Zlmonej \right] \\
	                &= \Cov\left[ \Nlsum {\Zli}^2 + \Nlsum \Nlsumjneqi {\Zli} \Zlj, \Nlsum {\Zlmonei}^2 + \Nlsum \Nlsumjneqi {\Zlmonei} \Zlmonej \right] \\
	                &= \Cov\left[ \Nlsum {\Zli}^2, \Nlsum {\Zlmonei}^2 \right] + \Cov\left[\Nlsum {\Zli}^2, \Nlsum \Nlsumjneqi {\Zlmonei} \Zlmonej \right] \\
	                &+ \Cov\left[\Nlsum \Nlsumjneqi {\Zli} \Zlj, \Nlsum {\Zlmonei}^2 \right] + \Cov\left[\Nlsum \Nlsumjneqi {\Zli} \Zlj, \Nlsum \Nlsumjneqi {\Zlmonei} \Zlmonej \right] \\
	                &= \Cov\left[ \Nlsum {\Zli}^2, \Nlsum {\Zlmonei}^2 \right] + 0 + 0 + \Cov\left[\Nlsum \Nlsumjneqi {\Zli} \Zlj, \Nlsum \Nlsumjneqi {\Zlmonei} \Zlmonej \right] \\
	                &= (c1) + \Cov\left[\Nlsum \Nlsumjneqi {\Zli} \Zlj, \Nlsum \Nlsumjneqi {\Zlmonei} \Zlmonej \right] \\
	                &= (c1) + \Nlsum \Nlsumjneqi \Cov \left[{\Zli} \Zlj, {\Zlmonei} \Zlmonej \right] \\
	                &= (c1) + \Nlsum \Nlsumjneqi \mathbb{E}[{\Zli} \Zlj  {\Zlmonei} \Zlmonej] - \mathbb{E}[{\Zli} \Zlj] \mathbb{E}[{\Zlmonei} \Zlmonej] \\
	                &= (c1) + \Nlsum \Nlsumjneqi \mathbb{E}[{\Zli} {\Zlmonei}] \mathbb{E}[ \Zlj \Zlmonej] - \mathbb{E}[{\Zli}]\mathbb{E}[\Zlj] \mathbb{E}[{\Zlmonei}] \mathbb{E}[\Zlmonej] \\
	                &= (c1) + \Nlsum \Nlsumjneqi \mathbb{E}[{\Zli} {\Zlmonei}] \mathbb{E}[ \Zlj \Zlmonej] \\
	                &= (c1) + \Nlsum \Nlsumjneqi \mathbb{E}\left[(\Qli - \monel) (\Qlmonei - \monelmone)\right] \mathbb{E}\left[ (\Qlj - \monel) (\Qlmonej - \monelmone)\right]
\end{split}
\end{equation}
\begin{equation}
\begin{split}
\text{(c4) cont. } 
	                &= (c1) + \Nlsum \Nlsumjneqi \left(\mathbb{E}[\Qli \Qlmonei] - \mathbb{E}[\Qli \monelmone]  - \mathbb{E}[\monel \Qlmonei] + \mathbb{E}[\monel \monelmone] \right) \\
	                &\left( \mathbb{E}[\Qlj \Qlmonej] - \mathbb{E}[\Qlj \monelmone]  - \mathbb{E}[\monel \Qlmonej] + \mathbb{E}[\monel \monelmone]\right) \\
	                &= (c1) + \Nlsum \Nlsumjneqi \left(\mathbb{E}[\Qli \Qlmonei] - \mathbb{E}[\Qli] \monelmone  - \monel \mathbb{E}[\Qlmonei] +\monel \monelmone \right) \\
	                &\left( \mathbb{E}[\Qlj \Qlmonej] - \mathbb{E}[\Qlj] \monelmone  - \monel  \mathbb{E}[\Qlmonej] + \monel \monelmone \right) \\
	                &= (c1) + \Nlsum \Nlsumjneqi \left(\mathbb{E}[\Qli \Qlmonei] - \monel \monelmone  - \monel \monelmone +\monel \monelmone \right) \\
	                &\left( \mathbb{E}[\Qlj \Qlmonej] - \monel \monelmone  - \monel  \monelmone + \monel \monelmone \right) \\
	                &= (c1) + \Nlsum \Nlsumjneqi \left(\mathbb{E}[\Qli \Qlmonei] - \monel \monelmone \right) \left( \mathbb{E}[\Qlj \Qlmonej] - \monel \monelmone \right) \\
	                &= (c1) + \Nl \Nlmone (\mathbb{E}[\Ql \Qlmone] - \monel \monelmone)^2
\end{split}
\end{equation}
Bring it all together and simplify
\begin{align*}
\Cov[\cmtwoapproxl, \cmtwoapproxlmone] &= \frac{1}{(\Nlmone)^2} \Nl \bigg( \Cov[\Ql^2,  \Qlmone^2]- 2 \monelmone \Cov[\Ql^2, \Qlmone] \\
																			  &- 2 \monel \Cov[\Ql, \Qlmone^2]+ 4 \monel \monelmone \Cov[\Ql, \Qlmone]\bigg) \tag{c1}\\
																			  &-  \frac{1}{\Nl (\Nlmone)^2} \Nl \bigg( \Cov[\Ql^2,  \Qlmone^2] - 2 \monelmone \Cov[\Ql^2, \Qlmone] \\
																			  &- 2 \monel \Cov[\Ql, \Qlmone^2] + 4 \monel \monelmone \Cov[\Ql, \Qlmone]\bigg) \tag{c2}\\
																			  &-  \frac{1}{\Nl (\Nlmone)^2} \Nl \bigg( \Cov[\Ql^2,  \Qlmone^2]		 - 2 \monelmone \Cov[\Ql^2, \Qlmone] \\
																			  &- 2 \monel \Cov[\Ql, \Qlmone^2] + 4 \monel \monelmone \Cov[\Ql, \Qlmone]\bigg) \tag{c3}\\
																			  &+  \frac{1}{\Nl^2 (\Nlmone)^2} \bigg[ \Nl \bigg( \Cov[\Ql^2,  \Qlmone^2]		 - 2 \monelmone \Cov[\Ql^2, \Qlmone] \\
																			  &- 2 \monel \Cov[\Ql, \Qlmone^2] + 4 \monel \monelmone \Cov[\Ql, \Qlmone]\bigg) \\
																			  &+ \Nl (\Nlmone) (\mathbb{E}[\Ql \Qlmone] - \monel \monelmone)^2 \bigg] \tag{c4} \\
																			  &= \frac{\Nl^2 - 2\Nl + 1}{\Nl (\Nlmone)^2} \bigg( \Cov[\Ql^2,  \Qlmone^2]- 2 \monelmone \Cov[\Ql^2, \Qlmone] \\
																			  &- 2 \monel \Cov[\Ql, \Qlmone^2]+ 4 \monel \monelmone \Cov[\Ql, \Qlmone]\bigg) \\
																			  &+ \frac{1}{\Nl (\Nlmone)} (\mathbb{E}[\Ql \Qlmone] - \monel \monelmone)^2\\
																			  &= \frac{1}{\Nl} \bigg( \Cov[\Ql^2,  \Qlmone^2]- 2 \monelmone \Cov[\Ql^2, \Qlmone] \\
																			  &- 2 \monel \Cov[\Ql, \Qlmone^2]+ 4 \monel \monelmone \Cov[\Ql, \Qlmone]\bigg) \\
																			  &+ \frac{1}{\Nl (\Nlmone)} (\mathbb{E}[\Ql \Qlmone] - \monel \monelmone)^2
\end{align*}
Substitute covariance term and $\mu$ by expected value
\begin{align*}
\Cov[\cmtwoapproxl, \cmtwoapproxlmone] &= \frac{1}{\Nl} \bigg( \Cov[\Ql^2,  \Qlmone^2]- 2 \monelmone \Cov[\Ql^2, \Qlmone] \\
																			  &- 2 \monel \Cov[\Ql, \Qlmone^2]+ 4 \monel \monelmone \Cov[\Ql, \Qlmone]\bigg) \\
																			  &+ \frac{1}{\Nl (\Nlmone)} (\mathbb{E}[\Ql \Qlmone] - \monel \monelmone)^2 \\
																			  &= \frac{1}{\Nl} \bigg( \mathbb{E}[\Ql^2 \Qlmone^2] - \mathbb{E}[\Ql^2]\mathbb{E}[ \Qlmone^2] - 2 \mathbb{E}[\Qlmone] (\mathbb{E}[\Ql^2 \Qlmone] - \mathbb{E}[\Ql^2]\mathbb{E}[\Qlmone]) \\
																			  &- 2 \mathbb{E}[\Ql] (\mathbb{E}[\Ql \Qlmone^2] - \mathbb{E}[\Ql]\mathbb{E}[\Qlmone^2]) \\
																			  &+ 4 \mathbb{E}[\Ql] \mathbb{E}[\Qlmone] (\mathbb{E}[\Ql \Qlmone] - \mathbb{E}[\Ql]\mathbb{E}[\Qlmone])\bigg) \\
																			  &+ \frac{1}{\Nl (\Nlmone)} (\mathbb{E}[\Ql \Qlmone]^2 -2 \mathbb{E}[\Ql \Qlmone]\mathbb{E}[\Ql] \mathbb{E}[\Qlmone] + \mathbb{E}[\Ql]^2  \mathbb{E}[\Qlmone]^2) \\
																			  &= \frac{1}{\Nl} \bigg( \mathbb{E}[\Ql^2 \Qlmone^2] - \mathbb{E}[\Ql^2]\mathbb{E}[ \Qlmone^2] \\
																			  &- 2 \mathbb{E}[\Qlmone] \mathbb{E}[\Ql^2 \Qlmone] + 2 \mathbb{E}[\Qlmone]^2 \mathbb{E}[\Ql^2] \\
																			  &- 2 \mathbb{E}[\Ql] \mathbb{E}[\Ql \Qlmone^2] + 2 \mathbb{E}[\Ql]^2\mathbb{E}[\Qlmone^2] \\
																			  &+ 4 \mathbb{E}[\Ql] \mathbb{E}[\Qlmone] \mathbb{E}[\Ql \Qlmone] - 4 \mathbb{E}[\Ql]^2\mathbb{E}[\Qlmone]^2)\bigg) \\
																			  &+ \frac{1}{\Nl (\Nlmone)} \bigg(\mathbb{E}[\Ql \Qlmone]^2 - 2 \mathbb{E}[\Ql \Qlmone] \mathbb{E}[\Ql]  \mathbb{E}[\Qlmone] + (\mathbb{E}[\Ql]  \mathbb{E}[\Qlmone])^2 \bigg) \\
\end{align*}
\end{proof}
\end{Appendix}

We note that the product of expected values results in biased estimators, even if each expected value is independently approximated with unbiased estimators. Therefore, as explained in~\ref{APP:mean_products}, we derived unbiased estimators for the double, triple and quadruple products of expected values in Eq.~\eqref{eq:covcostterm}. These expressions are derived in the lemmas reported in \ref{APP:mean_products}, along with their full derivation.

\begin{Appendix}
 \label{APP:mean_products}
\begin{lemma}
Let $(\moneapproxl \moneapproxlmone)_{\text{biased}} = \frac{1}{\Nl} \Nlsum \Qli \frac{1}{\Nl} \Nlsum \Qlmonei$ be a biased estimator for the product of expected value estimators. Then, an unbiased estimator is given as
\begin{equation}
\label{eq:unbiasedpairmean}
\moneapproxl \moneapproxlmone = \frac{\Nl}{\Nl - 1} (\moneapproxl \moneapproxlmone)_{\text{biased}}  -  \frac{1}{\Nl - 1} \moneapproxl[\Ql \Qlmone]
\end{equation}
\end{lemma}
\begin{proof}
\begin{equation}
\begin{split}
\mathbb{E}[ \moneapproxl \moneapproxlmone ] &= \mathbb{E}\bigg[\frac{\Nl}{\Nl - 1} (\moneapproxl \moneapproxlmone)_{\text{biased}}  -  \frac{1}{\Nl - 1} \moneapproxl[\Ql \Qlmone] \bigg] \\
&= \frac{\Nl}{\Nl - 1} \mathbb{E}\bigg[(\moneapproxl \moneapproxlmone)_{\text{biased}}\bigg]  -  \frac{1}{\Nl - 1} \mathbb{E}\bigg[\moneapproxl[\Ql \Qlmone] \bigg] \\
&= \frac{\Nl}{\Nl - 1} \mathbb{E}\bigg[\frac{1}{\Nl^2} \Nlsum \Nlsumj \Qli \Qlmonej \bigg]  -  \frac{1}{\Nl - 1} \mathbb{E}\bigg[\frac{1}{\Nl} \Nlsum \Qli \Qlmonei \bigg] \\
&= \frac{\Nl}{\Nl - 1} \frac{1}{\Nl^2} \Nlsum \Nlsumj \mathbb{E}\bigg[\Qli \Qlmonej \bigg]  -  \frac{1}{\Nl - 1} \frac{1}{\Nl} \Nlsum \mathbb{E}\bigg[\Qli \Qlmonei \bigg] \\
&= \frac{\Nl}{\Nl - 1} \frac{1}{\Nl^2} \Nlsum \mathbb{E}\bigg[\Qli \Qlmonei \bigg] + \frac{\Nl}{\Nl - 1} \frac{1}{\Nl^2} \Nlsum \Nlsumjneqi \mathbb{E}\bigg[\Qli\bigg] \mathbb{E}\bigg[\Qlmonej \bigg] \\ &-  \frac{1}{\Nl - 1} \frac{1}{\Nl} \Nlsum \mathbb{E}\bigg[\Qli \Qlmonei \bigg] \\
&= \frac{1}{\Nl - 1} \mathbb{E}\bigg[\Ql \Qlmone \bigg] + \mathbb{E}\bigg[\Ql\bigg] \mathbb{E}\bigg[\Qlmone \bigg] -  \frac{1}{\Nl - 1} \mathbb{E}\bigg[\Ql \Qlmone \bigg] \\
&= \mathbb{E}\bigg[\Ql\bigg] \mathbb{E}\bigg[\Qlmone \bigg] 
\end{split}
\end{equation}
\end{proof}

\begin{lemma}
Let $\moneapproxlone \moneapproxltwo \moneapproxlthree)_{\text{biased}} = \frac{1}{\Nl} \Nlsum \Qlonei \frac{1}{\Nl} \Nlsum \Qltwoi \frac{1}{\Nl} \Nlsum \Qlthreei$ be a biased estimator for the triple products of expected value estimators. Additionally, assume an unbiased product of mean estimators based on Eq.~\eqref{eq:unbiasedpairmean}. Then, an unbiased estimator is given as
\begin{equation}
\label{eq:unbiasedtriplemean}
\begin{split}
\moneapproxlone \moneapproxltwo \moneapproxlthree = &\frac{\Nl^2}{(\Nl - 1)(\Nl - 2)} (\moneapproxlone \moneapproxltwo \moneapproxlthree)_{\text{biased}} \\
																										   &- \frac{1}{\Nl - 2} \bigg(\moneapproxl[\Qlone \Qltwo] \moneapproxl[\Qlthree]
																										   + \moneapproxl[\Qlone \Qlthree] \moneapproxl[\Qltwo]
																										   +  \moneapproxl[\Qltwo \Qlthree] \moneapproxl[\Qlone] \bigg) \\
																										   &- \frac{1}{(\Nl - 1)(\Nl - 2)} \moneapproxl[\Qlone \Qltwo \Qlthree] 
\end{split}
\end{equation}
\end{lemma}
\begin{proof}
\begin{equation}
\begin{split}
&\mathbb{E}\bigg[ \moneapproxlone \moneapproxltwo \moneapproxlthree \bigg] = \mathbb{E}\bigg[ \frac{\Nl^2}{(\Nl - 1)(\Nl - 2)} (\moneapproxlone \moneapproxltwo \moneapproxlthree)_{\text{biased}} - \frac{1}{(\Nl - 1)(\Nl - 2)} \moneapproxl[\Qlone \Qltwo \Qlthree] \\
																										   &- \frac{1}{\Nl - 2} \bigg(\moneapproxl[\Qlone \Qltwo] \moneapproxl[\Qlthree] 
																										   + \moneapproxl[\Qlone \Qlthree] \moneapproxl[\Qltwo] 
																										   +  \moneapproxl[\Qltwo \Qlthree] \moneapproxl[\Qlone] \bigg)\bigg]  \\
																										   &= \frac{\Nl^2}{(\Nl - 1)(\Nl - 2)} \mathbb{E}\bigg[ (\moneapproxlone \moneapproxltwo \moneapproxlthree)_{\text{biased}}\bigg] 
																										   - \frac{1}{\Nl - 2} \bigg(\mathbb{E}\bigg[\moneapproxl[\Qlone \Qltwo] \moneapproxl[\Qlthree]\bigg] \\
																										   &- \frac{1}{(\Nl - 1)(\Nl - 2)} \mathbb{E}\bigg[\moneapproxl[\Qlone \Qltwo \Qlthree] \bigg] 
																										   + \mathbb{E}\bigg[\moneapproxl[\Qlone \Qlthree] \moneapproxl[\Qltwo]\bigg] 
																										   +  \mathbb{E}\bigg[\moneapproxl[\Qltwo \Qlthree] \moneapproxl[\Qlone]\bigg] \bigg) \\
																										   &= \frac{\Nl^2}{(\Nl - 1)(\Nl - 2)} \frac{1}{\Nl^3} \Nlsum \Nlsumj \Nlsumk \mathbb{E}[ \Qlonei \Qltwoj \Qlthreek] - \frac{1}{(\Nl - 1)(\Nl - 2)} \mathbb{E}\bigg[ \Qlone \Qltwo \Qlthree \bigg] \\
																										   &- \frac{1}{\Nl - 2} \bigg( \mathbb{E}\bigg[ \Qlone \Qltwo \bigg] \mathbb{E}\bigg[ \Qlthree \bigg]
																										   + \mathbb{E}\bigg[ \Qlone \Qlthree \bigg] \mathbb{E}\bigg[\Qltwo \bigg]
																										   +  \mathbb{E}\bigg[ \Qltwo \Qlthree \bigg] \mathbb{E}\bigg[ \Qlone \bigg] \bigg) \\
																										   &= \frac{1}{\Nl (\Nl - 1)(\Nl - 2)} \bigg( \\ 
																										   													&\quad \Nl (\Nl -1) (\Nl - 2) \monelone \moneltwo \monelthree \text{ (case: $i \neq j \neq k$)} \\
																																							&+ \Nl (\Nl - 1) \mathbb{E}[\Qlone \Qltwo] \mathbb{E}[\Qlthree] \text{ (case: $i == j \neq k$)} \\
																																							&+ \Nl (\Nl - 1) \mathbb{E}[\Qlone \Qlthree] \mathbb{E}[\Qltwo] \text{ (case: $i == k \neq j$)} \\
																																							&+ \Nl (\Nl - 1) \mathbb{E}[\Qltwo \Qlthree] \mathbb{E}[\Qlone] \text{ (case: $j == k \neq i$)} \\
																																							&+ \Nl \mathbb{E}[\Qlone \Qltwo \Qlthree]\bigg)- \frac{1}{(\Nl - 1)(\Nl - 2)} \mathbb{E}\bigg[ \Qlone \Qltwo \Qlthree \bigg] \\
																										   &- \frac{1}{\Nl - 2} \bigg( \mathbb{E}\bigg[ \Qlone \Qltwo \bigg] \mathbb{E}\bigg[ \Qlthree \bigg] 
																										   + \mathbb{E}\bigg[ \Qlone \Qlthree \bigg] \mathbb{E}\bigg[\Qltwo \bigg]
																										   +  \mathbb{E}\bigg[ \Qltwo \Qlthree \bigg] \mathbb{E}\bigg[ \Qlone \bigg] \bigg) \\
																										   &= \monelone \moneltwo \monelthree
\end{split}
\end{equation} 
\end{proof}

\begin{lemma}
Let $(\moneapproxl^2 \moneapproxlmone^2)_{\text{biased}} = \frac{1}{\Nl} \Nlsum \Qlone \frac{1}{\Nl} \Nlsum \Qlone \frac{1}{\Nl} \Nlsum \Qlmonei  \frac{1}{\Nl} \Nlsum \Qlmonei $ be a biased estimator. Additionally, assume an unbiased estimator for double and triple product of mean estimators is given based on Eq.~\eqref{eq:unbiasedpairmean} and Eq.~\eqref{eq:unbiasedtriplemean}, respectively. Then, an unbiased estimator for $\moneapproxl^2 \moneapproxlmone^2$ is given as
\begin{equation}
\begin{split}
{\moneapproxl}^2 {\moneapproxlmone}^2 = &\frac{\Nl^3}{(\Nl -1) (\Nl - 2) (\Nl - 3)} (\moneapproxl^2 \moneapproxlmone^2)_{\text{biased}}  \\
 																			- \frac{1}{\Nl - 3} \bigg(&\moneapproxl[\Ql^2] \moneapproxl[\Qlmone]^2 +4\moneapproxl[\Ql \Qlmone] \moneapproxl[\Ql] \moneapproxl[\Qlmone] \\
 																			&+\moneapproxl[\Ql]^2 \moneapproxl[\Qlmone^2] \bigg) \\ 
 																			- \frac{1}{(\Nl -2)(\Nl -3)} \bigg(&\moneapproxl[\Ql^2] \moneapproxl[\Qlmone^2] +2\moneapproxl[\Ql \Qlmone]^2 \\	
																															   +2&\moneapproxl[\Ql^2 \Qlmone] \moneapproxl[\Qlmone]+2\moneapproxl[\Ql] \moneapproxl[\Ql \Qlmone^2] \bigg)  \\		
																			- \frac{1}{(\Nl -1)(\Nl -2)(\Nl -3)} &\moneapproxl[\Ql^2 \Qlmone^2].		
\end{split}
\end{equation}
\end{lemma}
\begin{proof}
\begin{equation}
\begin{split}
&\mathbb{E}\bigg[{\moneapproxl}^2 {\moneapproxlmone}^2 \bigg] = \mathbb{E}\bigg[\frac{\Nl^3}{(\Nl -1) (\Nl - 2) (\Nl - 3)} (\moneapproxl^2 \moneapproxlmone^2)_{\text{biased}}  \\
&- \frac{1}{\Nl - 3} \bigg(\moneapproxl[\Ql^2] \moneapproxl[\Qlmone]^2 +4\moneapproxl[\Ql \Qlmone] \moneapproxl[\Ql] \moneapproxl[\Qlmone] 
+\moneapproxl[\Ql]^2 \moneapproxl[\Qlmone^2] \bigg) \\ 
&- \frac{1}{(\Nl -2)(\Nl -3)} \bigg(\moneapproxl[\Ql^2] \moneapproxl[\Qlmone^2] +2\moneapproxl[\Ql \Qlmone]^2 \\	
&+2\moneapproxl[\Ql^2 \Qlmone] \moneapproxl[\Qlmone]+2\moneapproxl[\Ql] \moneapproxl[\Ql \Qlmone^2] \bigg) - \frac{1}{(\Nl -1)(\Nl -2)(\Nl -3)} \moneapproxl[\Ql^2 \Qlmone^2] \bigg] \\
&= \frac{\Nl^3}{(\Nl -1) (\Nl - 2) (\Nl - 3)} \frac{1}{\Nl^4} \Nlsum \Nlsumj \Nlsumk \Nlsumh \mathbb{E}[ \Qli \Qlj \Qlmonek \Qlmoneh ] \\
&- \frac{1}{\Nl - 3}\bigg(\mathbb{E}[\Ql^2] \mathbb{E}[\Qlmone]^2 +4\mathbb{E}[\Ql \Qlmone] \mathbb{E}[\Ql] \mathbb{E}[\Qlmone]  +\mathbb{E}[\Ql]^2 \mathbb{E}[\Qlmone^2] \bigg)\\ 
&- \frac{1}{(\Nl -2)(\Nl -3)} \bigg(\mathbb{E}[\Ql^2] \mathbb{E}[\Qlmone^2] +2\mathbb{E}[\Ql \Qlmone]^2	+ 2\mathbb{E}[\Ql^2 \Qlmone] \mathbb{E}[\Qlmone]+2\mathbb{E}[\Ql] \mathbb{E}[\Ql \Qlmone^2] \bigg) \\
&- \frac{1}{(\Nl -1)(\Nl -2)(\Nl -3)} \mathbb{E}[\Ql^2 \Qlmone^2] \\
&= \frac{1}{\Nl (\Nl -1) (\Nl - 2) (\Nl - 3)} \bigg(\Nl (\Nl -1) (\Nl - 2) (\Nl - 3) {\monel}^2 {\monelmone}^2 \\
&(\text{one pair: ij, ik, ih, jk, jh, kh}) + \Nl (\Nl -1) (\Nl - 2) \bigg(\mathbb{E}[\Ql^2] \mathbb{E}[\Qlmone]^2 +\mathbb{E}[\Ql \Qlmone] \mathbb{E}[\Ql] \mathbb{E}[\Qlmone] \\
&+\mathbb{E}[\Ql \Qlmone] \mathbb{E}[\Ql] \mathbb{E}[\Qlmone]  +\mathbb{E}[\Ql] \mathbb{E}[\Ql \Qlmone] \mathbb{E}[\Qlmone]  +\mathbb{E}[\Ql] \mathbb{E}[\Qlmone] \mathbb{E}[\Ql \Qlmone]  \\
&+\mathbb{E}[\Ql]^2 \mathbb{E}[\Qlmone^2] \bigg) \\
&(\text{two pairs: ih kh, ik jh, ih jk}) + \Nl (\Nl -1) \bigg(\mathbb{E}[\Ql^2] \mathbb{E}[\Qlmone^2] +\mathbb{E}[\Ql \Qlmone]^2  +\mathbb{E}[\Ql \Qlmone]^2 \bigg) \\		
&(\text{triplets: ijk, ijh, jhk, ihk}) + \Nl (\Nl -1) \bigg(\mathbb{E}[\Ql^2 \Qlmone] \mathbb{E}[\Qlmone] +\mathbb{E}[\Ql^2 \Qlmone] \mathbb{E}[\Qlmone] \\				
&+\mathbb{E}[\Ql] \mathbb{E}[\Ql \Qlmone^2] +\mathbb{E}[\Ql] \mathbb{E}[\Ql \Qlmone^2] \bigg) +(\text{quadruple}) \Nl \mathbb{E}[\Ql^2 \Qlmone^2] \bigg) \\	
&- \frac{1}{\Nl - 3}\bigg(\mathbb{E}[\Ql^2] \mathbb{E}[\Qlmone]^2 +4\mathbb{E}[\Ql \Qlmone] \mathbb{E}[\Ql] \mathbb{E}[\Qlmone]  +\mathbb{E}[\Ql]^2 \mathbb{E}[\Qlmone^2] \bigg)\\ 
&- \frac{1}{(\Nl -2)(\Nl -3)} \bigg(\mathbb{E}[\Ql^2] \mathbb{E}[\Qlmone^2] +2\mathbb{E}[\Ql \Qlmone]^2	+ 2\mathbb{E}[\Ql^2 \Qlmone] \mathbb{E}[\Qlmone]+2\mathbb{E}[\Ql] \mathbb{E}[\Ql \Qlmone^2] \bigg) \\
&- \frac{1}{(\Nl -1)(\Nl -2)(\Nl -3)} \mathbb{E}[\Ql^2 \Qlmone^2] \\
&= {\monel}^2 {\monelmone}^2.
\end{split}
\end{equation}
\end{proof}
\end{Appendix}

Having derived these unbiased estimator we derive an unbiased estimator for the covariance based on the linearity of the expected value.

\begin{lemma}
Let $\cmtwoapproxl$ and $\cmtwoapproxlmone$ be unbiased single level estimators for the respective level $\ell$ and $\ell - 1$ as described in~\eqref{eq:unbiasedcmtwo}. Additionally, let $\moneapprox$ be unbiased estimators for the respective expected value as described in~\eqref{eq:unbiasedmone}. An unbiased estimator for the covariance term in Lemma~\ref{lem:covariance} is given as
\begin{equation}
\label{eq:covvarvar}
\begin{split}
\Cov[\cmtwoapproxl, \cmtwoapproxlmone] &\approx \frac{1}{N_\ell} \moneapprox[\cmtwoapproxl \cmtwoapproxlmone] \\
																			  &+ \frac{1}{N_\ell (N_\ell - 1)} \bigg(\moneapproxl[\Ql \Qlmone] - 2 \moneapproxl[\Ql \Qlmone] \moneapproxl \moneapproxlmone - (\moneapproxl \moneapproxlmone)^2 \bigg)
\end{split}
\end{equation}
where 
\begin{equation}
\label{eq:covcosttermunbiased}
\begin{split}
\moneapprox[\cmtwoapproxl \cmtwoapproxlmone] &= \moneapproxl\left[\Ql^2 \Qlmone^2\right] - 2 \moneapproxl\left[\Ql^2 \Qlmone\right] \moneapproxl\left[\Qlmone\right] \\
																					     &+ 2 \moneapproxl\left[\Qlmone\right]^2 \moneapproxl\left[\Ql^2\right] - 2 \moneapproxl\left[\Ql\right] \moneapproxl\left[\Ql \Qlmone^2\right] \\
																					     &+ 4 \moneapproxl\left[\Qlmone\right] \moneapproxl\left[\Ql\right] \moneapproxl\left[\Ql \Qlmone\right] + 2 \moneapproxl\left[\Ql\right]^2 \moneapproxl\left[\Qlmone^2\right] \\
																					     &- 4 \moneapproxl\left[\Ql\right]^2 \moneapproxl\left[\Qlmone\right]^2 - \moneapproxl\left[\Ql^2\right] \moneapproxl\left[\Qlmone^2\right].
\end{split}
\end{equation}
\begin{proof}
 See \ref{prf:unbiasedestimatorforcovarianceofvariance} for the proof.
\end{proof}
\end{lemma}
\begin{Appendix}
\section{Proof: Unbiased estimator for covariance of variance}
\label{prf:unbiasedestimatorforcovarianceofvariance}
\begin{proof}
We use the previously proven unbiased estimators for products of expected values to show that the estimator is unbiased:
\begin{align*}
\mathbb{E}\bigg[\approxsym{\Cov}[\cmtwoapproxl, \cmtwoapproxlmone]\bigg] &= \frac{1}{N_\ell} \mathbb{E}\bigg[ \moneapprox[\cmtwoapproxl \cmtwoapproxlmone] \bigg] \\
																			  &+ \frac{1}{N_\ell (N_\ell - 1)} \bigg( \mathbb{E}\bigg[\moneapproxl[\Ql \Qlmone]\bigg] - 2 \mathbb{E}\bigg[\moneapproxl[\Ql \Qlmone] \moneapproxl \moneapproxlmone\bigg] - \mathbb{E}\bigg[(\moneapproxl \moneapproxlmone)^2\bigg] \bigg) \\
																			  &= \frac{1}{N_\ell} \mathbb{E}[\cmtwoapproxl \cmtwoapproxlmone] \\
																			  &+ \frac{1}{N_\ell (N_\ell - 1)} \bigg( \mathbb{E}[\Ql \Qlmone] - 2 \mathbb{E}[\Ql \Qlmone] \monel \monelmone - (\monel \monelmone)^2 \bigg)
\end{align*}
\end{proof}
\end{Appendix}

Bringing all parts together and combing the covariance approximation in Eq.~\eqref{eq:covvarvar} with the unbiased variance of variance estimator in Eq.~\eqref{eq:unbiasedvarianceofvariance}, we obtain an unbiased estimator for Eq.~\eqref{eq:mlmcvarianceofvariance}.

\subsection{Multilevel sample allocation for the  standard deviation estimator}
\label{ssec:mlmcforstddev}
For the standard deviation, we leverage the previous results and use the following (biased) MLMC estimator
\begin{equation}
\sigmaapproxml := \sqrt{\cmtwoapproxml},
\end{equation}
based on the MLMC estimator defined in the previous section in Eq.~\eqref{eq:mlmcvarianceestimator}.

The resource allocation problem now reads 
\begin{equation}
\label{eq:mlmcstddev}
\begin{split}
\accentset{\ast}{N}_\ell^{\sigma}= \argmin_{\Nl^{\sigma}} C_T^{\sigma} := \mlsum \Cl \Nl^{\sigma}, \\
\text{s.t. } \mathbb{V}[\sigmaapproxml] = \epsilon_{\sigma}^2.
\end{split}
\end{equation}
where we \emph{target the standard deviation}. 

The complexity of this case stems from the lack of a closed-form solution for the estimator variance. We propose to approximate its variance by resorting again to the delta method,
which, by assuming a normal distribution for the underlying estimator, allows us to write
\begin{equation}
\label{eq:mlmcvarianceofstddev}
\mathbb{V}[\sigmaapproxml] \approx \dfrac{1}{4 \cmtwoapprox^{\text{ML}}} \mathbb{V}[\cmtwoapproxml],
\end{equation}
similarly to what we presented in Section~\ref{ssec:single_fidelity} for its single fidelity counterpart. Also in this case, the solution of the resource allocation problem in Eq.~\eqref{eq:mlmcstddev} requires a numerical optimization.

\subsection{Multilevel sample allocation for the scalarization estimator}
\label{ssec:mlmcforscalarization}
Arguably one of the most important reliability measures for OUU is the linear combination of mean and standard deviation of the QoI, which we denote as scalarization. This approach allows the formulation of the problem without resorting to a multi-objective formulation to deal with mean and standard deviation.  
We use, as often done in literature, a linear combination of the mean and the standard deviation, i.e., $R_{\mu+\alpha\sigma}^b(x) \approx \mathbb{E}[b(x, \theta)] + \alpha \mathbb{V}[b(x, \theta)]^{\frac{1}{2}} $
, where the weight $\alpha$ is introduced to control the variability of the solution. 
 
In this case, the definition of the MLMC estimator is straightforward, since we just need to sum the estimators for the mean and standard deviation defined in the previous sections: $\moneapproxml + \alpha \sigmaapproxml$. The resource allocation, with the variance of the estimator to be defined, is also consistent with the previous cases and reads 
\begin{equation}
\label{eq:mlmcscalarization}
\begin{split}
&\accentset{\ast}{N}_\ell^{\mathbb{E} + \alpha \sigma}= \argmin_{\Nl^{\mathbb{\mathbb{E} + \alpha \sigma}}} C_T^{\mathbb{\mathbb{E} + \alpha \sigma}} := \mlsum \Cl \Nl^{\mathbb{\mathbb{E} + \alpha \sigma}}, \\
&\text{s.t. } \mathbb{V}\left[ \moneapproxml + \alpha \sigmaapproxml\right] = \epsilon_{\mathbb{\mathbb{E} + \alpha \sigma}}^2.
\end{split}
\end{equation}
where we now \emph{target the scalarization}. 

The major challenge is to obtain a traceable expression for the estimator variance which, in a first step, can be expanded as
\begin{equation}
\label{eq:mlmcvarianceofscalarization}
\begin{split}
&\mathbb{V}\left[ \moneapproxml + \alpha \sigmaapproxml\right] = \mathbb{V}\left[\moneapproxml\right] + \alpha^2  \mathbb{V}\left[\sigmaapproxml\right] + 2 \alpha \Cov\left[\moneapproxml, \sigmaapproxml\right].
\end{split}
\end{equation}
We can reuse the previous results for the two variance terms $\mathbb{V}\left[\moneapproxml\right]$ from Eq.~\eqref{eq:mlmcvarianceofmean} and $\mathbb{V}\left[\sigmaapproxml\right]$ from Eq.~\eqref{eq:mlmcvarianceofstddev}. The 
covariance term 
\begin{equation}
\Cov\left[\moneapproxml, \sigmaapproxml\right] = \mlsum \Cov\left[\moneapproxl - \moneapproxlmone, \sigmaapproxl - \sigmaapproxlmone\right],
\end{equation}
however, requires additional derivations and approximations. Indeed, we can write the term further for each level as
\begin{equation}
\label{eq:mlmccovarianceterm}
\begin{split}
&\Cov\left[\moneapproxl - \moneapproxlmone, \sigmaapproxl - \sigmaapproxlmone\right] = \\ 
&\Cov\left[\moneapproxl, \sigmaapproxl \right]  - \Cov\left[\moneapproxl, \sigmaapproxlmone\right]  - \Cov\left[\moneapproxlmone, \sigmaapproxl \right]  + \Cov\left[\moneapproxlmone, \sigmaapproxlmone\right].
\end{split}
\end{equation}

The main difficulty in treating this term is the presence of the square root (in the standard deviation estimator); we propose three different approximations for this terms:
\begin{itemize}
 \item An upper bound based on correlation, named covariance \emph{Pearson} upper bound;
 \item An approximation based on the correlation between mean and variance, named covariance approximation with \emph{correlation lift};
 \item An approximation based on bootstrapping, named covariance approximation using \emph{bootstrap}.
\end{itemize}
We use the terms in italic as a shorthand to refer to the different approximation strategies later on in the result section.

\subsubsection{Covariance Pearson upper bound}
\label{sssec:covarianceupperbound}
The first approach is to use an upper bound for the covariance term by employing its relation with the Pearson correlation coefficient. % in the following lemma:
It is given as 
\begin{equation}
\rho[\moneapproxml, \sigmaapproxml]  = \frac{\Cov\left[\moneapproxml,  \sigmaapproxml \right]}{\sqrt{\mathbb{V}[\moneapproxml]\mathbb{V}[\sigmaapproxml]}},
\end{equation}
where $-1 \leq \rho[\moneapproxml, \sigmaapproxml]  \leq 1$. We can use this lower and upper bound on $\rho[\moneapproxml, \sigmaapproxml]$ to get an upper bound for~\eqref{eq:mlmcvarianceofscalarization} as
 \begin{equation}
\label{eq:variancescalarizationbound}
 \mathbb{V}\left[\moneapproxml + \alpha \sigmaapproxml\right] \leq \mathbb{V}\left[\moneapproxml\right] + \alpha^2  \mathbb{V}\left[\sigmaapproxml\right] + 2 |\alpha| \sqrt{\mathbb{V}[\moneapproxml]\mathbb{V}[\sigmaapproxml]}.
\end{equation}
Using the Pearson correlation we, however, obtain a very conservative estimate since we assume $\rho[\moneapproxml, \sigmaapproxml] = 1$ given that the covariance term can be even negative. 

\subsubsection{Covariance approximation with correlation lift}
\label{sssec:covariancecorrlift}
In the second approach, we present an approximation for the covariance term instead of an upper bound. We use a relationship  for the covariance of the mean and variance estimator (see~\cite{Dodge1999, Oneill2014}):
\begin{lemma}
\label{lem:covmeanvariance}
The covariance of the unbiased sample estimators for mean $\moneapproxl = \frac{1}{\Nl} \Nlsum \Qli$ and variance $\cmtwoapproxl = \frac{1}{\Nl - 1} \Nlsum (\Qli - \moneapproxl)^2$ is given as
\begin{equation}
\label{eq:covmeanvariance}
\Cov[\moneapproxl, \cmtwoapproxl] = \frac{\cmthreel}{\Nl}.
\end{equation} 
where $\cmthreel = \mathbb{E}[(\Ql - \monel)^3]$ is the third central moment.
\end{lemma}
\begin{proof}
See \ref{prf:samelevelcovmoneapproxlcmtwoapproxl} for the proof.
\end{proof}

\begin{Appendix}
\section{Proof: Same level $\Cov[\moneapproxl, \cmtwoapproxl]$}
\label{prf:samelevelcovmoneapproxlcmtwoapproxl}
\begin{proof}
To proof this relation we first need a few ingredients by following the proof given in \cite{Zhang2007}. Similar to the proof for the covariance term of Eq.~\eqref{eq:covariance} we use centered variables $\Zli = \Qli - \monel$ and $\Zlmonei = \Qlmonei - \monelmone$. Using that variable we know that
\begin{equation}
\label{eq:centeredvartonormalvar}
\begin{split}
\cmtwoapproxl[\Zl] &= \frac{1}{\Nl-1} \Nlsum\left(\Zli -  \frac{1}{\Nl} \Nlsumj \Zlj \right) ^2 \\
									&= \frac{1}{\Nl-1} \Nlsum\left(\Qli - \monel -  \frac{1}{\Nl} \Nlsumj \Qlj - \monel \right) ^2 \\
									&= \frac{1}{\Nl-1} \Nlsum\left(\Qli -  \frac{1}{\Nl} \Nlsumj \Qlj \right) ^2 \\
									&= \cmtwoapproxl[\Ql] \\
\end{split}
\end{equation}
We furthermore know that 
\begin{equation}
\label{eq:varformulations}
\begin{split}
\cmtwoapproxl[\Zl] &= \frac{1}{\Nl-1} \Nlsum\left(\Zli -  \frac{1}{\Nl} \Nlsumj \Zlj \right) ^2 \\
								   &= \frac{1}{\Nl-1} \Nlsum (\Zli)^2 - \frac{\Nl}{\Nl-1} \left(\frac{1}{\Nl}\Nlsum\Zli\right)^2
\end{split}
\end{equation}

Additionally, we will later on need the following two equalities for the product of centered random variables:
\begin{equation}
\label{eq:centeredsinglepairllproduct}
\begin{split}
\mathbb{E}[\Nlsum \Nlsumj \Zli (\Zlj)^2] &=\mathbb{E}[ \underbrace{\Nl (\Zli)^3}_{i=j} + \underbrace{\Nl (\Nl-1) \Zli (\Zlj)^2}_{i \neq j} ] \\
																	 &= \Nl \mathbb{E}[(\Zli)^3] + \Nl (\Nl-1) \underbrace{\mathbb{E}[ \Zli ]}_{=0} \mathbb{E}[ (\Zlj)^2 ] \\
																	 &=\Nl \cmthreel
\end{split}
\end{equation}
and
\begin{equation}
\label{eq:centeredcubedllproduct}
\begin{split}
\mathbb{E}[\Nlsum \Nlsumj \Nlsumk  \Zli \Zlj \Zlk] & = \mathbb{E}[ \underbrace{\Nl (\Zli)^3}_{i=j=k} + 3 \underbrace{\Nl (\Nl-1) (\Zli)^2 \Zlj}_{i \neq j = k \lor i = j \neq k \lor i \neq k = j  }  + \Nl (\Nl-1) (\Nl - 2) \Zli \Zlj \Zlk]  \\
&= \Nl \mathbb{E}[ (\Zli)^3] + 3 \Nl (\Nl-1)  \mathbb{E}[(\Zli)^2]  \underbrace{\mathbb{E}[\Zlj]}_{=0} + \Nl (\Nl-1) (\Nl - 2)  \underbrace{\mathbb{E}[\Zli]}_{=0}  \underbrace{\mathbb{E}[\Zlj]}_{=0} \underbrace{\mathbb{E}[\Zlk]}_{=0} \\
&= \Nl \cmthreel 
\end{split}
\end{equation}

Next, we use the following relationship
\begin{equation}
\label{eq:covmeanvarstart}
\begin{split}
\Cov[\moneapproxl, \cmtwoapproxl] &= \mathbb{E}[\moneapproxl \cmtwoapproxl] - \mathbb{E}[\moneapproxl] \mathbb{E}[\cmtwoapproxl] \\
																  &= \mathbb{E}[(\moneapproxl + \monel - \monel )  \cmtwoapproxl] - \monel \cmtwol \\
																  &= \mathbb{E}[(\moneapproxl - \monel) \cmtwoapproxl + \monel \cmtwoapproxl] - \monel \cmtwol \\
																  &= \mathbb{E}[(\moneapproxl - \monel) \cmtwoapproxl] + \monel \mathbb{E}[\cmtwoapproxl] - \monel \cmtwol \\
																  &= \mathbb{E}[(\moneapproxl - \monel) \cmtwoapproxl]
\end{split}
\end{equation}

Given Eq.~\eqref{eq:centeredvartonormalvar} and Eq.~\eqref{eq:varformulations} we can rewrite Eq.~\eqref{eq:covmeanvarstart} in centered form
\begin{equation}
\begin{split}
 \mathbb{E}\left[(\moneapproxl[\Ql] - \monel[\Ql]) \cmtwoapproxl[\Ql]\right]  
 				&= \mathbb{E}\left[\moneapproxl[\Zli] \cmtwoapproxl[\Zli] \right]\\
 			    &=  \mathbb{E}\left[ \left(\frac{1}{\Nl} \Nlsum \Zli \right) \left(\frac{1}{\Nl-1} \Nlsum (\Zli)^2 - \frac{\Nl}{\Nl-1} (\frac{1}{\Nl}\Nlsum\Zli)^2\right)\right]\\
 			    &=  \frac{1}{\Nl} \frac{1}{\Nl-1}   \mathbb{E}\left[ \Nlsum \Zli  \left(\Nlsum (\Zli)^2 -\frac{1}{\Nl} (\Nlsum\Zli)^2\right)\right]\\
 			    &=  \frac{1}{\Nl} \frac{1}{\Nl-1}   \mathbb{E}\left[ \Nlsum \Zli  \Nlsum (\Zli)^2 -\frac{1}{\Nl} \Nlsum \Zli  (\Nlsum\Zli)^2 \right]\\
 			    &=  \frac{1}{\Nl} \frac{1}{\Nl-1}   \mathbb{E}\left[ \Nlsum \Nlsumj \Zli (\Zlj)^2 -\frac{1}{\Nl} \Nlsum \Nlsumj \Nlsumk \Zli \Zlj \Zlk \right]\\
 			    &=  \frac{1}{\Nl} \frac{1}{\Nl-1}   \mathbb{E}\left[ \Nlsum \Nlsumj \Zli (\Zlj)^2 \right] -\frac{1}{\Nl}   \frac{1}{\Nl} \frac{1}{\Nl-1}   \mathbb{E}\left[\Nlsum \Nlsumj \Nlsumk \Zli \Zlj \Zlk \right].
\end{split}
\end{equation}

Next, we can use the relations Eq.~\eqref{eq:centeredsinglepairllproduct} and Eq.~\eqref{eq:centeredcubedllproduct} to finalize the proof
\begin{equation}
\begin{split}
\Cov[\moneapproxl, \cmtwoapproxl] &= \frac{1}{\Nl} \frac{1}{\Nl-1}   \mathbb{E}\left[ \Nlsum \Nlsumj \Zli (\Zlj)^2 \right] -\frac{1}{\Nl}   \frac{1}{\Nl} \frac{1}{\Nl-1}   \mathbb{E}\left[\Nlsum \Nlsumj \Nlsumk \Zli \Zlj \Zlk \right]  \\ 
&= \frac{1}{\Nl} \frac{1}{\Nl-1} \Nl \cmthreel - \frac{1}{\Nl}   \frac{1}{\Nl} \frac{1}{\Nl-1}  \Nl \cmthreel  \\
&=\frac{1}{\Nl-1}\cmthreel - \frac{1}{\Nl (\Nl-1)}  \cmthreel  \\
&=\frac{\cmthreel}{\Nl} \\
\end{split}
\end{equation}
\end{proof}
\end{Appendix}

In Eq.~\eqref{eq:mlmccovarianceterm} we see that we also have covariance terms with estimators of different levels. Specifically, we have term $\Cov[\moneapproxl, \cmtwoapproxlmone]$ and $\Cov[\moneapproxlmone, \cmtwoapproxl]$ for which the previous result does not hold. We derive similar relationships for those terms, i.e., when we have a difference of one level in the estimators but a dependence on the samples:
\begin{lemma}
\label{lem:covmeanvariancelmone}
The covariance of the unbiased sample estimators for mean $\moneapproxl = \frac{1}{\Nl} \Nlsum \Qli$ and variance $\cmtwoapproxlmone = \frac{1}{\Nl - 1} \Nlsum (\Qlmonei - \moneapproxlmone)^2$ is given as
\begin{equation}
\label{eq:covmeanvariancelmone}
\Cov[\moneapproxl, \cmtwoapproxlmone] = \frac{1}{\Nl}\left[\mathbb{E}[\Ql (\Qlmone)^2] - \mathbb{E}[\Ql]\mathbb{E}[(\Qlmone)^2]- 2 \mathbb{E}[\Qlmone] \mathbb{E}[\Ql \Qlmone] + 2 \mathbb{E}[\Ql] \mathbb{E}[\Qlmone]^2\right].
\end{equation} 
\end{lemma}
\begin{proof}
See \ref{prf:samelevelcovmoneapproxlcmtwoapproxlmone} for the proof.
\end{proof}

\begin{Appendix}
\section{Proof: Lower level variance $\Cov[\moneapproxl, \cmtwoapproxlmone]$}
\label{prf:samelevelcovmoneapproxlcmtwoapproxlmone}
\begin{proof}
Similarly to the first proof we use the relation
\begin{equation}
\label{eq:covmeanvarstartlmone}
\begin{split}
\Cov[\moneapproxl, \cmtwoapproxlmone] &= \mathbb{E}[\moneapproxl \cmtwoapproxlmone] - \mathbb{E}[\moneapproxl] \mathbb{E}[\cmtwoapproxlmone] \\
																  &= \mathbb{E}[(\moneapproxl + \monel - \monel )  \cmtwoapproxlmone] - \monel \cmtwolmone \\
																  &= \mathbb{E}[(\moneapproxl - \monel) \cmtwoapproxlmone + \monel \cmtwoapproxlmone] - \monel \cmtwolmone \\
																  &= \mathbb{E}[(\moneapproxl - \monel) \cmtwoapproxlmone] + \monel \mathbb{E}[\cmtwoapproxlmone] - \monel \cmtwolmone \\
																  &= \mathbb{E}[(\moneapproxl - \monel) \cmtwoapproxlmone]
\end{split}
\end{equation}

The relation for product of centered random variables gets a bit more complex:
\begin{equation}
\label{eq:centeredsinglepairllmoneproduct}
\begin{split}
\mathbb{E}&[\Nlsum \Nlsumj \Zli (\Zlmonej)^2] =\mathbb{E}[ \underbrace{\Nl \Zli (\Zlmonei)^2}_{i=j} + \underbrace{\Nl (\Nl-1) \Zli (\Zlmonej)^2}_{i \neq j} ] \\
																	 &= \Nl \mathbb{E}[ \Zli (\Zlmonei)^2] + \Nl (\Nl-1) \underbrace{\mathbb{E}[ \Zli ]}_{=0} \mathbb{E}[ (\Zlmonej)^2 ] \\
																	 &=\Nl \mathbb{E}[ \Zli (\Zlmonei)^2] \\
																	 &= \Nl ( \Cov[\Zli, (\Zlmonei)^2] + \underbrace{\mathbb{E}[\Zli]}_{=0} \mathbb{E}[(\Zlmonei)^2])\\
																	 &= \Nl ( \Cov[\Qli - \monel, (\Qlmonei - \monelmone)^2]) \\
																	 &= \Nl ( \Cov[\Qli, (\Qlmonei)^2 - 2 \Qlmonei \monelmone + (\monelmone)^2]) \\
																	 &= \Nl ( \Cov[\Qli, (\Qlmonei)^2] - 2\monelmone \Cov[\Qli, \Qlmonei] + \underbrace{\Cov[\Qli, (\monelmone)^2])}_{=0} \\
																	 &= \Nl \left[\left( \mathbb{E}[\Qli (\Qlmonei)^2] - \mathbb{E}[\Qli]\mathbb{E}[(\Qlmonei)^2]\right)- 2\monelmone \left(\mathbb{E}[\Qli  \Qlmonei] - \mathbb{E}[\Qli] \mathbb{E}[\Qlmonei]\right)\right]\\
																	 &= \Nl \left[\left( \mathbb{E}[\Qli (\Qlmonei)^2] - \mathbb{E}[\Qli]\mathbb{E}[(\Qlmonei)^2]\right)- 2\mathbb{E}[\Qlmonei] \left(\mathbb{E}[\Qli \Qlmonei] - \mathbb{E}[\Qli] \mathbb{E}[\Qlmonei]\right)\right]\\
																	 &= \Nl \left[\mathbb{E}[\Qli (\Qlmonei)^2] - \mathbb{E}[\Qli]\mathbb{E}[(\Qlmonei)^2]+ 2 \mathbb{E}[\Qlmonei] \mathbb{E}[\Qli \Qlmonei] - 2 \mathbb{E}[\Qli] \mathbb{E}[\Qlmonei]^2\right]\\
\end{split}
\end{equation}
and
\begin{equation}
\label{eq:centeredcubedllmoneproduct}
\begin{split}
\mathbb{E}&[\Nlsum \Nlsumj \Nlsumk  \Zli \Zlmonej \Zlmonek] = \mathbb{E}[ \underbrace{\Nl \Zli (\Zlmonei)^2}_{i=j=k} \\
&+ \underbrace{\Nl (\Nl-1) \Zli (\Zlmonej)^2}_{i \neq j = k}  
+ \underbrace{\Nl (\Nl-1) \Zli \Zlmonei \Zlmonek}_{i = j \neq k}  
+ \underbrace{\Nl (\Nl-1) \Zli (\Zlmonej)^2}_{i \neq k = j  }  \\
&+ \Nl (\Nl-1) (\Nl - 2) \Zli \Zlmonej \Zlmonek]  \\
&= \mathbb{E}[ \Nl \Zli (\Zlmonei)^2] \\
&+ \Nl (\Nl-1) \underbrace{\mathbb{E}[ \Zli]}_{=0} \mathbb{E}[ (\Zlmonej)^2]  
+\mathbb{E}[ \Zli \Zlmonei]  \underbrace{\mathbb{E}[\Zlmonek]}_{=0}
+ \underbrace{\mathbb{E}[ \Zli]}_{=0} \mathbb{E}[ (\Zlmonej)^2]\\
&+ \Nl (\Nl-1) (\Nl - 2) \underbrace{\mathbb{E}[ \Zli ]}_{=0} \underbrace{\mathbb{E}[ \Zlmonej ]}_{=0} \underbrace{\mathbb{E}[ \Zlmonek]}_{=0}\\
&= \mathbb{E}[ \Nl \Zli (\Zlmonei)^2] \\
& \overbrace{=}^{\eqref{eq:centeredsinglepairllmoneproduct}}  \Nl \left[\mathbb{E}[\Qli (\Qlmonei)^2] - \mathbb{E}[\Qli]\mathbb{E}[(\Qlmonei)^2]- 2 \mathbb{E}[\Qlmonei] \mathbb{E}[\Qli \Qlmonei] + 2 \mathbb{E}[\Qli] \mathbb{E}[\Qlmonei]^2\right]\\
\end{split}
\end{equation}

\begin{equation}
\begin{split}
 \mathbb{E}&\left[(\moneapproxl[\Ql] - \monel[\Ql]) \cmtwoapproxlmone[\Ql]\right]  = \mathbb{E}\left[\moneapproxl[\Zli] \cmtwoapproxlmone[\Zli] \right]\\
 			    &=  \mathbb{E}\left[ \left(\frac{1}{\Nl} \Nlsum \Zli \right) \left(\frac{1}{\Nl-1} \Nlsum (\Zlmonei)^2 - \frac{\Nl}{\Nl-1} (\frac{1}{\Nl}\Nlsum\Zlmonei)^2\right)\right]\\
 			    &=  \frac{1}{\Nl} \frac{1}{\Nl-1}   \mathbb{E}\left[ \Nlsum \Zli  \left(\Nlsum (\Zlmonei)^2 -\frac{1}{\Nl} (\Nlsum\Zlmonei)^2\right)\right]\\
 			    &=  \frac{1}{\Nl} \frac{1}{\Nl-1}   \mathbb{E}\left[ \Nlsum \Zli  \Nlsum (\Zlmonei)^2 -\frac{1}{\Nl} \Nlsum \Zlmonei  (\Nlsum\Zlmonei)^2 \right]\\
 			    &=  \frac{1}{\Nl} \frac{1}{\Nl-1}   \mathbb{E}\left[ \Nlsum \Nlsumj \Zli (\Zlmonej)^2 -\frac{1}{\Nl} \Nlsum \Nlsumj \Nlsumk \Zli \Zlmonej \Zlmonek \right]\\
 			    &=  \frac{1}{\Nl} \frac{1}{\Nl-1}   \mathbb{E}\left[ \Nlsum \Nlsumj \Zli (\Zlmonej)^2 \right] -\frac{1}{\Nl}   \frac{1}{\Nl} \frac{1}{\Nl-1}   \mathbb{E}\left[\Nlsum \Nlsumj \Nlsumk \Zli \Zlmonej \Zlmonek \right].
\end{split}
\end{equation}

Again, we use the relation for the product of centered variables from Eq.~\eqref{eq:centeredsinglepairllmoneproduct} and Eq.~\eqref{eq:centeredcubedllmoneproduct}
\begin{equation}
\begin{split}
\Cov[\moneapproxl, \cmtwoapproxlmone]  &=  \frac{1}{\Nl} \frac{1}{\Nl-1}   \mathbb{E}\left[ \Nlsum \Nlsumj \Zli (\Zlmonej)^2 \right] -\frac{1}{\Nl}   \frac{1}{\Nl} \frac{1}{\Nl-1}   \mathbb{E}\left[\Nlsum \Nlsumj \Nlsumk \Zli \Zlmonej \Zlmonek \right] \\
&= \frac{1}{\Nl} \frac{1}{\Nl-1}   \Nl \left[\mathbb{E}[\Qli (\Qlmonei)^2] - \mathbb{E}[\Qli]\mathbb{E}[(\Qlmonei)^2]- 2 \mathbb{E}[\Qlmonei] \mathbb{E}[\Qli \Qlmonei] - 2 \mathbb{E}[\Qli] \mathbb{E}[\Qlmonei]^2\right] \\
& -\frac{1}{\Nl}   \frac{1}{\Nl} \frac{1}{\Nl-1}   \Nl \left[\mathbb{E}[\Qli (\Qlmonei)^2] - \mathbb{E}[\Qli]\mathbb{E}[(\Qlmonei)^2]- 2 \mathbb{E}[\Qlmonei] \mathbb{E}[\Qli \Qlmonei] - 2 \mathbb{E}[\Qli] \mathbb{E}[\Qlmonei]^2\right]\\
&=\frac{1}{\Nl}\left[\mathbb{E}[\Qli (\Qlmonei)^2] - \mathbb{E}[\Qli]\mathbb{E}[(\Qlmonei)^2]- 2 \mathbb{E}[\Qlmonei] \mathbb{E}[\Qli \Qlmonei] + 2 \mathbb{E}[\Qli] \mathbb{E}[\Qlmonei]^2\right]
\end{split}
\end{equation}
\end{proof}
\end{Appendix}

\begin{lemma}
\label{lem:covmeanlmonevariance}
The covariance of the unbiased sample estimators for mean $\moneapproxlmone = \frac{1}{\Nl} \Nlsum \Qlmonei$ and variance $\cmtwoapproxl = \frac{1}{\Nl - 1} \Nlsum (\Qli - \moneapproxl)^2$ is given as
\begin{equation}
\label{eq:covmeanlmonevariance}
\Cov[\moneapproxlmone, \cmtwoapproxl] = \frac{1}{\Nl}\left[\mathbb{E}[\Qlmone (\Ql)^2] - \mathbb{E}[\Qlmone]\mathbb{E}[(\Ql)^2]- 2 \mathbb{E}[\Ql] \mathbb{E}[\Qlmone \Ql] + 2 \mathbb{E}[\Qlmone] \mathbb{E}[\Ql]^2\right].
\end{equation} 
\end{lemma}
\begin{proof}
The proof is the same as for \ref{prf:samelevelcovmoneapproxlcmtwoapproxlmone} by switching $\ell$ and $\ell - 1$.
\end{proof}

Finally, to estimate the covariance term, we make the assumption that the Pearson correlation coefficient of $\moneapprox$ and $\cmtwoapprox$ behaves similarly to the Pearson correlation of $\moneapprox$ and $\sigmaapprox$. Then, we can use the following relation
\begin{equation}
\label{eq:corrlift}
\begin{split}
\rho[\moneapproxli, \sigmaapproxlj] &\approx \rho[\moneapproxli, \cmtwoapproxlj] \\
\Leftrightarrow \frac{\Cov\left[\moneapproxli, \sigmaapproxlj\right]}{\sqrt{\mathbb{V}[\moneapproxli]\mathbb{V}[\sigmaapproxlj]}} &\approx \frac{\Cov\left[\moneapproxli, \cmtwoapproxlj\right]}{\sqrt{\mathbb{V}[\moneapproxli]\mathbb{V}[\cmtwoapproxlj]}} \\
\Leftrightarrow \Cov\left[\moneapproxli, \sigmaapproxlj\right] &\approx \Cov\left[\moneapproxli, \cmtwoapproxlj\right]\sqrt{\frac{\mathbb{V}[\cmtwoapproxlj]}{\mathbb{V}[\sigmaapproxlj]}}, (\ell_i - \ell_j) \in \{-1, 0, 1\},
\end{split}
\end{equation} 
to give us estimator for the different covariance terms in Eq.~\eqref{eq:mlmccovarianceterm}.

\subsubsection{Covariance approximation using bootstrapping}
\label{sssec:covariancebootstrap}
Third, we use bootstrapping as a direct approach for approximating the covariance. The idea of bootstrapping is quite simple. We draw with replacement from the already available set of samples to repeatably compute the estimators. By doing so, we can compute estimates, e.g., for the standard error or bias of an estimator. In our case, we use it to estimate the covariance of the two terms, as, e.g., shown in \cite{Efron1994, Efron2016}. 

Given the set of samples $\{Q_i\}_{i=1}^{\Nl}$ for each level $\ell = 1, ..., L$, we draw with replacement $B$ new \textit{bootstrapped} sets $\{Q_i^b\}_{i=1}^{\Nl}, b = 1, ..., B$. From those we can compute $B$ estimators for the bootstrap mean $\moneapproxbootl = \frac{1}{\Nl} \sum_{i=1}^{\Nl} Q_i^b$ and standard deviation $\sigmaapproxbootl  = \sqrt{\frac{1}{\Nl-1} \sum_{i=1}^{\Nl} (Q_i^b - \moneapproxbootl)^2}$. Finally, we can estimate the covariance as:
\begin{equation}
\label{eq:bootstrap}
\begin{split}
&\Cov[\moneapproxl, \sigmaapproxl] \approx \frac{1}{B-1} \sum_{b=1}^B \left(\moneapproxbootl - \frac{1}{B} \sum_{i=1}^B \moneapproxbootl \right) \left(\sigmaapproxbootl - \frac{1}{B} \sum_{i=1}^B \sigmaapproxbootl \right)
\end{split}
\end{equation}

Results show that this choice and approximation of the covariance term in \eqref{eq:mlmcvarianceofscalarization} is crucial to the performance of the overall method, which we will discuss in the result section.

\subsection{Analytic approximation}
\label{ssec:analyticapproximation}
While we can solve the resource allocation problem for the mean estimator given in Eq.~\eqref{eq:mlmcmean} analytically following the approach as described in Section~\ref{ssec:mlmcformean}, we cannot rely on a closed form solution for the higher order terms presented in our work. Hence, we rely on numerical optimization to solve the optimization problem approximately. We combine this with an approach introduced in~\cite{Krumscheid2020} where the authors introduce an analytical approximation of the resource allocation problem for higher order central moments. 

In their work they state the main assumption that the variance of any higher order sampling estimators, $\mathbb{V}[\approxsym{\mu}_i], i \geq 2,$ decreases in $\mathcal{O}(\frac{1}{N})$ with the number of samples $N$. Following that assumption, one can introduce the variance estimator $ \mathbb{V}[\approxsym{\mu}_i] = \frac{\mathcal{V}[\approxsym{\mu}_i]}{N}$ where the higher order terms of $N$ are now included into $\mathcal{V}[\approxsym{\mu}_i]$, while the only explicit dependence on $N$ is in the term $\frac{1}{N}$; $\mathbb{V}[\approxsym{\mu}_i]$ has the same structure as $\mathbb{V}[\approxsym{\mu}_1] = \frac{\mathbb{V}[Q]}{N}$. 

Extending this approximation to multiple levels, the general variance of a multilevel central moment estimator of higher order is given as 
\begin{equation}
\mathbb{V}[\approxsym{\mu}_{i, \textml}] = \mlsum \frac{\mathcal{V}[\approxsym{\mu}_{i, \ell}]}{\Nl}, i \geq 2,
\end{equation} 
where $\mathcal{V}[\approxsym{\mu}_{i, \ell}] = \mathbb{V}[\approxsym{\mu}_{i, \ell}] \Nl$.
We can solve the problem analytically as for the multilevel mean following~\cite{Giles2008} as
\begin{equation}
\lambda = \epsilon_{\mathbb{X}}^{-2} \mlsum \sqrt{\mathcal{V}[\approxsym{\mu}_{i, \ell}] \Cl},
\end{equation}
for different $\epsilon_{\mathbb{X}}, \mathbb{X} \in \{ \mathbb{E}, \mathbb{V}, \sigma, \mathbb{S} \}$, to compute the resource allocation
\begin{equation}
\Nl^{\mathbb{E}} = \left\lceil \lambda \sqrt{\frac{\mathcal{V}[\approxsym{\mu}_{i, \ell}]}{\Cl}} \right\rceil.
\end{equation}

In our work, we build upon this idea and extend it to the approximation of the variance of the standard deviation $\sigmaapproxml$ and the scalarization $\moneapproxml + \alpha \sigmaapproxml$. While the authors of the work in~\cite{Krumscheid2020} rely on h-statistics to derive general expressions for the variance of central moments of higher order, the variances of standard deviation and scalarization are a main contribution of our work. Additionally, we use this approach in combination with the numerical optimization to compute the resource allocation. Whereas this approach solves the problem analytically by disregarding higher order terms of $\Nl$, the numerical optimization solves the problem numerically while taking higher order terms into account. From an algorithmic standpoint, we can either use the analytic approximation directly, or just use it as an initial guess for the numerical optimization. The two approaches are compared in the next section.

\section{Numerical results}
\label{sec:numericalresults}
In this section, we show the performance of the MLMC estimators in the context of the optimal allocation of resources for different statistics. We consider two numerical problems. The first is a one-dimensional constrained toy problem that we named \say{Problem 18}, which we discuss in Section~\ref{ssec:problem18}. The second case is a popular and well-established optimization benchmark, namely the constrained Rosenbrock function from~\cite{Rosenbrock}. In both cases, we extend the literature test cases to stochastic problems and to also include multiple levels/approximations from which we can define the MLMC estimators. 

By using these examples, we present two sets of results. We first discuss the effectiveness of our MLMC estimators in the evaluation of the statistics, which corresponds to a forward UQ analysis for a fixed design. In this first set, we, hence, focus on sampling only without including the OUU workflow as outer loop. We test numerically, how well the new estimators match their specific target $\epsilon^2_\mathbb{X}, \mathbb{X} \in \{\mathbb{E}, \mathbb{V}, \sigma, \mathbb{E}+\alpha \sigma \}$, where this target $\epsilon^2_\mathbb{X}$ is computed from a Monte Carlo reference solution, i.e. a MC solution with prescribed computational cost. Thus, we can compare how well the newly developed multilevel Monte Carlo statistics match the respective MC reference. We furthermore can compare the performance between the different MLMC estimators targeting different statistics (using different sample allocations $\accentset{\ast}{N}_\ell^{\mathbb{X}}$), but still approximating a particular statistic of interest even if a different target was used for sample allocation. This is useful to demonstrate that, in general, if we prescribe a fixed MC cost and derive the corresponding precision for different targets, e.g., mean and standard deviation, the MLMC sample profile that guarantees the required accuracy in the mean, does not match the precision in the standard deviation at the same time. The reason for this is to demonstrate that, while different statistics can be obtained from post-processing a MLMC allocation that targets the mean, the OUU goals should rather be considered in the allocation in order to obtain the best allocation and the desired accuracy.

Here, we also compare algorithmic choices: we compare the numerical optimization of the resource allocation problem to the approach adapted from~\cite{Krumscheid2020} and presented in Section~\ref{ssec:analyticapproximation}. We further compare the use of iterations as presented in Section~\ref{ssec:implementationdetails}. 

Afterwards, we discuss the combined OUU workflow with both forward UQ at each design location and optimization over the design space.  Furthermore, we emphasize the impact of the covariance term approximation, as discussed in Sections~\ref{sssec:covarianceupperbound},~\ref{sssec:covariancecorrlift} and~\ref{sssec:covariancebootstrap}, on the efficiency of MLMC in the scalarization case. For all the results, we also show the performance enhancements that MLMC can provide with respect to its MC counterpart.

\subsection{Distance metrics}
\label{ssec:distancemetrics}
Before discussing the actual numerical results, we present here the error quantities that we use to extract a quantitative assessment of the OUU workflows. To compare different approaches we run $M^{\mathbb{X}} = 25$ optimization runs for each method where $\mathbb{X} \in \{ \mathbb{E}, \mathbb{S} \}$ stands for the MLMC target. A set of $M^{\text{MC}} = 25$ independent Monte Carlo optimization runs are also used as reference. For each of those runs, we extract the final optimal design $x_i^{\mathbb{X}} \in \mathbb{R}^d$ and $x_i^{\text{MC}} \in \mathbb{R}^d$ found after a fixed number of iterations of the optimizer. We define $x_{i,j}^{\mathbb{X}}$ as the $j$-th element of $x_{i}^{\mathbb{X}}$ and similarly for $x_{i, j}^{\text{MC}}$. This results in sets of final designs for MLMC and MC as $\{x_i^{\mathbb{X}} \}_{i=1}^{M^{\mathbb{X}}}$ and $\{x_i^{\text{MC}} \}_{i=1}^{M^{\text{MC}}}$. We define three metrics:

The first metric is the distance between the centers of the two sets of designs
\begin{equation*}
\text{Dist}^{\mathbb{X}}_C = e(\overline{x}^{\mathbb{X}}, \overline{x}^{\text{MC}}),
\end{equation*}
where $e(x, y) = \sqrt{\sum_{i=1}^d (x_i - y_i) ^2}$ is the Euclidean distance and $\overline{x} = \frac{1}{M} \sum_{i=1}^{M}[ x_{i, 0}, ..., x_{i, d}]$ denotes the vector of centers.  

The second metric considers the standard deviation values of the two sets
\begin{equation*}
\text{Dist}^{\mathbb{X}}_\sigma = d(\hat{\sigma}(x^{\mathbb{X}}), \hat{\sigma}({x}^{\text{MC}})),
\end{equation*}
where $d(x, y) = [|x_0 - y_0|,..., |x_d - y_d|] \in \mathbb{R}^d$ and $\hat{\sigma}(x) = \sqrt{\frac{1}{M-1} \sum_{i=1}^M d (x_{i}, \overline{x})^2}$. Note, that this is a $d$-dimensional metric which shows the difference in the standard deviation in each dimension.

Finally, as the third metric, we take the root-mean-square deviation~\cite{Coutsias2004,Omelyan2019}. Since we are also interested in difference between point clouds with respect to a reference solution, this is a fitting metric, which can also be interpreted as a combination of the previous two metrics. We use it to compute root-mean-square deviation to the Monte Carlo reference solution as
\begin{equation*}
\text{Dist}^{\mathbb{X}}_{\text{RMSdev}} = \sqrt{\frac{1}{M^{\mathbb{X}}} \sum_{i=1}^{N^{\mathbb{X}}}  e(x_i^{\mathbb{X}}, \overline{x}^{\text{MC}})^2 }.
\end{equation*}

\subsection{Implementation details}
\label{ssec:implementationdetails}
In order to couple a given UQ strategy with SNOWPAC, for a generic objective function $R_i^\textml$ it is necessary to provide an estimation of the standard error SE for $R_i^\textml$ as estimate for ${\varepsilon}^b$ in Eq.~\eqref{eq:lower_bound_on_trust_region} and Eq.~\eqref{eq:adjusted_black_box_evaluations}. While the standard error for the sample mean $\moneapprox$ is simply obtained as $\sqrt{\mathbb{V}[\moneapprox]} \approx \sqrt{\frac{\cmtwoapprox}{N}}$, we can now use the derived variances from the previous chapters to also compute the standard error for the other estimators as given in Eq.~\eqref{eq:mlmcvarianceofvariance}, Eq.~\eqref{eq:mlmcvarianceofstddev} and Eq.~\eqref{eq:mlmcvarianceofscalarization} by taking the square root. 

When computing the resource allocation problem, e.g., for the scalarization in Eq.~\eqref{eq:mlmcscalarization}, we need to know all the quantities. Of course, we do not know them in advance or can compute them with a high number of samples (which could be used to solve the problem itself). Hence, we solve the resource allocation problem iteratively. We start with a set of pilot samples which we evaluate and from which we compute a first estimate for the number of samples required on each level. We evaluate this suggested set of samples and can update our estimators. Using this update, we can again solve the resource allocation problem and get an updated estimate for the number of samples. We iteratively proceed until either the suggested number of samples is smaller than the already evaluated number of samples or we reach a maximum number of iterations. We compare the choice of iterations in the result section.

The developed MLMC algorithms and estimators are implemented in DAKOTA \cite{Dakota}, a software toolkit that offers state-of-the-art research and robust, usable algorithms for optimization and UQ. It furthermore includes SNOWPAC as an externally-developed solver. By its modular design, a number of surrogate models can be employed with the optimization algorithm of the user's choice. Applications can be easily coupled, e.g., through the exchange of input and output files. This offers a powerful tool to utilize and exchange a multitude of algorithms in a straight-forward manner. A schematic representation of the interaction among SNOWPAC, the forward UQ problem in Dakota, and the application of interest is visualized in Fig.~\ref{fig:dakotaloop}.
\begin{figure}[h]
\centering
\includegraphics[width=0.5\textwidth]{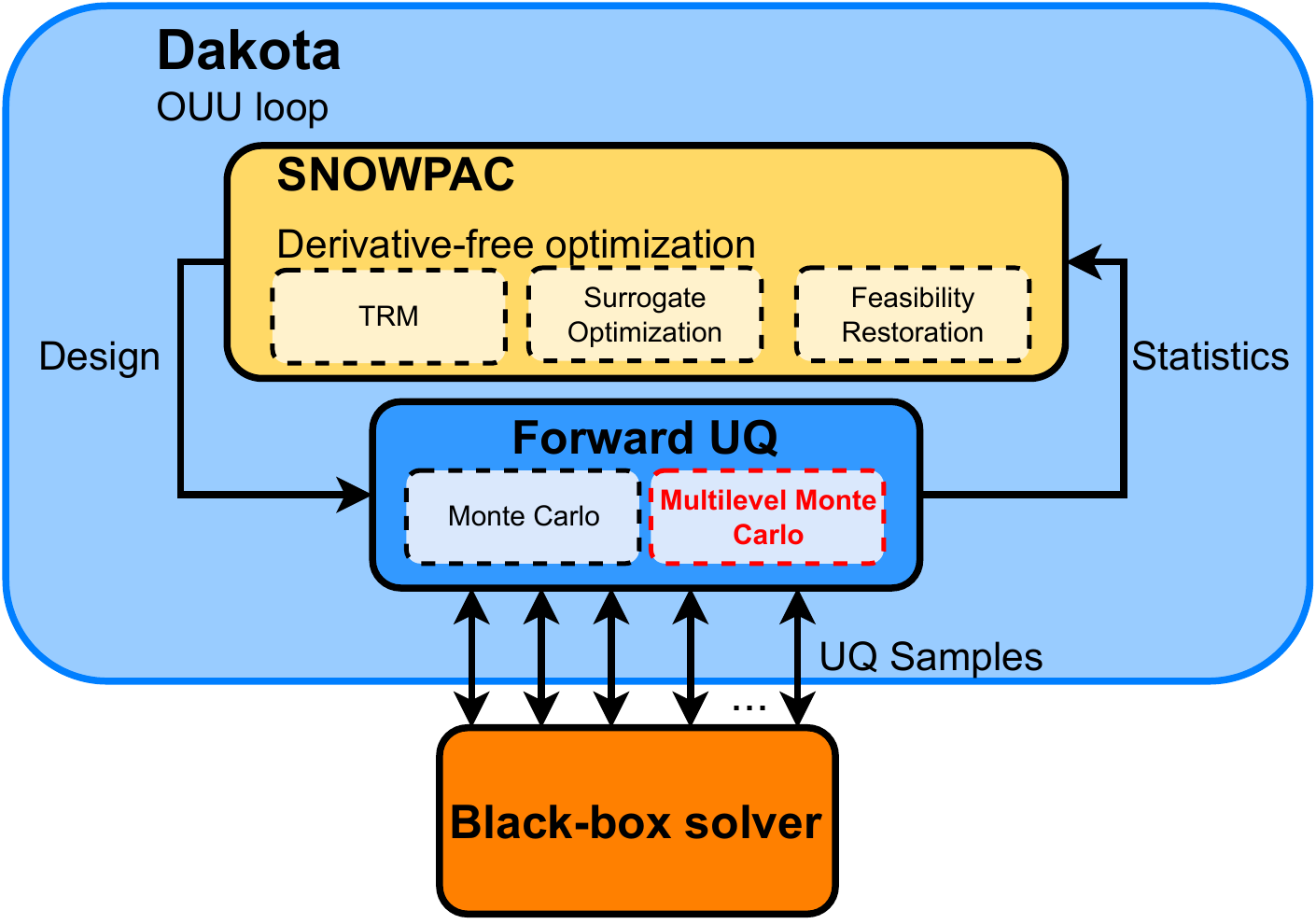}
\caption{Coupling of general surrogate models for a black-box solver with SNOWPAC as outerloop optimization method where the whole process is controlled by Dakota. SNOWPAC is using a trust region management (TRM), surrogate optimization and a feasiblity restoration mode. It communicates the next design step to Dakota which sets up the forward UQ problem at the current design. Dakota is responsible for the sampling, the communication with the black-box application and the collection of the results, employing the previously presented MLMC estimators. Finally, the computed statistics for the objective and constrained are forwarded to SNOWAC for the next optimization step. }
\label{fig:dakotaloop}
\end{figure}

\subsection{Problem 18}
\label{ssec:problem18}
We adapt problem 18 from the website~\cite{1DBenchmark}, which holds a collection of optimization benchmark problems for this first test case. It is a one-dimensional optimization problem where we minimize the function $f_{\text{det}}: \mathbb{R} \rightarrow \mathbb{R}$ in the deterministic case. We also include a linear constraint $g: \mathbb{R} \rightarrow \mathbb{R}$. The two functions are given as 
\begin{equation}
\begin{split}
&f_{\text{det}}(x) = \begin{cases}
    (x - 2)^2       & \quad \text{if } x \leq 3\\
    2 \ln(x - 2) + 1  & \quad \text{if } x > 3
  \end{cases} 
  \end{split}
\end{equation}
and 
\begin{equation}
\begin{split}
&g(x) = \frac{4\ln(1.5)}{5} (x - 1).
  \end{split}
\end{equation}

To create a stochastic problem, we add a random variable $\xi \sim \mathcal{U}(-0.5, 0.5)$ and use different correlation parameters to generate four different levels of $f$ for the multilevel case. This results in the four levels
\begin{equation}
\label{eq:problem18_4levels}
\begin{split}
&f_{4}(x, \xi) = f_{det}(x) + \xi^3, \\
&f_{3}(x, \xi) = f_{det}(x) + 1.1 \xi^3, \\
&f_{2}(x, \xi) = f_{det}(x) + \left( \frac{1}{60} x + 1.2 \right) \xi^3, \\
&f_{1}(x, \xi) = f_{det}(x) + \frac{3}{2} \xi^3,
\end{split}
\end{equation}
where we consider $f_{4}$ as the finest resolution while $\{f_i\}_{i=1}^3$ are coarser levels with set cost $C_1 < C_2 < C_3 < C_4$. We employ the cost ratio $\frac{C_i}{C_{i-1}} = 10$, with $C_4=1$. Due to the additive nature of the stochastic term, we can easily compute a reference solution, e.g., for $\mathbb{E}[f_4(x, \xi)] = f_{\text{det}}(x)$ or $\mathbb{V}[f_4(x, \xi)] = \mathbb{V}[\xi^3] = \frac{0.5^6}{7}$. Note also that for $f_{2}(x, \xi)$ the contribution of $\xi$ depends on $x$, which results in a varying correlation over the levels and, therefore, in a varying resource allocation over $x$. This will become relevant for the optimization results.

\subsubsection{Sampling}
As a first test case, we look at the sampling problem of estimating different measures for the objective function $f_4$ at a certain location $x$. We compare the performance of the standard single-level Monte Carlo estimator with the new multilevel Monte Carlo estimators presented in the previous sections. We assess the estimation quality for approximating the mean $\moneapprox$, the variance $\cmtwoapprox$, the standard deviation $\sigmaapprox$ and the scalarization term $\moneapprox + 3 \sigmaapprox$. Additionally, we compare the different algorithmic choices. The notation of these algorithmic choices for the legends of the upcoming figure are given here in parentheses. We compare the impact of iteratively computing the resource allocation for all estimators in a single (\textit{1 iter}) or 20 iterations (\textit{20 iter}) as described in Section~\ref{ssec:implementationdetails}. We also compare the impact of using a numerical optimization (\textit{Opt)} to compute the resource allocation  
compared to the analytic approximation (\textit{AA}) extended from~\cite{Krumscheid2020} and presented in Section~\ref{ssec:analyticapproximation}. Finally, for $\moneapprox + 3 \sigmaapprox$ we also compare the choice of the covariance approximation as described in Section~\ref{sssec:covarianceupperbound} (\textit{Pearson}),~\ref{sssec:covariancecorrlift} (\textit{CorrLift}) and~\ref{sssec:covariancebootstrap} (\textit{Bootstrap}). 

Without loss of generality, we fix the location $x = 1$ and compute the estimators $1000$ times with a random seed and, each time, we compute the respective resource allocation. From those samples of estimators we plot histograms to show their distribution; additionally, we compute the mean and variance of the distributions. We compare the different MLMC approaches on different targets among each other and also to a single level Monte Carlo.
By fixing $\epsilon_{\mathbb{X}}^2, \mathbb{X} \in \{\mathbb{E}, \mathbb{V}, \sigma, \mathbb{E}+\alpha \sigma \}, $ to be equal to the respective variance of the Monte Carlo reference solution (based on $1000$ samples) we expect the MLMC estimators to match the performance of the single level Monte Carlo at a reduced cost, since both target the same variance. This reference variance for $1000$ samples is computed analytically for this test case, e.g., for the mean: $\epsilon_{\mathbb{E}}^2 = \frac{\mathbb{V}[f_4(x, \xi)]}{1000} = \frac{0.5^6}{7000} \approx 2.2321\text{e-}6$. 

We first estimate the expectation at $\moneapprox$ as a proof of concept and to familiarize ourselves with the results. We use $\epsilon_{\mathbb{E}}^2 \approx 2.2321\text{e-}6$ for the variance $\mathbb{V}[\moneapproxml]$ to compute the resource allocation. The resulting histogram of $1000$ independently computed estimators is given in Figure~\ref{fig:4level_mean_x1} on the left. The two distributions clearly match, with red representing the reference Monte Carlo solution and blue representing the MLMC estimator, as described by \cite{Giles2008}. 
We also investigate the effect of the number of iterations on the quality of the results. If the resource allocation is obtained by taking 20 iterations (20 iter), we see a better match with the reference solution compared to a single iteration (1 iter) of the algorithm. This is due to the error in the statistics introduced by the limited pilot sampling in this latter case, while in the iterated case more samples are added and the statistics are refined until convergence. Finally, we consider the cost in Fig.~\ref{fig:4level_mean_x1} (right), where we are able to observe a large cost reduction compared to the single level Monte Carlo solution. Comparing the number of iterations, we see a narrower peak, and thus a more robust computational cost, when using an iterative approach compared to a single iteration. We also show this quantitatively in Table~\ref{tbl:4level_mean_x1}, where we compare expectations and variances computed from the histograms. We also know the exact value that we want to target, both for the expectation of our estimators as well as for its variance. We see that MLMC Mean (20 iter) performs best in approximating the expected value, but also the variance of the estimator is closest to the target.
\begin{figure}[h]
\centering
\includegraphics[width=0.49\textwidth]{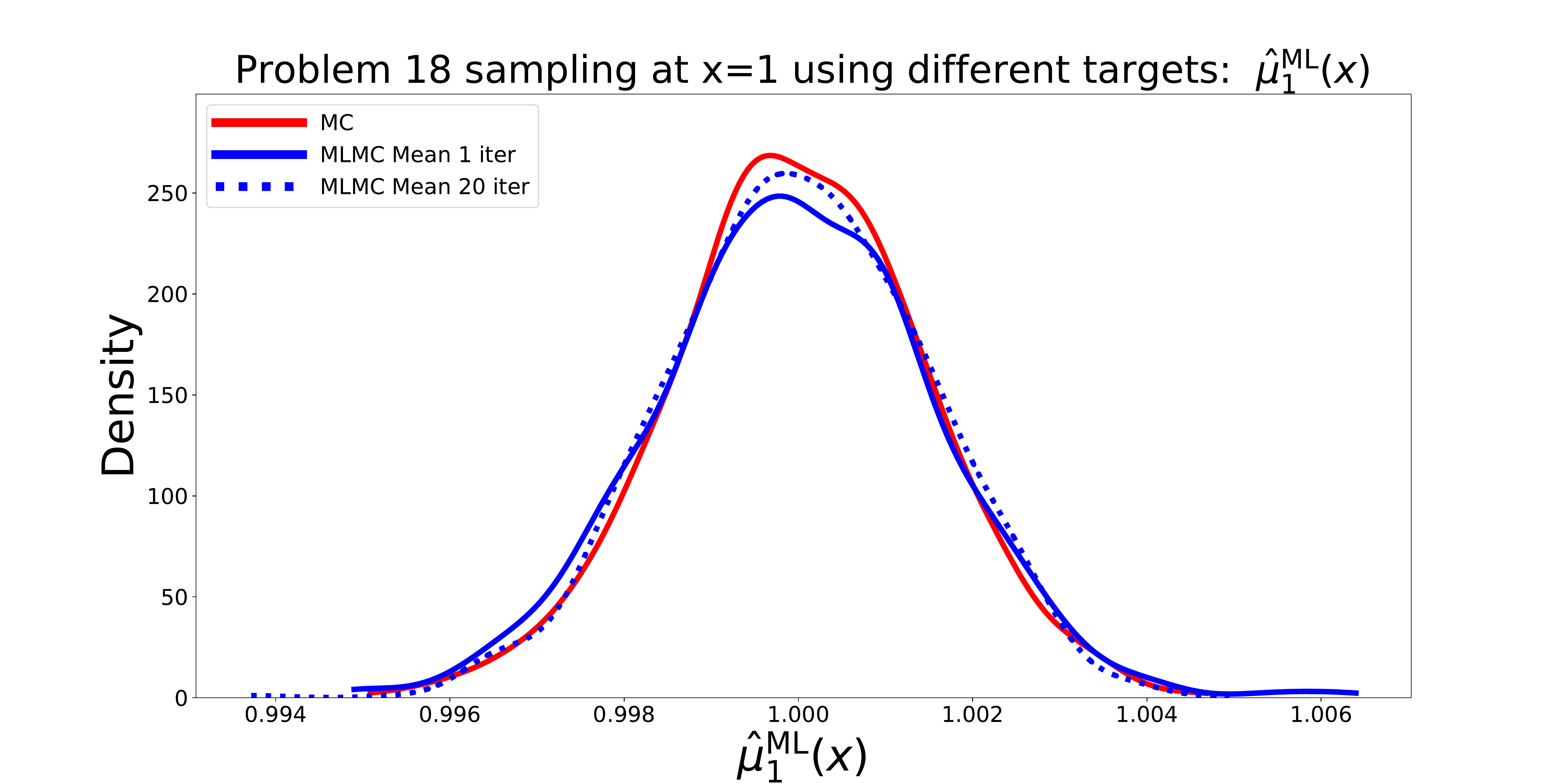}
\includegraphics[width=0.49\textwidth]{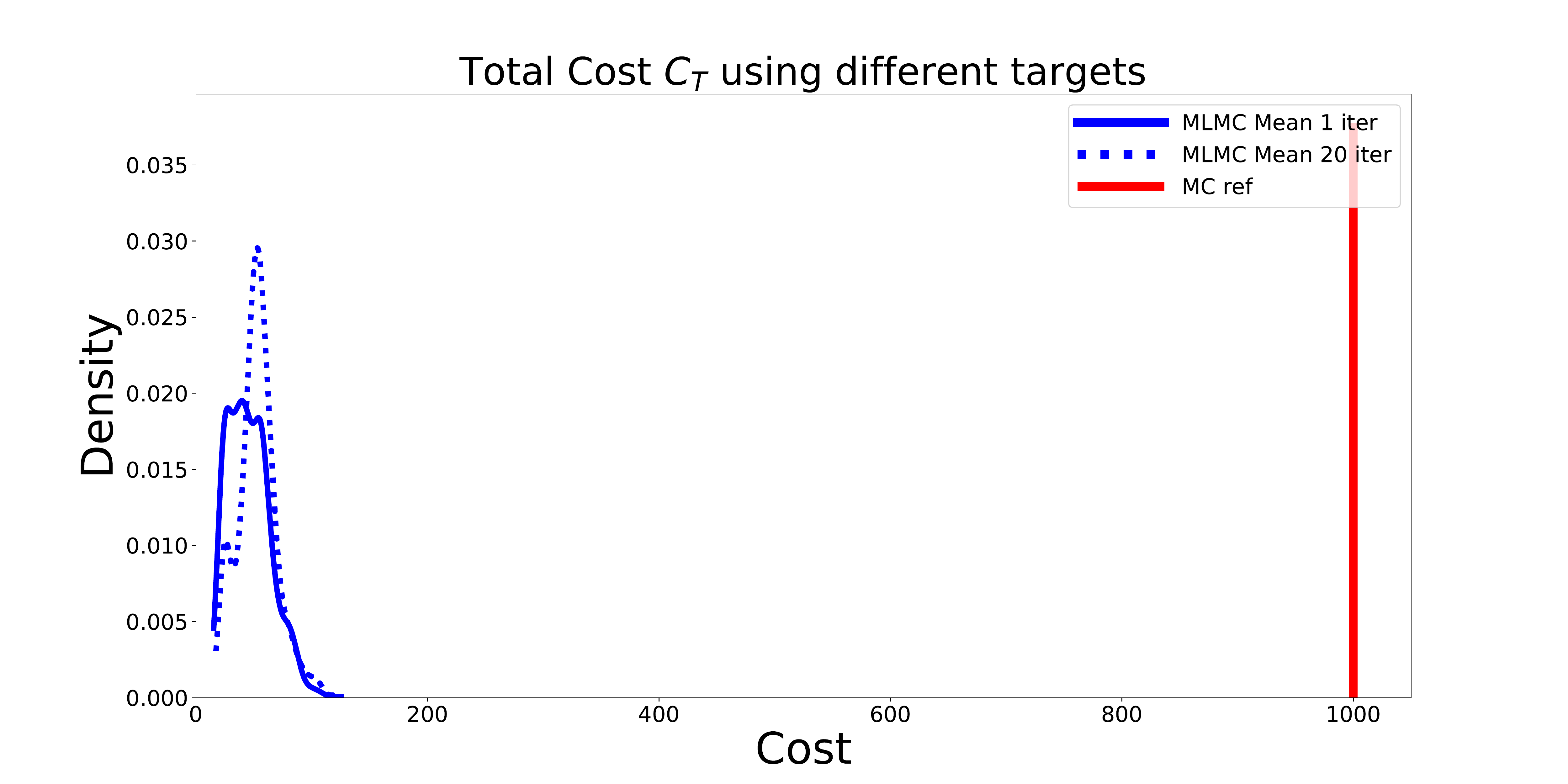}
\caption{\textbf{Mean.} Left: Histogram over $1000$ samples of $\moneapproxml$ for $x=1$ using the different estimators described in Eq.~\eqref{eq:mlmcmean} in blue compared to a reference Monte Carlo estimator $\moneapprox[f_4]$ in red. We compare different algorithmic choices of using an (non-)iterative approach with 1 (1 iter) as solid line or 20 iterations (20 iter) as dashed line for finding the resource allocation. The resource allocation problem is solved analytically. Right: Respective cost for the different estimators.}
\label{fig:4level_mean_x1}
\end{figure}

\begin{table}
\begin{center}
\begin{tabular}{| c || c | c || c | c | }
\hline
Method                &    Mean        &   Exact                    &  Variance  &  Exact \\ \hline
MC                    &    0.99999516  &  \multirow{3}{*}{1.0} &  2.1271e-6 &  \multirow{3}{*}{2.2321e-6} \\ 
MLMC Mean (1 iter)    &    0.99996403  &                            &  2.5850e-6 &  \\ 
MLMC Mean (20 iter)   &    0.99999945  &                            &  2.1821e-6 &  \\  
\hline
\end{tabular}
\end{center}
\caption{Expectations and variances from the histograms of the different approaches in Fig.~\ref{fig:4level_mean_x1}. The column labeled \textit{Exact} shows the target value for expectation and variance.}
\label{tbl:4level_mean_x1}
\end{table}

In the second case, we consider the variance $\cmtwoapprox$, and we compute the resource allocation using $\epsilon_{\mathbb{V}}^2 \approx 1.3823\text{e}-8$ as the target for its variance $\mathbb{V}[\cmtwoapproxml]$. This target is again computed analytically for a Monte Carlo reference solution using 1000 samples. We compare the result not only to a Monte Carlo reference solution but also to a MLMC estimator targeting the mean; in this case, as for the previous case, this estimator uses $\epsilon_{\mathbb{E}}^2 \approx 2.2321\text{e-}6$ as target for its variance $\mathbb{V}[\moneapproxml]$ as given in its resource allocation problem~\eqref{eq:mlmcmean}. The resulting histogram of $1000$ independently computed estimators is given in Fig.~\ref{fig:4level_variance_withmean_x1} where we only compare to the MLMC mean estimator on the left. This is the first case where we see the importance of allocating resources according to the statistics of interest. While we see a good match between the MLMC estimator targeting the variance in orange, we see that the MLMC estimator targeting the mean in blue is under-resolving the estimator, which results in a much wider peak. This wider peak is the result of an underestimation in the resource allocation. This also results in a much lower computational cost for standard MLMC (Note the logarithmic scale on x). This clearly shows the advantage of synchronizing the allocation target with the statistical goal; i.e., by just allocating resources for the mean we cannot expect to compute MLMC estimators for a different statistics while still preserving (or improving) on the MC reference. Regarding the algorithmic implementations for the variance, we compare either using the analytic approximation (AA) as described in Section~\ref{ssec:analyticapproximation} or combining that with a numerical optimization (Opt) in  Fig.~\ref{fig:4level_variance_x1}. Furthermore, for both options, we can use an iterative approach (1 iter or 20 iter). First, we notice an advantage in the iterative approach by reducing the variance in the cost distribution of the estimator. Second, using the numerical optimization in addition to the analytic approximation offers a small improvement in the approximation quality of the estimator targeting $\epsilon_{\mathbb{V}}^2$.  
\begin{figure}[h]
\centering
\includegraphics[width=0.49\textwidth]{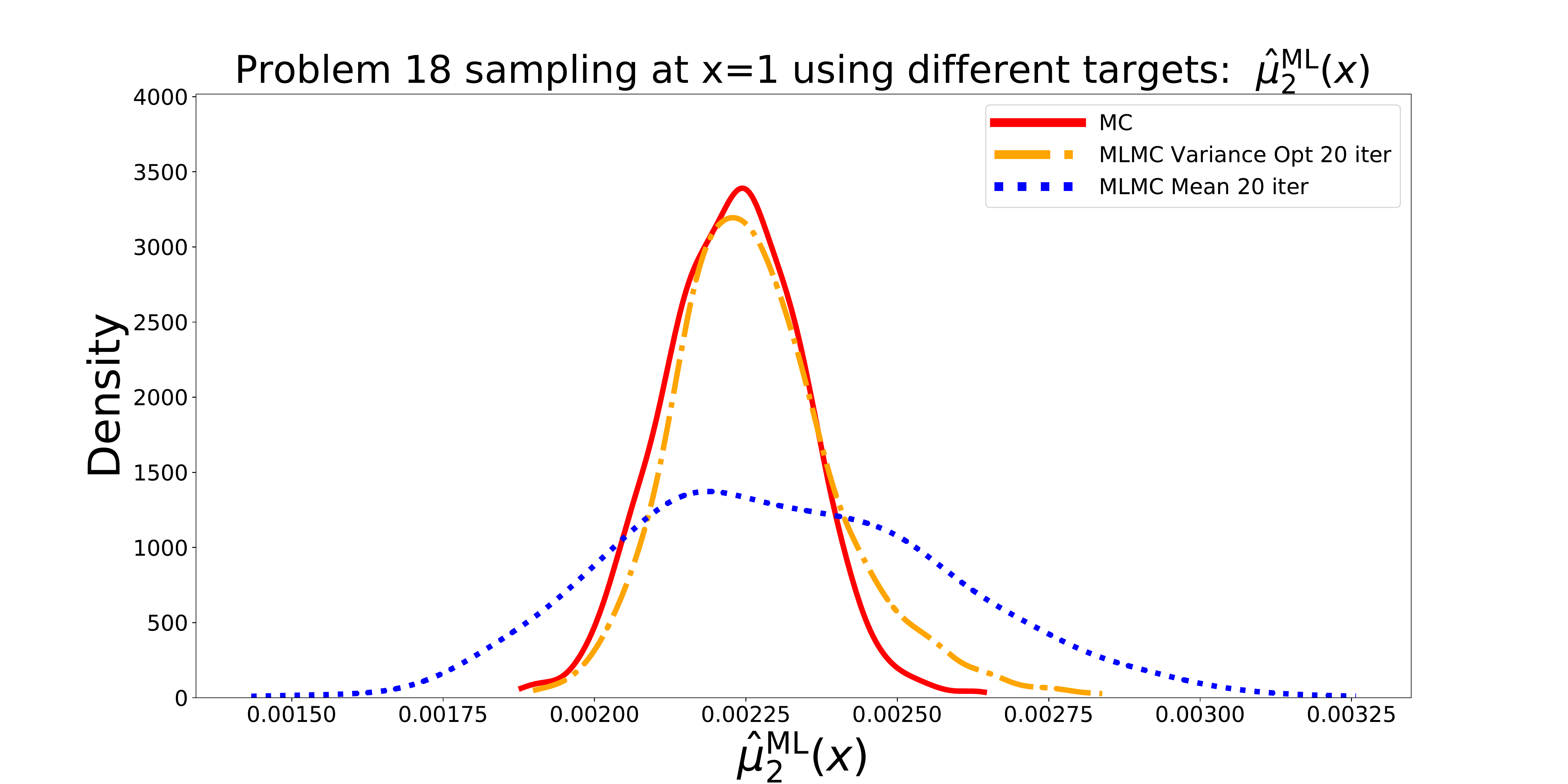}
\includegraphics[width=0.49\textwidth]{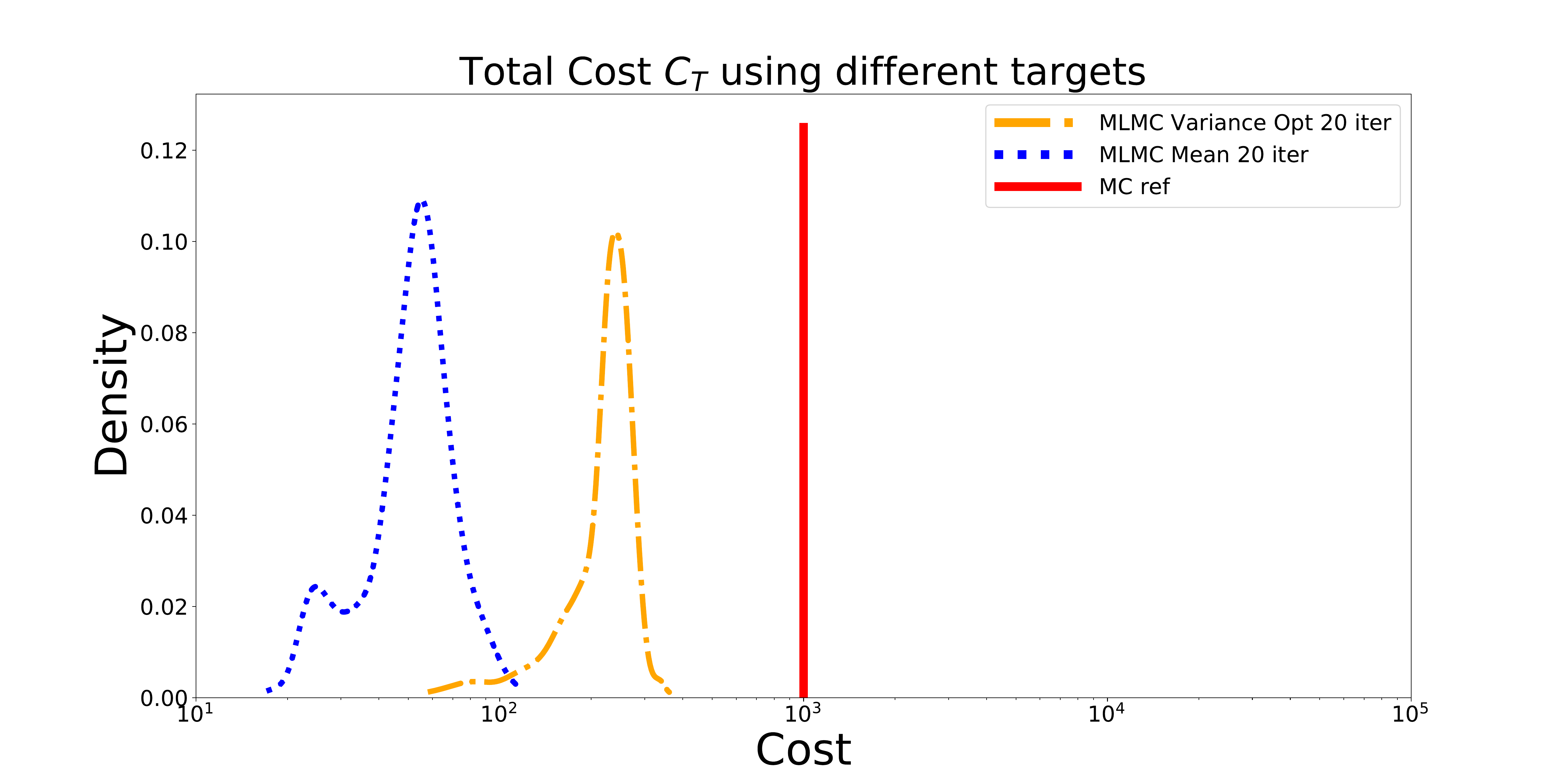}
\caption{\textbf{Variance. }Left: Histogram over $1000$ samples of $\cmtwoapproxml[x]$ for $x=1$ comparing the new estimator described in Eq.~\eqref{eq:mlmcvariance} in orange compared to a reference Monte Carlo estimator in red and using the standard MLMC estimator targeting the mean in blue. Right: Respective cost for the different estimators.}
\label{fig:4level_variance_withmean_x1}
\end{figure}
\begin{figure}[h]
\centering
\includegraphics[width=0.49\textwidth]{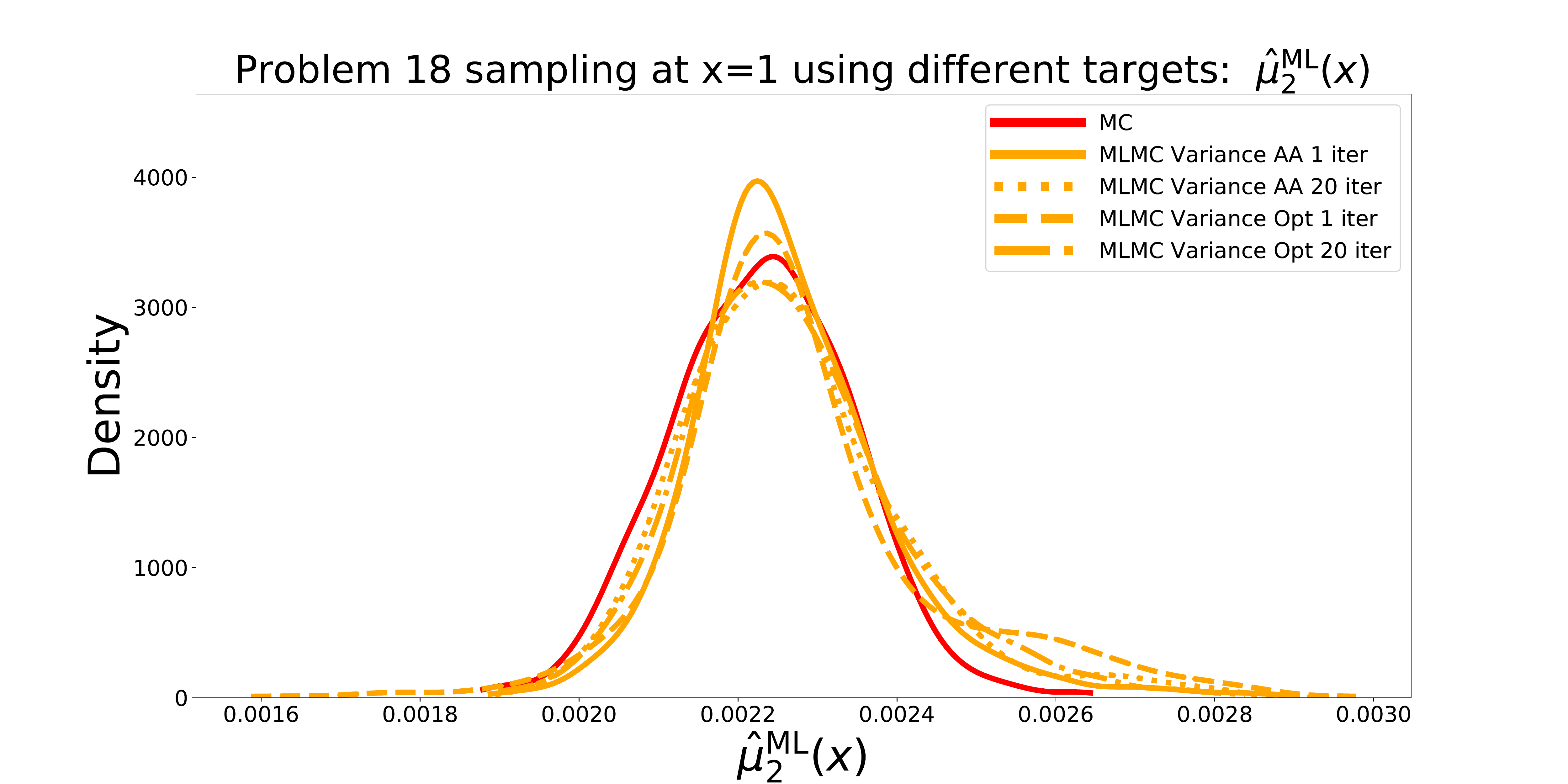}
\includegraphics[width=0.49\textwidth]{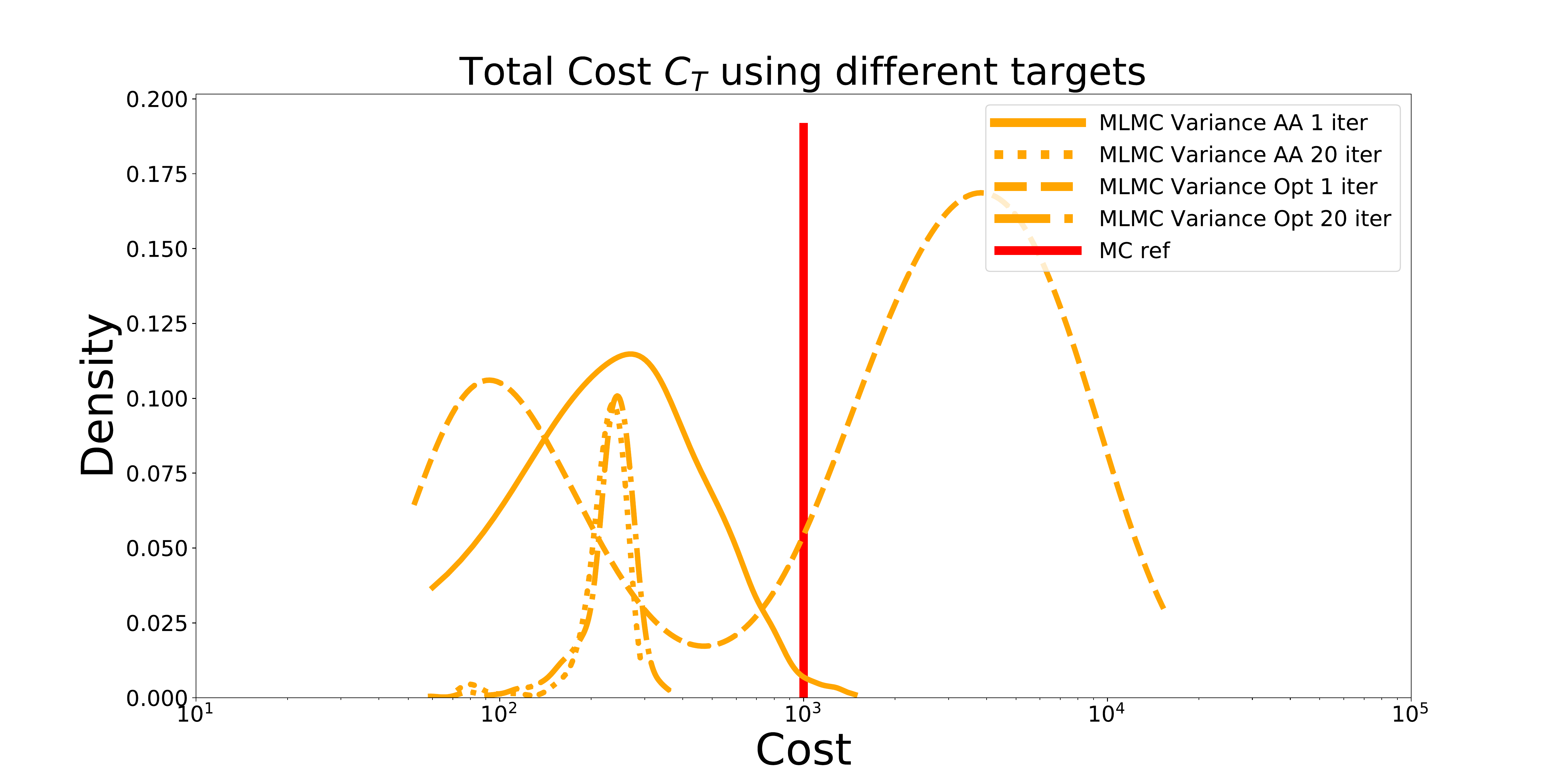}
\caption{\textbf{Variance. }Left: Histogram over $1000$ samples of $\cmtwoapproxml[x]$ for $x=1$ comparing different algorithmic choices for computing the new estimators described in Eq.~\eqref{eq:mlmcvariance} in orange compared to a reference Monte Carlo estimator in red. We compare using an (non-)iterative approach for finding the resource allocation and to use only the analytical approximation or combine it with numerical optimization (solid, dashed, dotted and dashed-dotted). Right: Respective cost for the different estimators.}
\label{fig:4level_variance_x1}
\end{figure}
This is also shown quantitatively in Table~\ref{tbl:4level_variance_x1}.  We see a closer match to the exact solution using our MLMC estimator for the variance, in its expectation as well as in the variance of the estimator itself. Comparing the different approaches, we note that MLMC Variance AA (1 iter) seems to perform best. However, we also have to take the computational cost into account where we point out the large variance of the cost of that approach, even sometimes exceeding the cost of standard Monte Carlo. Therefore, we again prefer the iterative approaches for a better robustness with respect to computational cost.
\begin{table}
\begin{center}
\begin{tabular}{| c || c | c || c | c | }
\hline
Method                &    Mean        &   Exact                    &  Variance  &  Exact \\ \hline
MC                    &    2.2311e-3  &  \multirow{6}{*}{2.2321e-3} &  1.3112e-8 &  \multirow{6}{*}{1.3823e-8} \\ 
MLMC Mean (20 iter)   &    2.3093e-3  &                            &  7.6308e-8 &  \\  
MLMC Variance AA (1 iter)   &    2.2656e-3  &                            &  1.5904e-8 &  \\  
MLMC Variance AA (20 iter)   &    2.2660e-3  &                            &  2.0144e-8 &  \\  
MLMC Variance Opt (1 iter)   &    2.2799e-3  &                            &  2.8185e-8 &  \\  
MLMC Variance Opt (20 iter)   &    2.2680e-3  &                            & 1.8959e-8 &  \\  
\hline
\end{tabular}
\end{center}
\caption{Expectations and variances from the histograms of the different approaches in Fig.~\ref{fig:4level_variance_withmean_x1} and Fig~\ref{fig:4level_variance_x1}. The column labeled \textit{Exact} shows the target value for expectation and variance.}
\label{tbl:4level_variance_x1}
\end{table}

While estimators for the variance were also presented in the work by \cite{Krumscheid2020} (albeit using h-statistics), we now move to the standard deviation $\sigmaapprox$ where we use $\epsilon_{\sigma}^2 \approx 1.5493\text{e-}6$ as the target for the variance $\mathbb{V}[\sigmaapproxml]$ to compute the resource allocation. This target is computed numerically since there is no analytically exact solution (without using an approximation like the delta method) by recomputing the estimator for $1000$ samples for $1000000$ times and computing its variance. Again, when targeting the mean, we use $\epsilon_{\mathbb{E}}^2 \approx 2.2321\text{e-}6$ as in the result for the expected value. This is the first case where we adapt the analytic approximation introduced by \cite{Krumscheid2020} to these new estimators as described in Section~\ref{ssec:analyticapproximation}. The resulting histograms of $1000$ independently computed estimators for different approaches are given in Fig.~\ref{fig:4level_sigma_withmean_x1} and Fig.~\ref{fig:4level_sigma_x1} on the left. Similarly to the variance case, we again see in Fig.~\ref{fig:4level_sigma_withmean_x1} that the MLMC estimator targeting the mean is not well suited for estimating the standard deviation. The MLMC estimator computed using Eq.~\eqref{eq:mlmcstddev}, on the other hand, matches the single level Monte Carlo estimator well. We see the advantage of using an iterative approach combined with numerical optimization in Fig.~\ref{fig:4level_sigma_x1}, especially when regarding the computational cost. When using only a single iteration, the computational cost highly vary, even exceeding the cost of Monte Carlo. The iterative approach on the other side results in a smaller variance in the cost. We also note that adapting the analytic approximation described in Section~\ref{ssec:analyticapproximation} combined with iterations works quite well in this case. These results are also documented quantitatively when we compute the expectations of the histogram and its variance (which approximates to the variance of the estimator) in Table~\ref{tbl:4level_sigma_x1}. We clearly note that better performance in expectation and variance of the estimator for the newly developed MLMC estimator for the standard deviation. As a side note, the effect of $\sigmaapproxml$ being a biased estimator is also visible in a small offset compared to the reference solution.
\begin{figure}[h]
\centering
\includegraphics[width=0.49\textwidth]{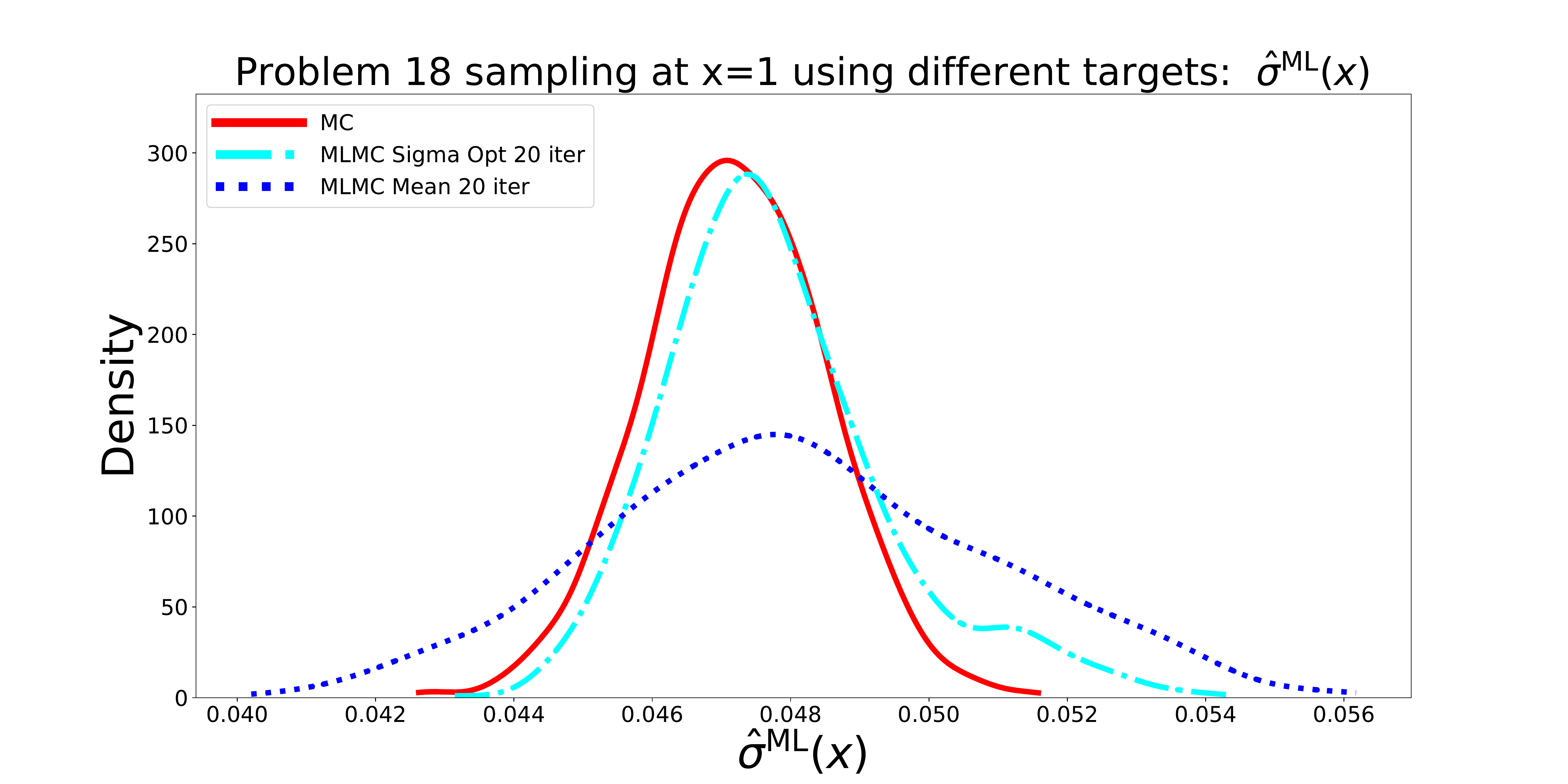}
\includegraphics[width=0.49\textwidth]{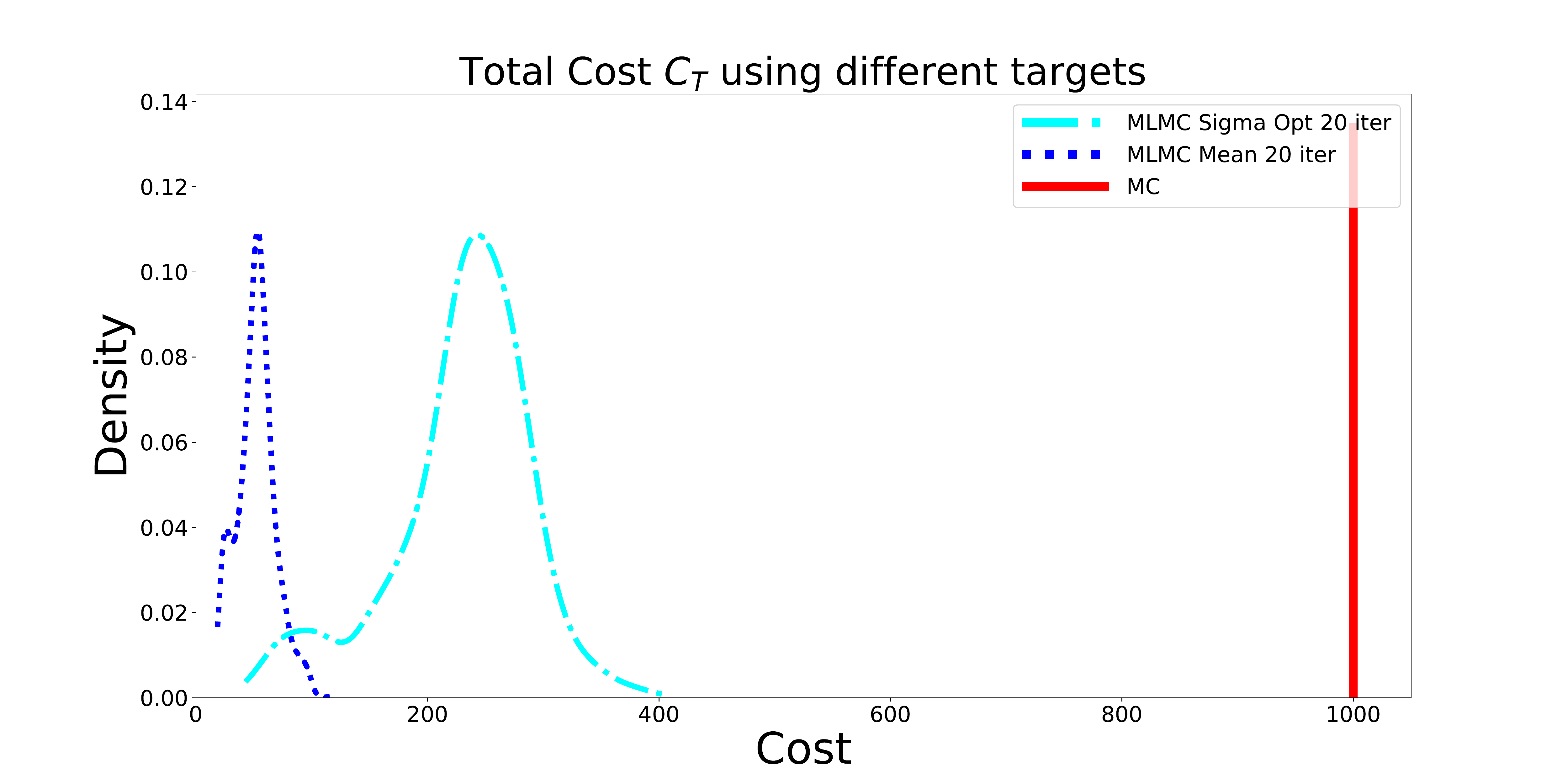}
\caption{\textbf{Standard deviation. }Left: Histogram over $1000$ samples of $\sigmaapproxml[x]$ for $x=1$ comparing the new estimator described in Eq.~\eqref{eq:mlmcstddev} in cyan compared to a reference Monte Carlo estimator in red and using the standard MLMC estimator targeting the mean in blue. Right: Respective cost for the different estimators.}
\label{fig:4level_sigma_withmean_x1}
\end{figure}
\begin{figure}[h]
\centering
\includegraphics[width=0.49\textwidth]{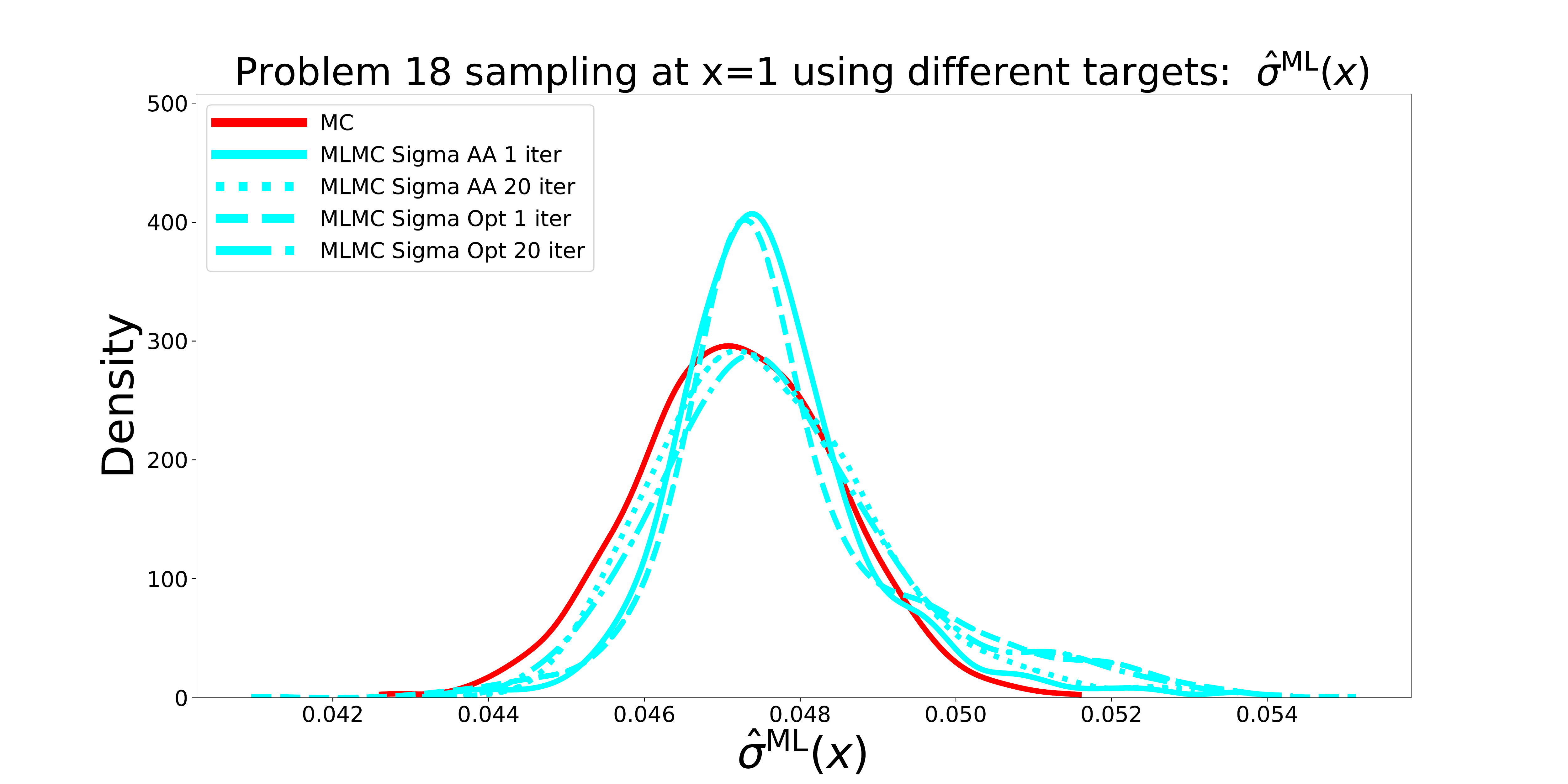}
\includegraphics[width=0.49\textwidth]{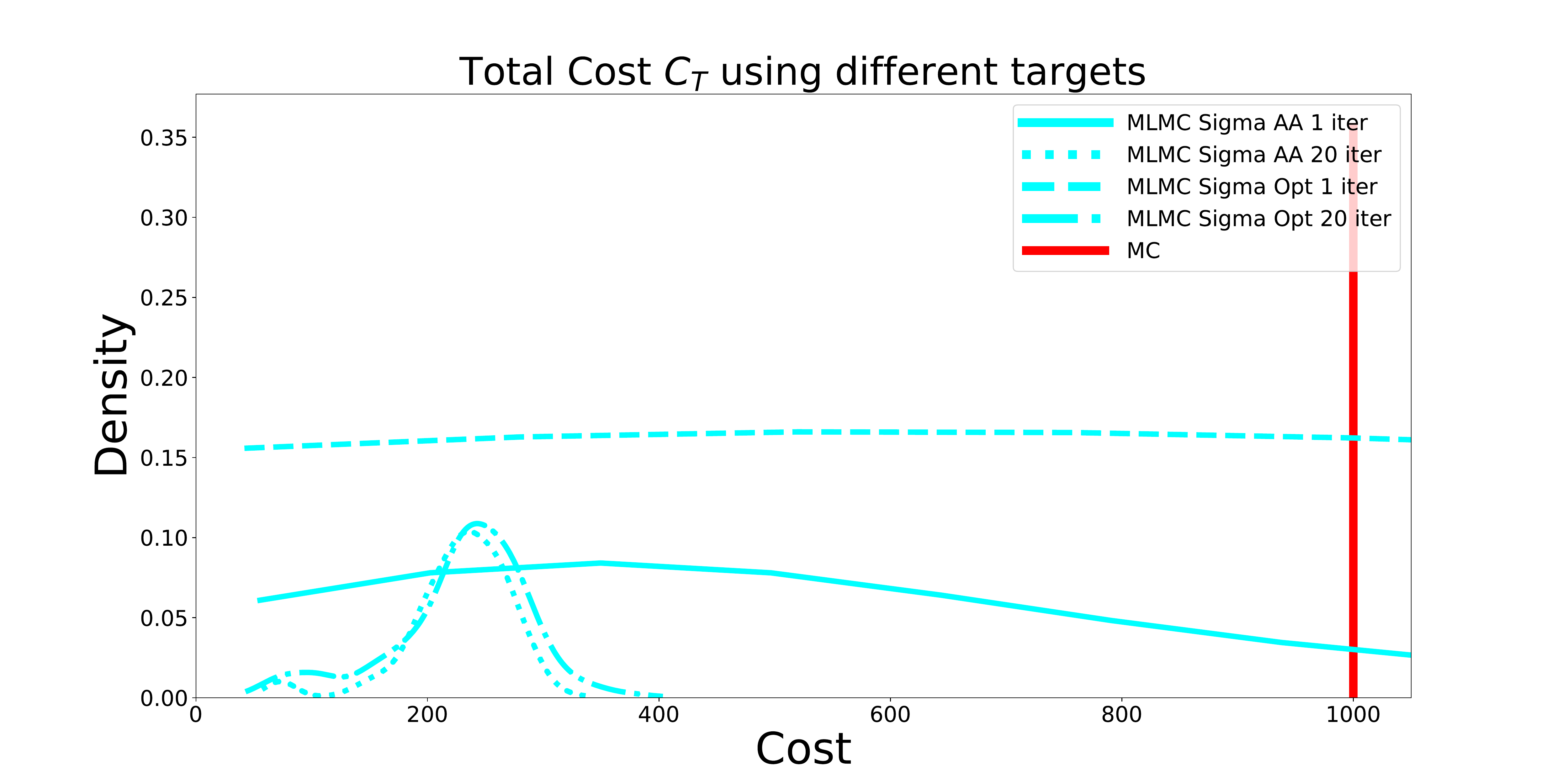}
\caption{\textbf{Standard deviation: }Left: Histogram over $1000$ samples of $\sigmaapproxml[x]$ for $x=1$ comparing different algorithmic choices for computing the new estimators described in Eq.~\eqref{eq:mlmcstddev} in cyan compared to a reference Monte Carlo estimator in red. We compare using an iterative approach for finding the resource allocation (1 iter or 20 iter) and to use only the analytical approximation (AA) or combine it with numerical optimization (Opt). Right: Respective cost for the different estimators.}
\label{fig:4level_sigma_x1}
\end{figure}
\begin{table}
\begin{center}
\begin{tabular}{| c || c | c || c | c | }
\hline
Method                &    Mean        &   Exact                    &  Variance  &  Exact \\ \hline
MC                    &    4.7198e-2  &  \multirow{6}{*}{4.7246e-2} &  1.6967e-6 &  \multirow{6}{*}{1.5493e-6} \\ 
MLMC Mean (20 iter)  & 4.8073e-2 & & 7.8906e-6 &\\ 
MLMC Sigma AA (1 iter)  & 4.7601e-2 &  &  1.6047e-6 &\\ 
MLMC Sigma AA (20 iter)  & 4.7598e-2 & &  2.0572e-6 & \\ 
MLMC Sigma Opt (1 iter)  & 4.7848e-2 &  & 2.6471e-6 &\\
MLMC Sigma Opt (20 iter)  & 4.7785e-2 & & 2.6936e-6 &\\ \hline
\end{tabular}
\end{center}
\caption{Expectations and variances from the histograms of the different approaches in Fig.~\ref{fig:4level_sigma_withmean_x1} and Fig~\ref{fig:4level_sigma_x1}. The column labeled \textit{Exact} shows the target value for expectation and variance.}
\label{tbl:4level_sigma_x1}
\end{table}

Finally, the primary contribution of this work is the new scalarization estimator, $\moneapprox + \alpha \sigmaapprox$. To compute the resource allocation, we choose $\alpha = 3$, which results in $\epsilon_{\mu + \alpha \sigma}^2 \approx 1.6175\text{e-}5$  as the target for the variance $\mathbb{V}[\moneapprox + \alpha \sigmaapprox]$. Also this target is computed numerically by repeatedly computing the estimator using $1000$ samples and computing its variance. By computing histograms from $1000$ samples of the estimator and extracting the corresponding computational cost, we compare three different approaches for computing the covariance term as described in Section~\ref{ssec:mlmcforscalarization}. We plot those against the Monte Carlo reference solution that we want to match. 

In the first case, we use the Pearson correlation to compute the covariance term of $\mathbb{V}[\moneapprox + \alpha \sigmaapprox]$ as described in Section~\ref{sssec:covarianceupperbound}. The result is visualized in Fig.~\ref{fig:4level_meansigma_x1_covpearson}. We see the effect of using the upper bound for estimating the variance: the estimators are over-resolved, which results in a smaller variance compared to the target. While we still have a good match with the reference solution, we incur an unnecessary amount of computational cost due to the conservative approximation. Note again that we see the effect of the biased estimator coming from the estimation of the standard deviation.
\begin{figure}[h]
\centering
\includegraphics[width=0.49\textwidth]{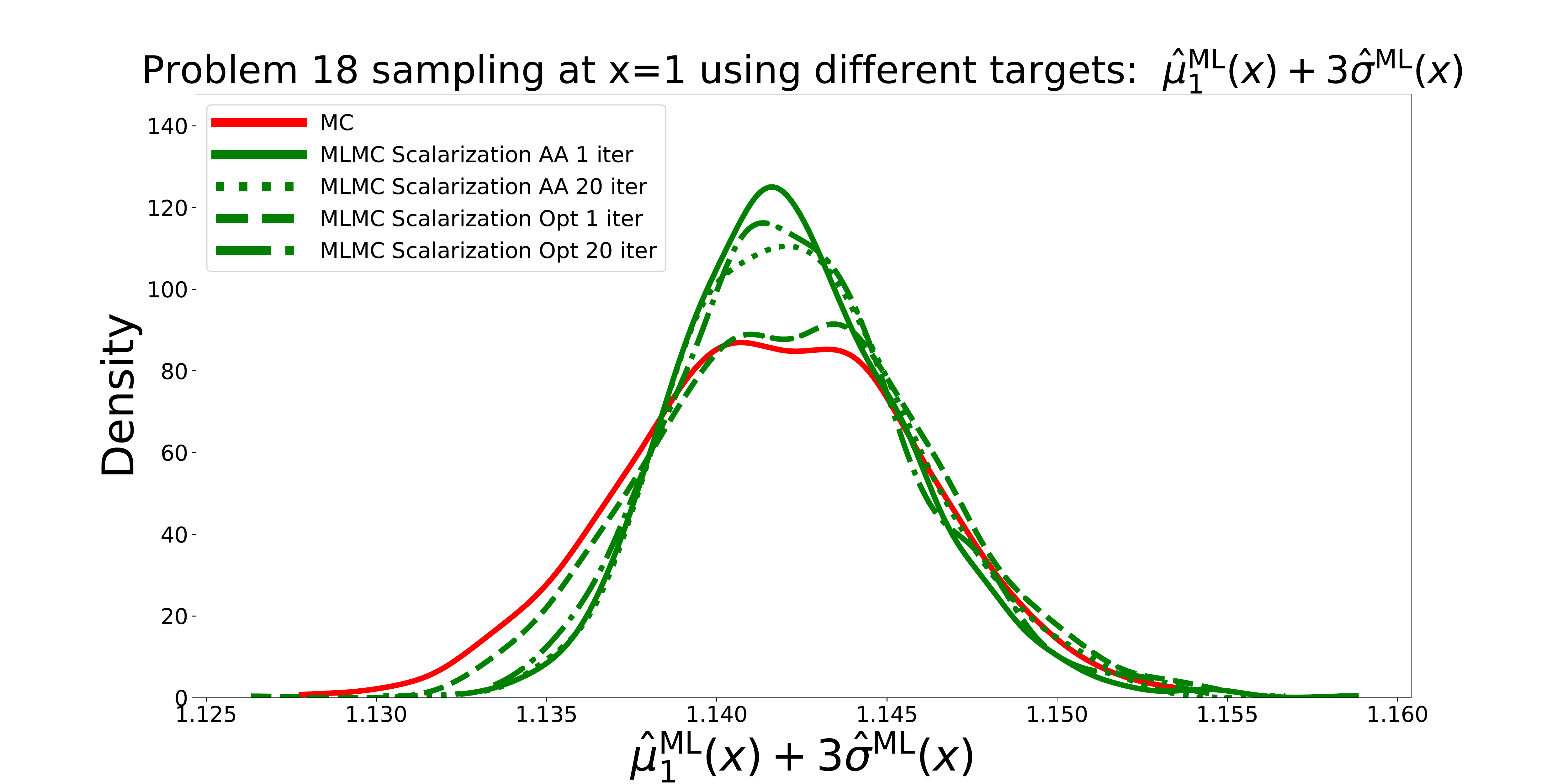}
\includegraphics[width=0.49\textwidth]{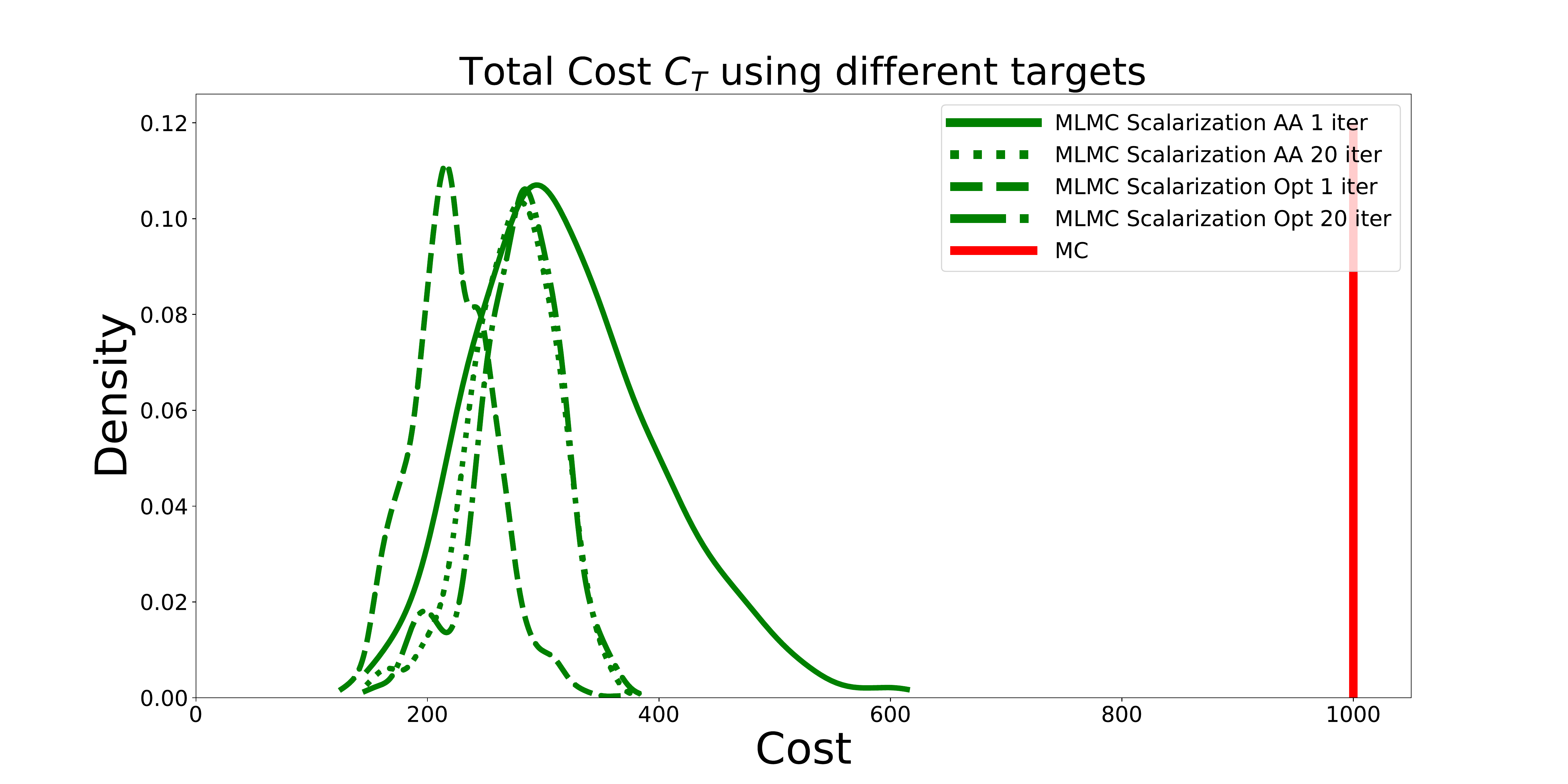}
\caption{\textbf{Scalarization (Pearson). }Left: Histogram over $1000$ samples of $\moneapprox[x] + \alpha \sigmaapprox[x]$ for $x=1$ using the scalarization estimator described in Eq.~\eqref{eq:mlmcscalarization} in green in combination with using the Pearson correlation property described in Section~\ref{sssec:covarianceupperbound} to bound the covariance term in Eq.~\eqref{eq:mlmcvarianceofscalarization}. We compare to a Monte Carlo reference estimator in red. We compare using an iterative approach for finding the resource allocation (1 iter or 20 iter) and to use only the analytical approximation (AA) or combine it with numerical optimization (Opt). Right: Respective cost for the different estimators.}
\label{fig:4level_meansigma_x1_covpearson}
\end{figure}

In the second case, we use the Bootstrap approximation to compute the covariance term of $\mathbb{V}[\moneapprox + \alpha \sigmaapprox]$ as described in Section~\ref{sssec:covariancebootstrap}. Hence, instead of a conservative upper bound, we now use an approximation. The improvement is visible in Fig.~\ref{fig:4level_meansigma_x1_covbootstrap}. We see an improved match of the histogram with the target function in the left figure. We also see the best results when using 20 iterations and a numerical optimization. When we look at the computational cost on the right, we can see that we require less computational cost than in Fig.~\ref{fig:4level_meansigma_x1_covpearson} given the approximation rather than the upper bound. One downside, not visualized here, is the computational cost of using the bootstrap. This gets especially expensive when combining it with the numerical optimization; while it might still be negligible when we apply the algorithm to expensive black-box functions, it is a non-negligible cost compared to the evaluation of analytic functions.
\begin{figure}[h]
\centering
\includegraphics[width=0.49\textwidth]{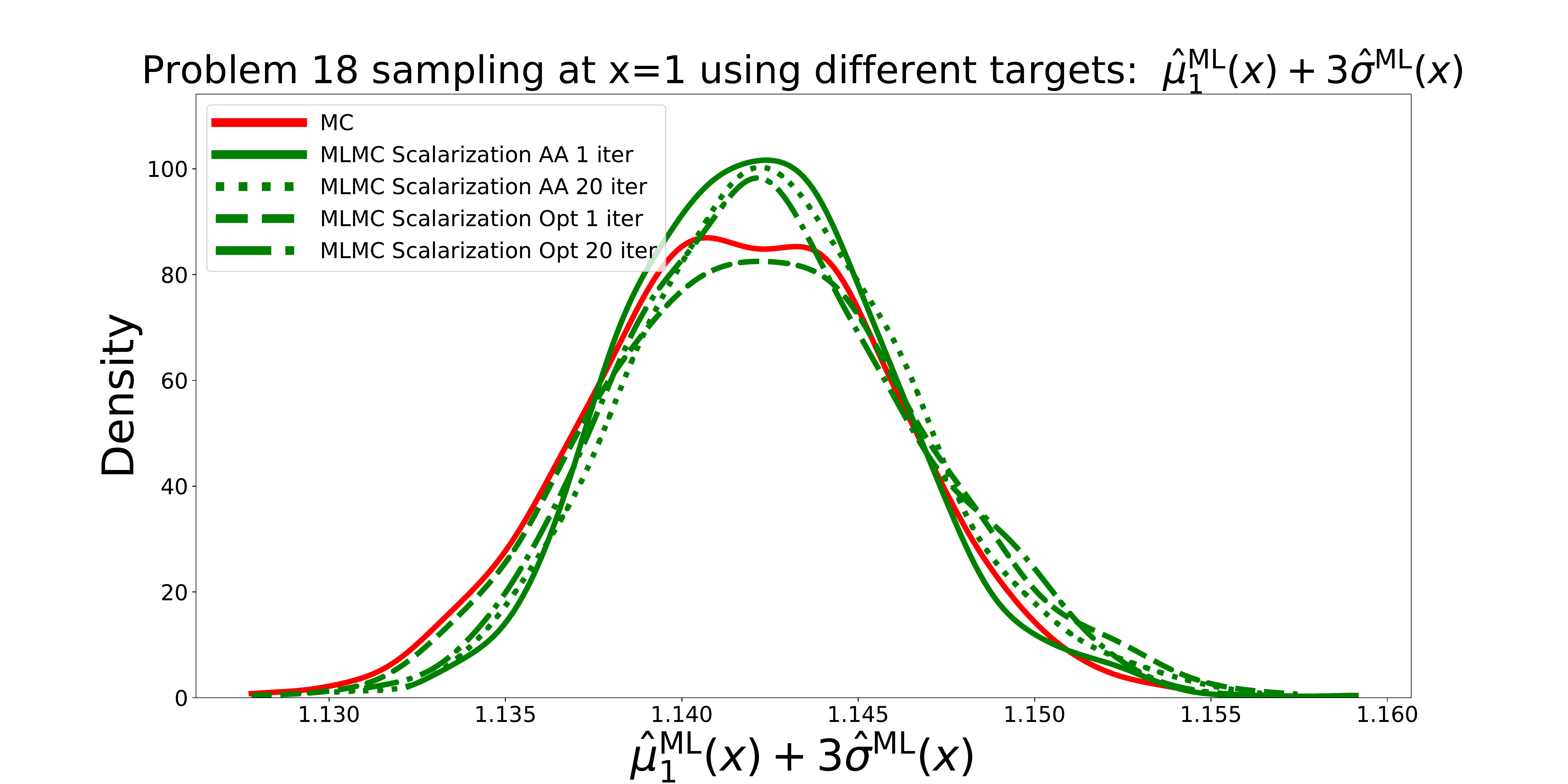}
\includegraphics[width=0.49\textwidth]{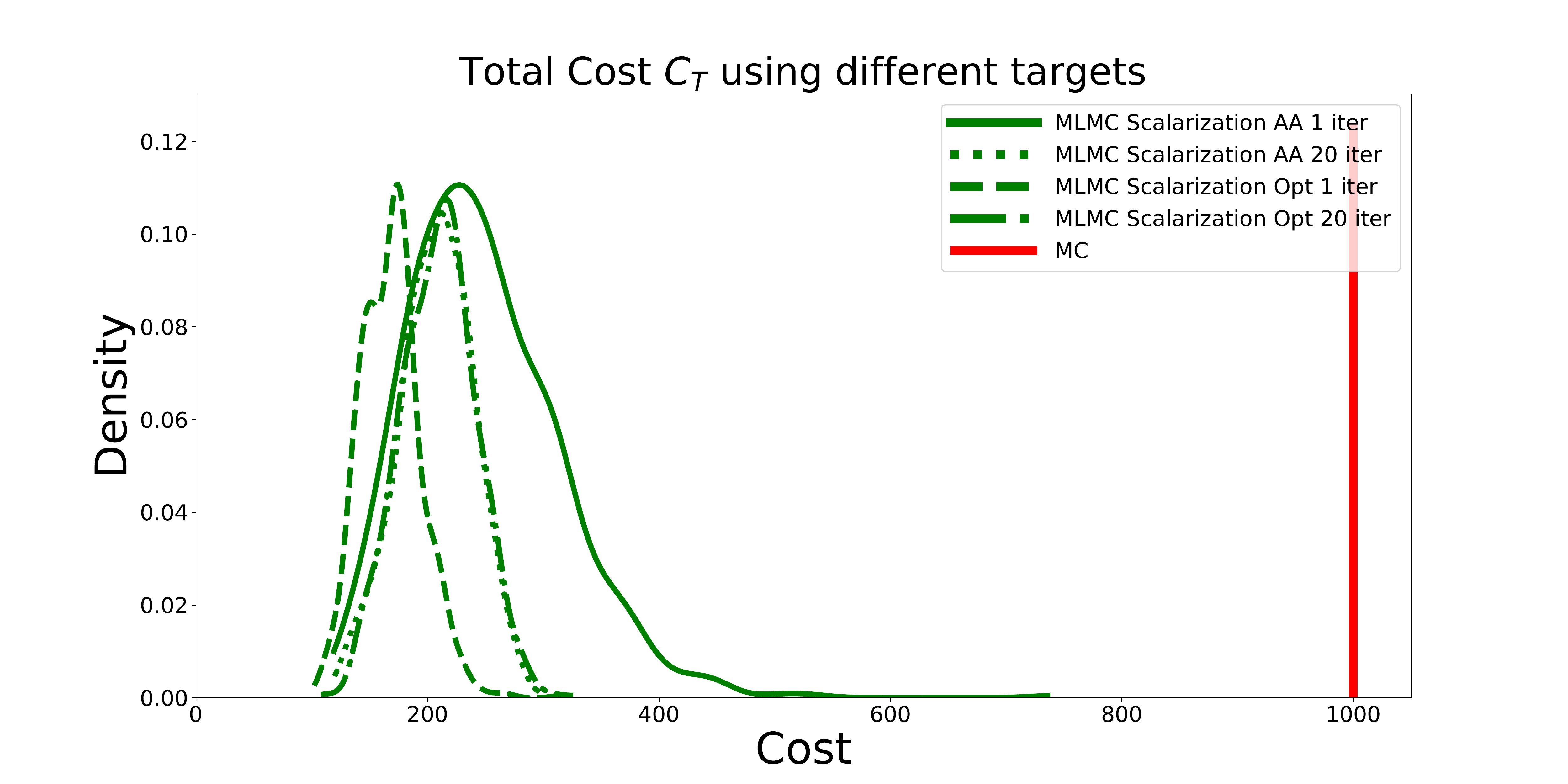}
\caption{\textbf{Scalarization (Bootstrap). }Left: Histogram over $1000$ samples of $\moneapprox[x] + \alpha \sigmaapprox[x]$ for $x=1$ using the scalarization estimator described in Eq.~\eqref{eq:mlmcscalarization} in green in combination with using the Bootstrap approximation described in Section~\ref{sssec:covariancebootstrap} to approximate the covariance term in Eq.~\eqref{eq:mlmcvarianceofscalarization}. We compare to a Monte Carlo reference estimator in red. We compare using an iterative approach for finding the resource allocation (1 iter or 20 iter) and to use only the analytical approximation (AA) or combine it with numerical optimization (Opt). Right: Respective cost for the different estimators.}
\label{fig:4level_meansigma_x1_covbootstrap}
\end{figure}

Therefore, we examine the third case, estimating the covariance term using the relationship between the covariance of mean and variance as described in Section~\ref{sssec:covariancecorrlift}. Fig.~\ref{fig:4level_meansigma_x1_covcorrlift} depicts the results. Again, we see a good match of histograms with the reference solution, with similar computational costs as for the bootstrap approach regarding the samples, but lower computational costs compared to using Pearson's correlation. The hidden computational cost for evaluating the covariance term itself is much cheaper than repeatedly evaluating the bootstrap term. Overall, this strategy seems to be the most efficient since it combines the best features of the two approaches: the low computational overhead of Pearson with the good approximation quality of bootstrap.
\begin{figure}[h] 
\centering
\includegraphics[width=0.49\textwidth]{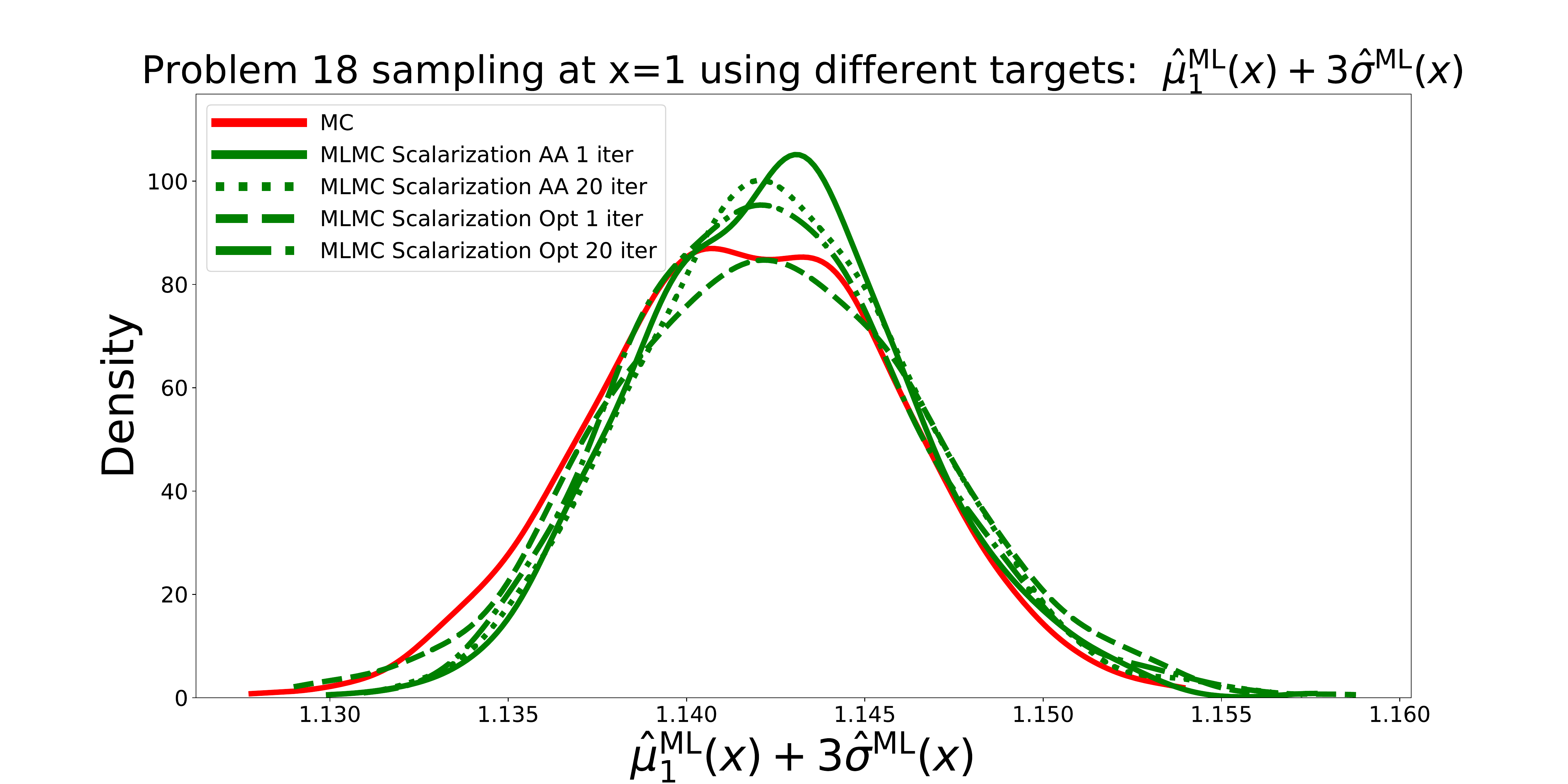}
\includegraphics[width=0.49\textwidth]{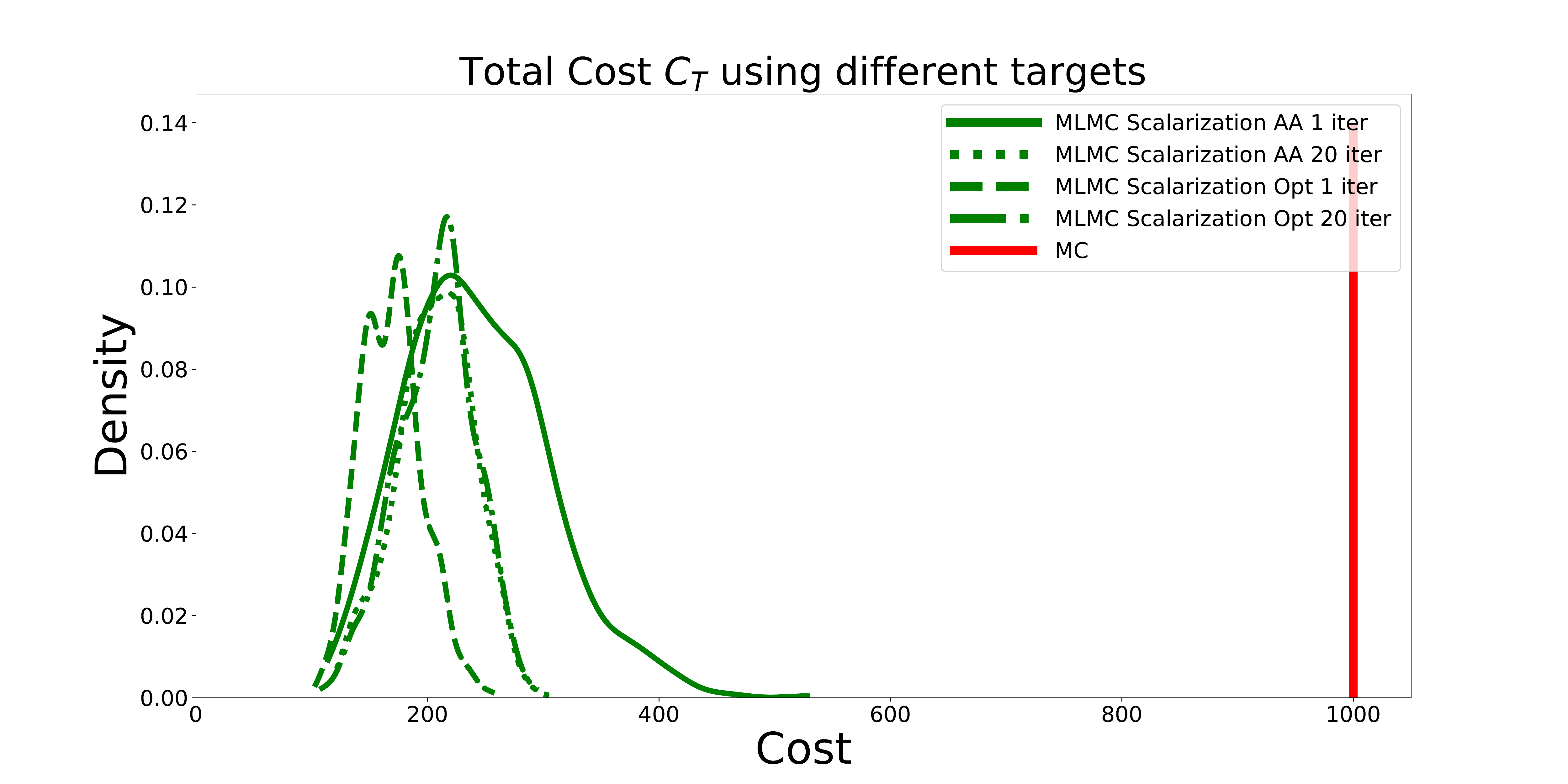}
\caption{\textbf{Scalarization Correlation Lift: }Left: Histogram over $1000$ samples of $\moneapprox[x] + \alpha \sigmaapprox[x]$ for $x=1$ using the scalarization estimator described in Eq.~\eqref{eq:mlmcscalarization} in green in combination with using the correlation lift approximation described in Section~\ref{sssec:covariancecorrlift} to bound the covariance term in Eq.~\eqref{eq:mlmcvarianceofscalarization}. We compare to a Monte Carlo reference estimator in red. We compare using an iterative approach for finding the resource allocation (1 iter or 20 iter) and to use only the analytical approximation (AA) or combine it with numerical optimization (Opt). Right: Respective cost for the different estimators.}
\label{fig:4level_meansigma_x1_covcorrlift}
\end{figure}

Finally, we also add the MLMC estimator targeting the mean in the resource allocation for the evaluation of the scalarization to the plot in Fig.~\ref{fig:4level_meansigma_x1_covcorrlift_withmean}. For comparison, we use again $\epsilon_{\mathbb{E}}^2 \approx 2.2321\text{e-}6$ as in the first case. Here, we use the correlation lift for approximating the scalarization and only show results for 20 iterations plus numerical optimization to clarify the presentation. We directly see the difference that the MLMC estimator targeting the mean greatly under-resolves, which results in a much larger variance at a lower computational cost. It does not devote nearly enough resources to achieve the accuracy we target given the Monte Carlo reference.
\begin{figure}[h] 
\centering
\includegraphics[width=0.49\textwidth]{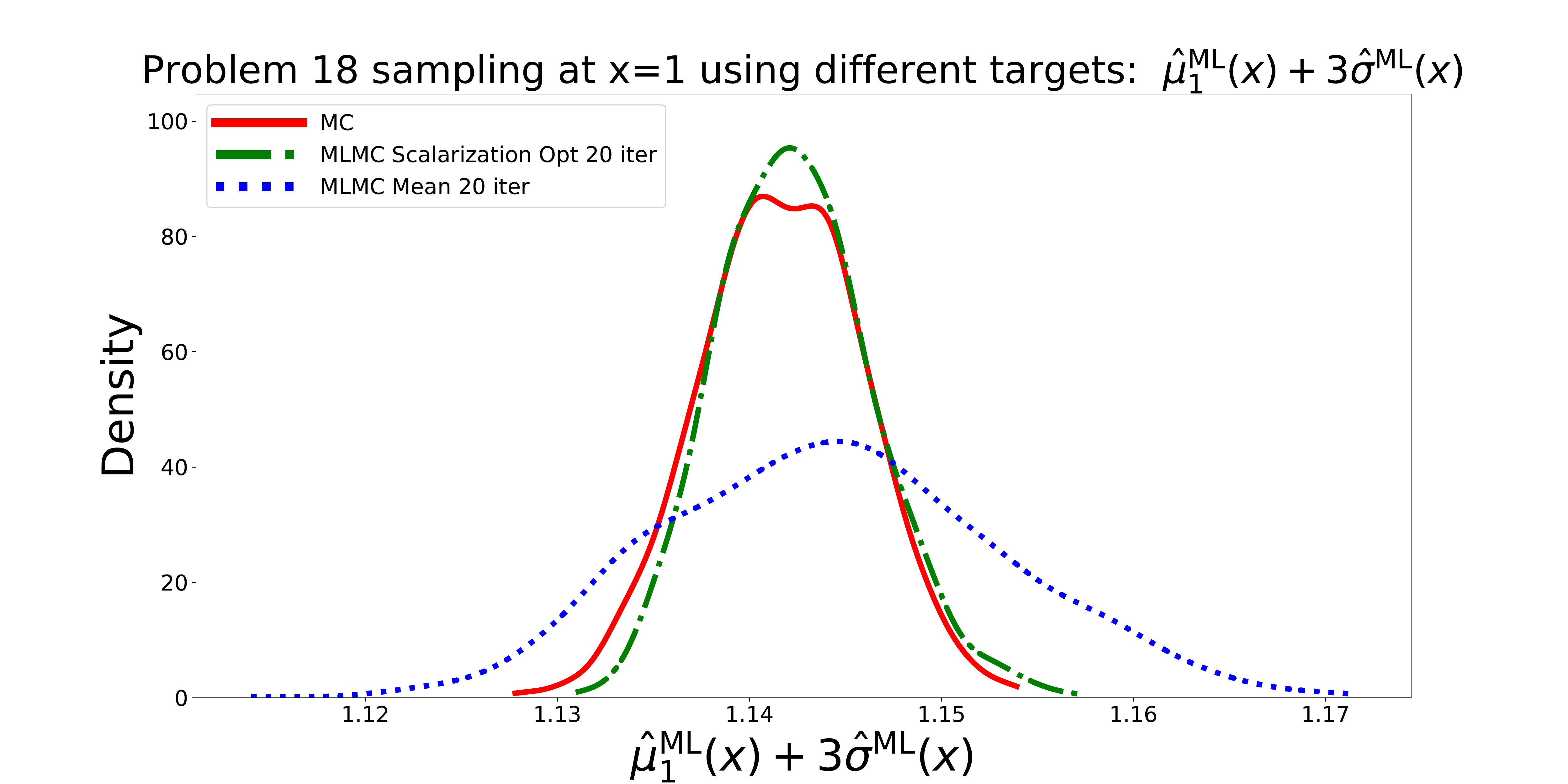}
\includegraphics[width=0.49\textwidth]{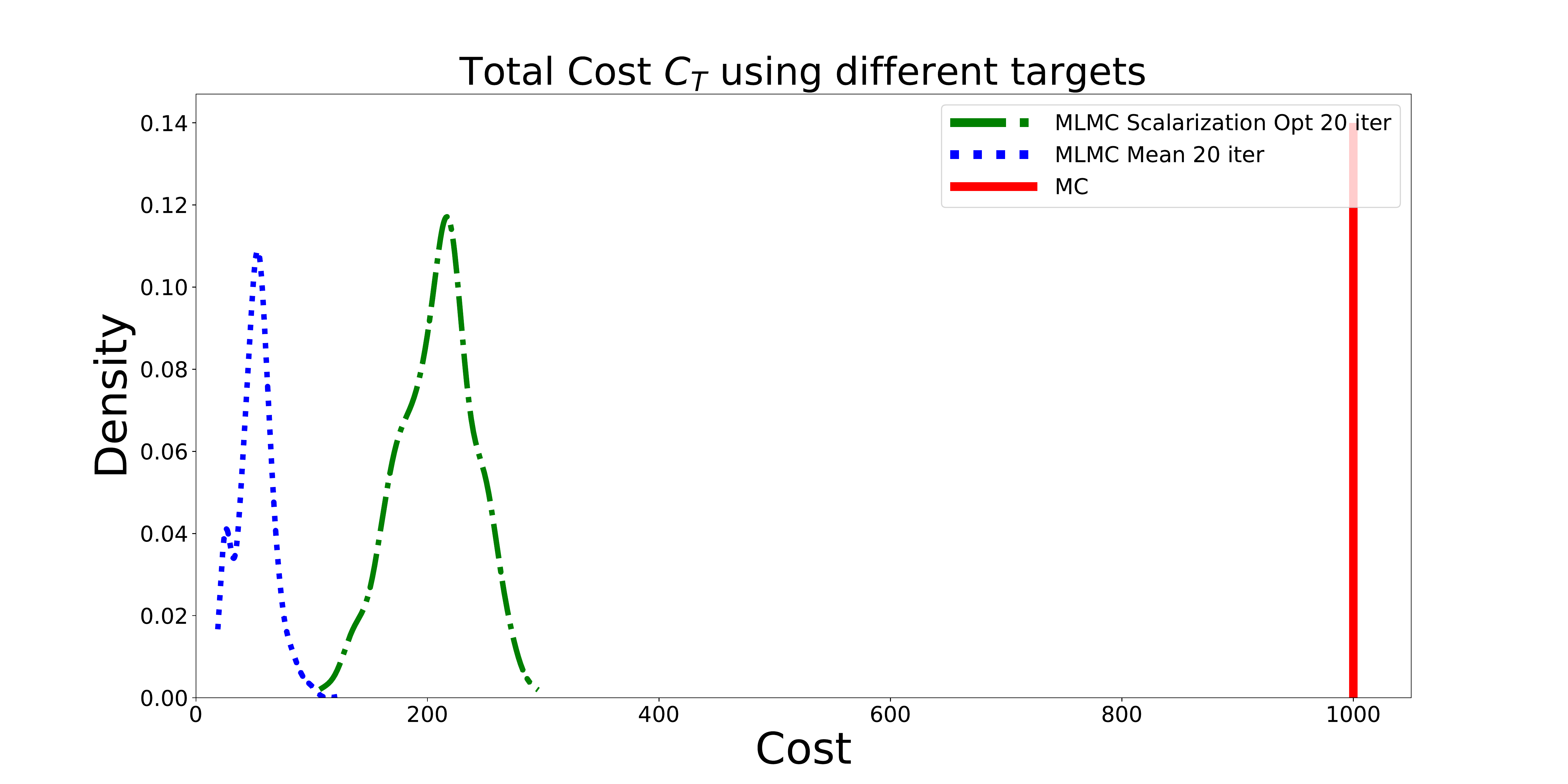}
\caption{\textbf{Scalarization (Correlation Lift including mean). }Left: Histogram over $1000$ samples of $\moneapprox[x] + \alpha \sigmaapprox[x]$ for $x=1$ using the different estimators described in Eq.~\eqref{eq:mlmcmean} in blue and Eq.~\eqref{eq:mlmcscalarization} in green compared to a reference Monte Carlo estimator in red. The covariance term of Eq.~\eqref{eq:mlmcvarianceofscalarization} is approximated using correlation lift approximation described in Section~\ref{sssec:covariancecorrlift}. We use numerical optimization and 20 iterations to find the resource allocation. Right: Respective cost for the different estimators.}
\label{fig:4level_meansigma_x1_covcorrlift_withmean}
\end{figure}

We show a quantitative comparisons for the scalarization case in Tables~\ref{tbl:4level_meansigma_x1_means} and~\ref{tbl:4level_meansigma_x1_variances}. We directly compare the three different approaches for the approximation of the covariance and their algorithmic implementations. First of all, we note again that all approaches improve the result compared to using the standard MLMC estimator which targets the mean. This reinforces the qualitative results in Fig.~\ref{fig:4level_meansigma_x1_covcorrlift_withmean}. Furthermore, while the approximation quality of the newly developed MLMC estimators targeting the scalarization are very similar, we note minor differences in the variance of the estimator. As in the previous results, using only a single iteration seem to match the target variance quite well when using the Pearson approximation, while using 20 iterations over-resolves which results in a smaller variance but larger computational cost. For both Bootstrap and Correlation Lift, we see that using 20 iterations improves the variance approximation. Given these results, the Correlation Lift approximation performs best and matches the target variance the closest (also taking into account computational cost) where we see a small improvement when using the numerical optimization compared to the analytic approximation.

\begin{table}[h]
\begin{center}
\begin{tabular}{| c || c | c | c | c |}
 \hline
& \multicolumn{4}{ c |}{Mean} \\ \hline
& Pearson & Bootstrap & Correlation Lift & Exact \\ \hline
MC  & \multicolumn{3}{ c | }{1.1417} & \multirow{6}{*}{1.1417} \\  \cline{2-4}
MLMC Mean (20 iter)  & 1.1443  & 1.1443 & 1.1443 & \\ 
MLMC Scalarization AA (1 iter)  & 1.1422 & 1.1423 & 1.1425 & \\ 
MLMC Scalarization AA (20 iter)  & 1.1424 & 1.1426 & 1.1426 &\\
MLMC Scalarization Opt (1 iter)  & 1.1423 & 1.1423  & 1.1423 & \\
MLMC Scalarization Opt (20 iter)  & 1.1423 &  1.1424 & 1.1424 & \\ \hline
\end{tabular}
\end{center}
\caption{Expectations from the histograms of the different approaches in Fig.~\ref{fig:4level_meansigma_x1_covpearson},~\ref{fig:4level_meansigma_x1_covbootstrap},~\ref{fig:4level_meansigma_x1_covcorrlift} and Fig~\ref{fig:4level_meansigma_x1_covcorrlift_withmean}. The column labeled \textit{Exact} shows the target value for the expectation.}
\label{tbl:4level_meansigma_x1_means}
\end{table}
\begin{table}[h]
\begin{center}
\begin{tabular}{| c || c | c | c | c |}
 \hline
& \multicolumn{4}{ c |}{Variance} \\ \hline
& Pearson & Bootstrap & Correlation Lift & Exact \\ \hline
MC  & \multicolumn{3}{ c |}{1.7739e-5} &  \multirow{6}{*}{1.6175e-5} \\ \cline{2-4}
MLMC Mean (20 iter)  & 1.1259e-5  & 7.7175e-5 & 7.7175e-5 & \\ 
MLMC Scalarization AA (1 iter)  &  1.6047e-5  & 1.3927e-5   & 1.4802e-5 & \\ 
MLMC Scalarization AA (20 iter)  &  1.1947e-5  &  1.5989e-5  & 1.5708e-5 & \\
MLMC Scalarization Opt (1 iter)   & 1.6534e-5 & 2.0454e-5   & 2.0591e-5 & \\
MLMC Scalarization Opt (20 iter)  & 1.1998e-5  & 1.6989e-5   & 1.6451e-5 &  \\ \hline
\end{tabular}
\end{center}
\caption{Variances from the histograms of the different approaches in Fig.~\ref{fig:4level_meansigma_x1_covpearson},~\ref{fig:4level_meansigma_x1_covbootstrap},~\ref{fig:4level_meansigma_x1_covcorrlift} and Fig~\ref{fig:4level_meansigma_x1_covcorrlift_withmean}. The column labeled \textit{Exact} shows the target value for the variance.}
\label{tbl:4level_meansigma_x1_variances}
\end{table}

We end this results section for sampling by looking at the samples' allocation profiles over the levels. We show the average resource allocation, i.e. the number of samples, for each level in Table~\ref{tbl:4level_sample_allocation} where we normalize with respect to the number of samples used at the finest level. The table shows the allocation for the different statistics (variance, standard deviation and scalarization) in which we compare the standard MLMC approach targeting the mean to our presented MLMC approaches targeting the respective statistics. Our major observation using this table is that we cannot simply scale the resource allocation of the standard MLMC targeting the mean for other statistics, but rather we see a redistribution of samples. Mainly, we see an increase in samples on the second level, while samples on the first and third level decrease. Hence, adapting the MLMC target to the statistic is critical and not avoidable by just scaling the standard MLMC estimator for the mean to vary its precision.
\begin{table}[h]
\begin{center}
\begin{tabular}{| c | c || c | c | c | c |}
 \hline
Statistic & Estimator Target & Level 1 & Level 2 & Level 3 & Level 4 \\ \hline
\multirow{2}{*}{Variance} & MLMC  Mean & 445.47 &  25.50 &  3.37 &   1\\ 
& MLMC Variance & 386.51 & 42.82 & 3.32 & 1\\ \hline
\multirow{2}{*}{Sigma} & MLMC  Mean & 452.31 & 25.92 & 3.43 & 1\\ 
& MLMC Sigma & 404.12 & 43.93 & 3.38 & 1\\ \hline
\multirow{2}{*}{Scalarization} & MLMC  Mean & 462.86 & 26.81 & 3.58 & 1\\ 
& MLMC Scalarization & 359.52 & 32.42 & 3.27 & 1\\ \hline
\end{tabular}
\end{center}
\caption{Averaged and normalized sample profiles for different statistics. For each sample profile we compare with respect to the standard MLMC estimator targeting the mean. We used 20 iterations for all approaches and numerical optimization for the MLMC estimators targeting variance, standard deviation and scalarization. It is possible to observe that for all cases a lower number of samples is used at the coarsest level, while a larger number of samples is used for the second level.}
\label{tbl:4level_sample_allocation}
\end{table}

With these results, we showed that using the correct MLMC estimators targeted at the OUU goal of interest is crucial to getting the accuracy expected at the lowest computational cost. We clearly note that using the MLMC estimator for the mean is not sufficient if we want to reach a certain target. The approaches that synchronize the OUU goal and the allocation target, on the other hand, are able to achieve the desired accuracy. 
Regarding the algorithmic choices, we consistently see the best results when using an iterative approach coupled with numerical optimization for finding the resource allocation. We also showed that the choice for the approximation of the covariance in the scalarization is crucial for performance. While the Pearson correlation might be a convenient and simple solution, it over-resolves the estimator, which results in unnecessary computational cost. While bootstrapping shows good results, it adds an additional computational cost for each evaluation due to the resampling. We see a good balance when using the correlation lift approach of Section~\ref{sssec:covariancecorrlift}. 

In the next section, we use these results when we move to optimization under uncertainty, where we find the optimal resource allocation in each optimization step. Based on the previous results, we restrict our algorithmic options to clarify the presentation. We restrict ourselves to using 20 iterations coupled with numerical optimization. We still compare the three different approaches for approximating the covariance term for the scalarization.

\subsubsection{Optimization under uncertainty}
We are interested in solving the following two optimizations problem. We consider $\Rmu{f}$ as a first test case and get
\begin{equation}
\label{eq:problem18ouu_mean}
\begin{split}
\min_x \,&\mathbb{E}[f_{4}(x, \xi)], \\
\text{s.t. } &f_{det}(x) \geq g(x).
\end{split}
\end{equation}
The second case considers the scalarization of mean and standard deviation, $\Rscal{f}$, to show our new developments 
\begin{equation}
\label{eq:problem18ouu_scalarization}
\begin{split}
\min_x \,&\mathbb{E}[f_{4}(x, \xi)] + 3 \sigma[f_{4}(x, \xi)], \\
\text{s.t. } &f_{det}(x) \geq g(x),
\end{split}
\end{equation}
where we pick $\alpha = 3$, a very common choice in the robust optimization field.

To get an intuition about the function, we visualize it in Fig.~\ref{fig:problem18ouu} for both test cases. We see the effect of the scalarization, which pushes the function up by its standard deviation. For both test cases, we proceed as follows for our numerical test. For the optimization runs, we begin at the initial position $x=0.25$. Next, we run 25 independent runs for each case using the different estimators presented in this work. In the resulting plots, we plot the 25 final designs found for each optimization run. We additionally plot the average computational cost for a single iteration by keeping track of the computational cost over the optimization process. The MLMC estimators are using the four levels as described in Eq.~\eqref{eq:problem18_4levels}. By comparing to a Monte Carlo reference solution that uses 1000 samples on the finest level $f_4$, the corresponding $\epsilon^2$ is obtained. Note here that we do not use a convergence criterion but stop the optimization after 100 iterations. Therefore, as in the sampling case, we expect the MLMC estimators to perform similarly to the Monte Carlo estimator since it is targeting the same accuracy in $\epsilon^2$. We are not too concerned about the final design since the optimization problem itself is very simple.
\begin{figure}
\centering
\includegraphics[width=0.5\textwidth]{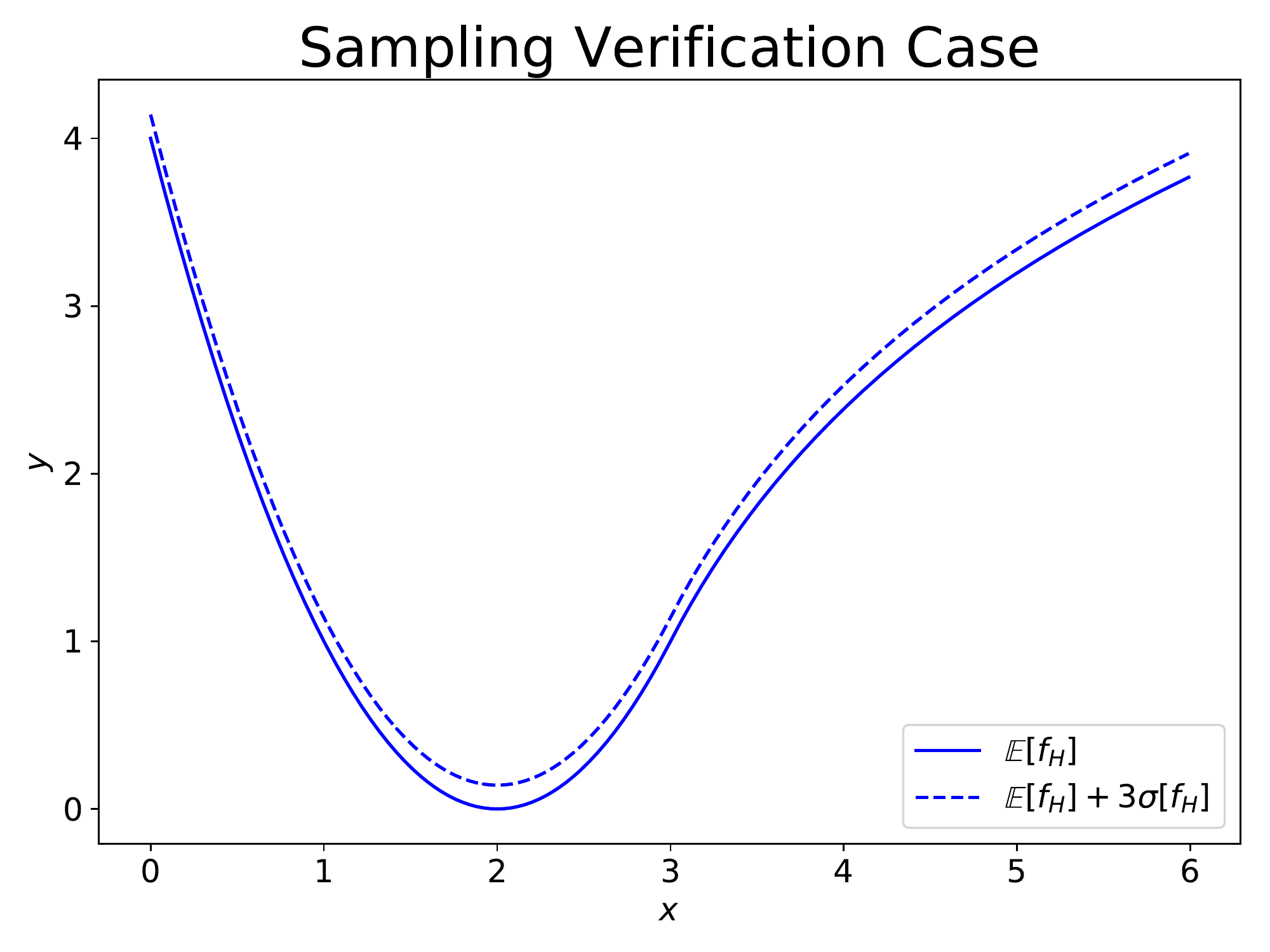}
\caption{Functions for expected value $\mathbb{E}[f_H]$ (solid) and scalarization $\mathbb{E}[f_H] + 3 \sigma[f_H]$ (dashed) of problem 18.}
\label{fig:problem18ouu}
\end{figure}

In the first case, we consider the expectation where we solve the optimization problem as given in Eq.~\eqref{eq:problem18ouu_mean} targeting $\epsilon_{\mathbb{E}}^2 \approx 2.2321\text{e-}6$. As we can see, we have a very nice match between the final designs found by the MLMC mean in blue crosses and the MC designs in red dots. If we look at the average computational cost for a single iteration, we see the big advantage of using MLMC methods again. In this case, we are able to reduce the computational cost by about a factor of 20 for this specific choice of computational cost. 
\begin{figure}
\centering
\includegraphics[width=0.49\textwidth]{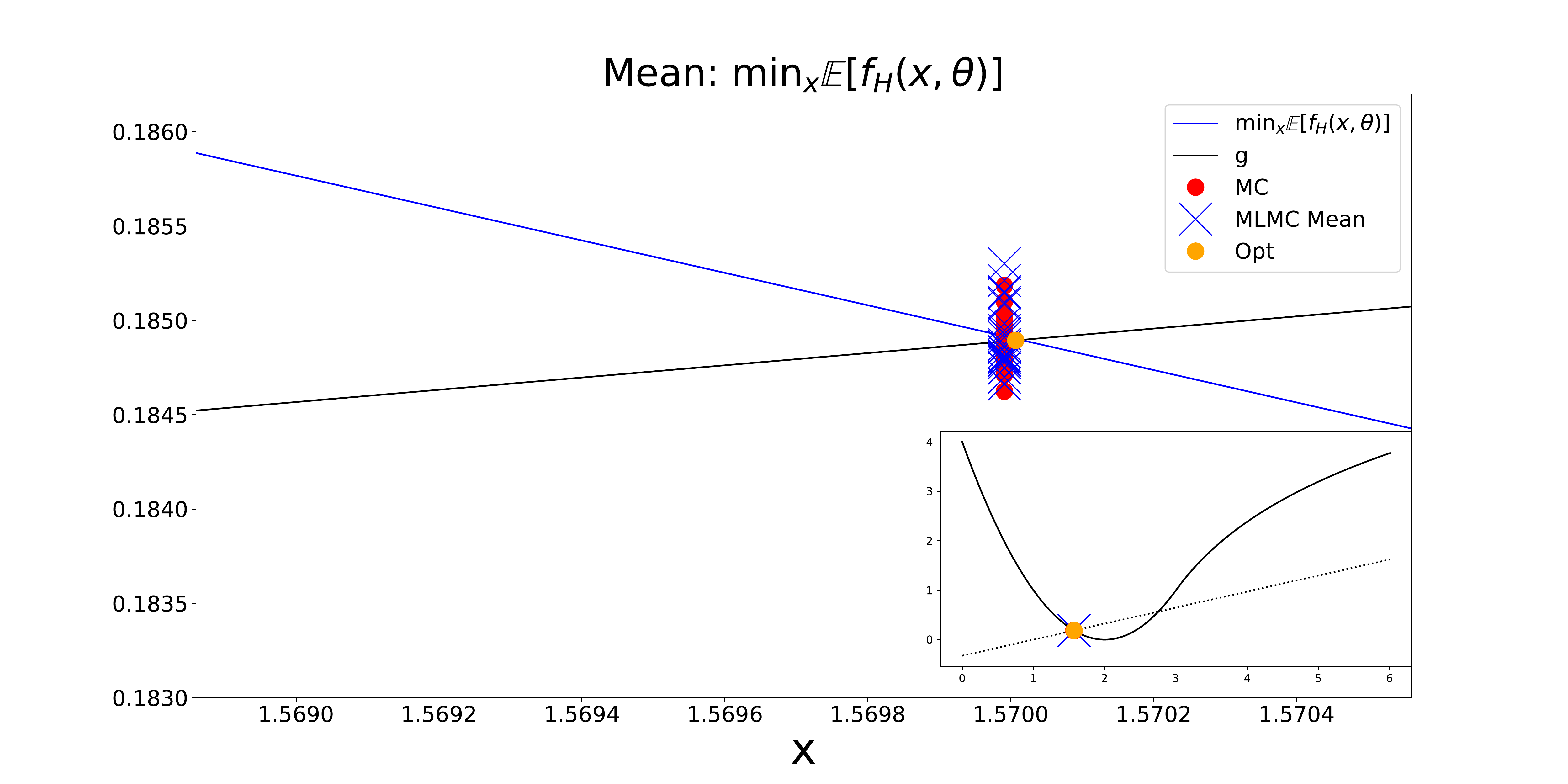}
\includegraphics[width=0.49\textwidth]{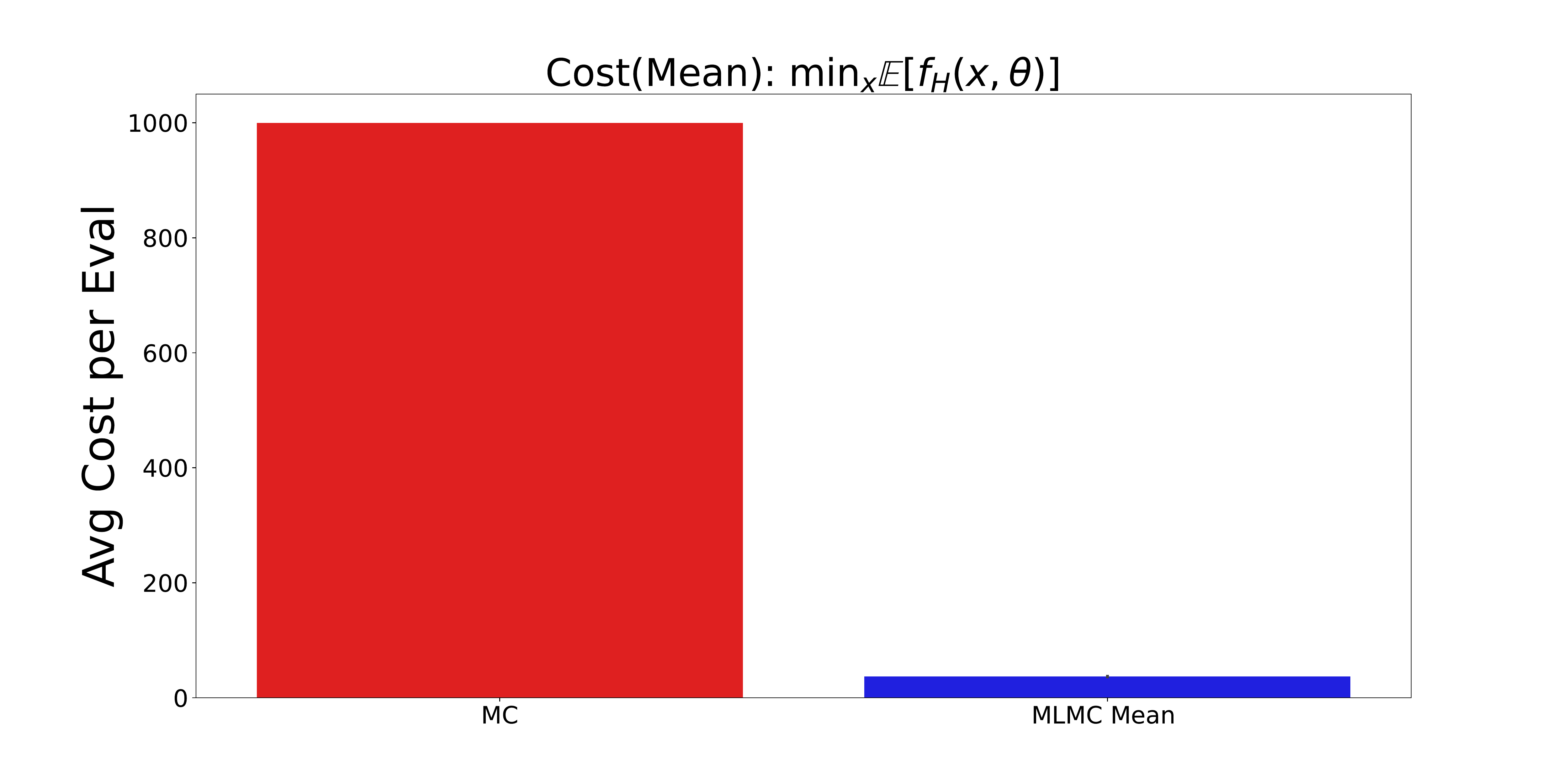}
\caption{\textbf{Mean.} Optimization results for 25 individual runs after 100 iterations. The blue line shows the objective function while the black line shows the constraint. The small figure in the bottom right shows the full function while we enlarge the area around the optimal design. Each marker corresponds to the final design found by the individual run. We show the results when using a standard Monte Carlo estimator using $1000$ samples using red dots and compare to the final design found using a MLMC estimator targeting the mean as blue x. The yellow dot shows the optimal design.}
\label{fig:4level_ouu_mean}
\end{figure}

We see the effect of the choice of the MLMC estimator when moving to the scalarization case as given in Eq.~\eqref{eq:problem18ouu_scalarization}, where we use $\epsilon_{\mathbb{E}}^2 \approx 2.2321\text{e-}6$ and $\epsilon_{\mathbb{\mathbb{E} + \alpha \sigma}}^2 \approx 1.6175\text{e-}5$ for the respective approach, in Fig.~\ref{fig:4level_ouu_meansigma}. The results for the various options to approximate the covariance are shown in the three plots, using the Pearson upper bound in Fig.~\ref{fig:4level_ouu_meansigma_Pearson},  the correlation lift in Fig.~\ref{fig:4level_ouu_meansigma_CorrLift} or the Bootstrap approximation in Fig.~\ref{fig:4level_ouu_meansigma_Bootstrap}. Similarly to the sampling study, we see a much larger variance in the optimal designs found by the MLMC approach targeting the mean, shown as blue crosses. As an effect, we see a bias in the distribution of the points targeting the mean, which comes from the larger noise introduced by these samples and taken into account by the optimization process of SNOWPAC. However, we see a good match between our newly developed scalarization estimators and the Monte Carlo reference solution in red. This is also reflected in the average computational cost for a single evaluation in Fig.~\ref{fig:4level_ouu_meansigma_Cost}. For the MLMC approach targeting the mean (blue bar), we notice a very small average evaluation cost. The estimator is underresolved by not using enough samples, which leads to a larger estimator variance. For the three MLMC strategies targeting the scalarization (green bars), we also see a cost reduction compared to the reference Monte Carlo, though less than for MLMC targeting the mean. This, however, matches the variance of the reference solution. Comparing the costs of the three strategies of approximating the covariance, we again see a similar picture: we get the highest cost for the Pearson upper bound since it is indeed an upper bound and conservative approach; we see slightly less cost for the correlation lift approximation and the bootstrap approximation, while we have to take into account additional computational cost for the bootstrap approach. Hence, we prefer the correlation lift as the most cost efficient strategy.

\begin{figure}
\centering
\begin{subfigure}{0.49\textwidth}
	\includegraphics[width=\textwidth]{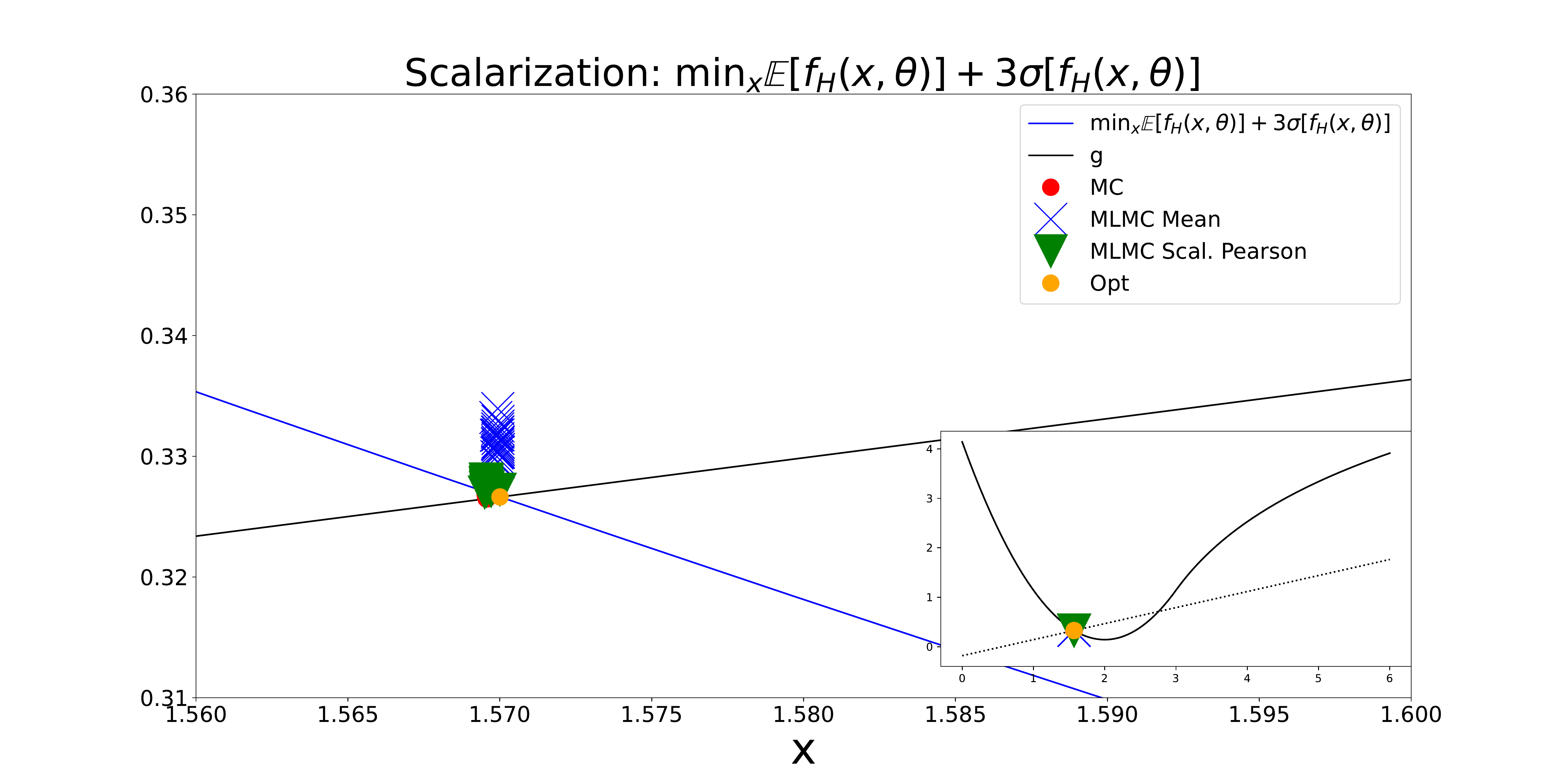}
    \caption{We use the Pearson upper bound as described in Section~\ref{sssec:covarianceupperbound} to bound the covariance term of Eq.~\eqref{eq:mlmcvarianceofscalarization}.}
    \label{fig:4level_ouu_meansigma_Pearson}
\end{subfigure}
\hfill
\begin{subfigure}{0.49\textwidth}
	\includegraphics[width=\textwidth]{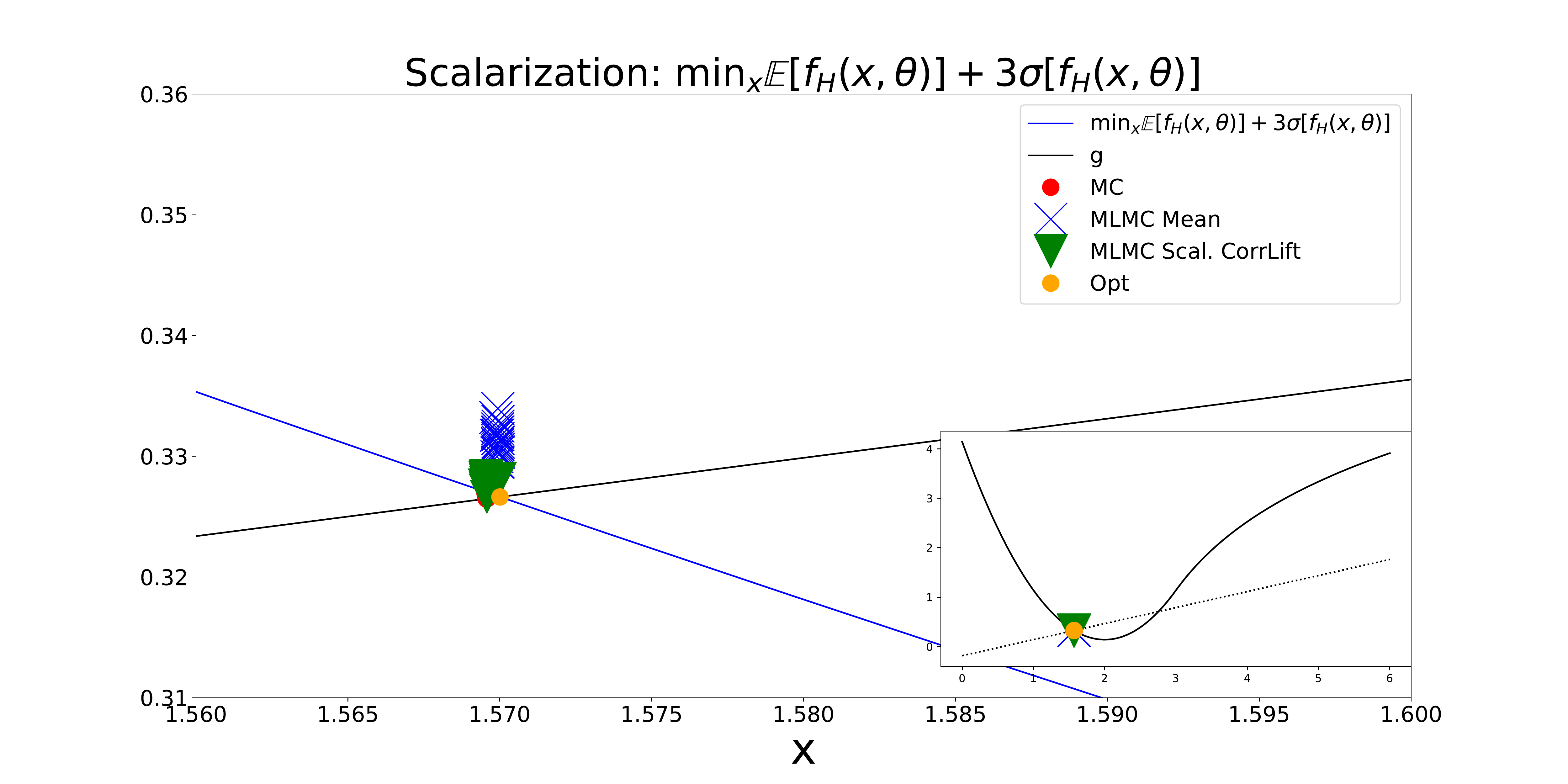}
    \caption{We use the Correlation lift approximation as described in Section~\ref{sssec:covariancecorrlift} to approximate the covariance term of Eq.~\eqref{eq:mlmcvarianceofscalarization}.}
    \label{fig:4level_ouu_meansigma_CorrLift}
\end{subfigure}
\hfill
\begin{subfigure}{0.49\textwidth}
\includegraphics[width=\textwidth]{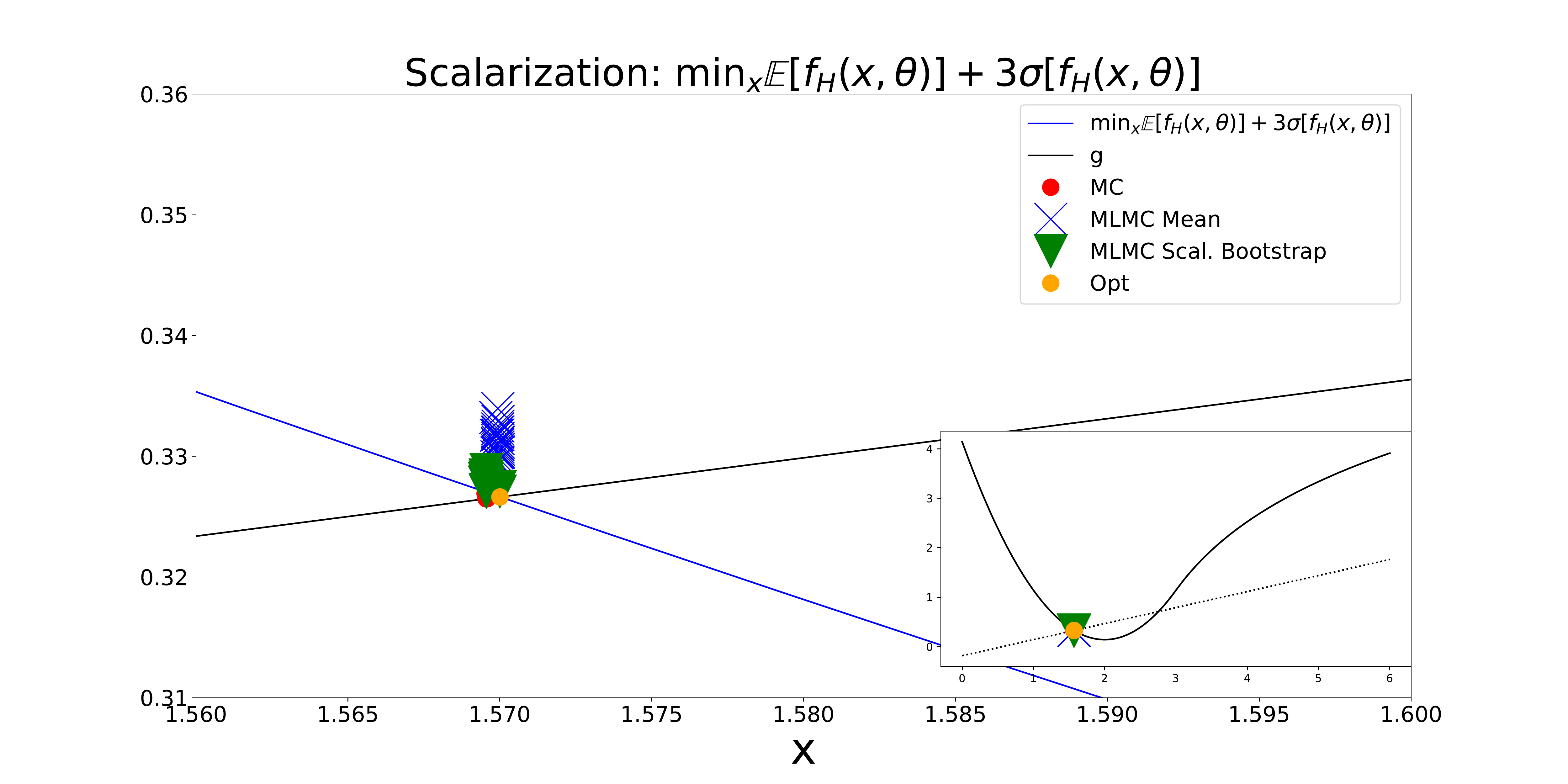}
    \caption{We use the Bootstrap approximation as described in Section~\ref{sssec:covariancebootstrap} to approximate the covariance term of Eq.~\eqref{eq:mlmcvarianceofscalarization}.}
    \label{fig:4level_ouu_meansigma_Bootstrap}
\end{subfigure}
\begin{subfigure}{0.49\textwidth}
	\includegraphics[width=\textwidth]{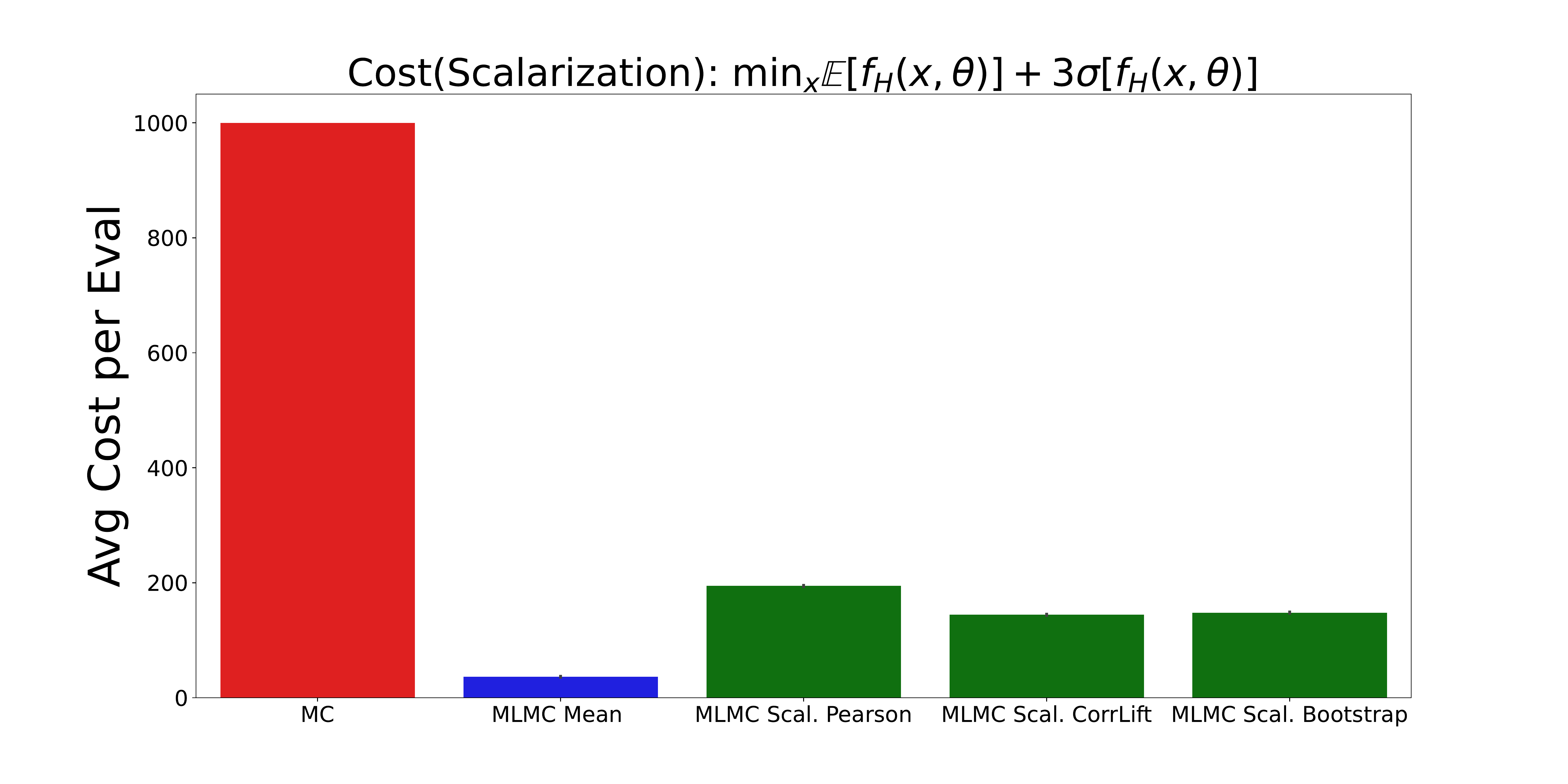}
    \caption{Cost average for a single evaluation for the different approaches. The cost are averaged over 25 optimization runs with 100 iterations each.}
    \label{fig:4level_ouu_meansigma_Cost}
\end{subfigure}
\caption{\textbf{Scalarization covariance approximation and cost.} Optimization results for 25 individual runs after 100 iterations. In all figures, The blue line shows the objective function while the black line shows the constraint. The small figure in the top left shows the full function while we enlarge the area around the optimal design. Each marker corresponds to the final design found by the individual run. We show the results when using a standard Monte Carlo estimator using $1000$ samples as red dot and compare to the final design found using a MLMC estimator targeting the mean as blue x and a MLMC estimator targeting the scalarization as green triangle. The yellow dot shows the optimal design. In the first three figures, we compare the different approximation strategies for the covariance term and the fourth figure shows a cost comparison for all three approaches.}
\label{fig:4level_ouu_meansigma}
\end{figure}

Finally, we look at quantitative metrics to measure the approximation quality of the new approaches in Table~\ref{tbl:problem18_results}. We use the metrics as described in Section~\ref{ssec:distancemetrics} to measure the distance to the final designs using the reference Monte Carlo solution. We clearly see that our new MLMC approach targeting scalarization is consistently closer to the reference designs compared to the standard MLMC approach targeting the mean. We mark the smallest value in each column in bold. In this case, the Pearson upper bound seems like a good conservative choice.
\begin{table}[h]
\centering
\begin{tabular}{c | c | c | c | c }
Method & $\mathbb{X}$ & $\text{Dist}^{\mathbb{X}}_{\text{C}}$  & $\text{Dist}^{\mathbb{X}}_\sigma$ & $\text{Dist}^{\mathbb{X}}_{\text{RMSdev}}$ \\ \hline
MLMC Mean (20 iter) & $\mathbb{E}$ & 4.1683e-3 & 7.2126e-4 & 4.3140e-3\\ \hline
MLMC Scalarization (20 iter, Pearson) & $\mathbb{\mathbb{E} + \alpha \sigma}$ & \textbf{5.6635e-4} & 1.1207e-4 & \textbf{5.4988e-4}\\ \hline
MLMC Scalarization (20 iter, CorrLift) & $\mathbb{\mathbb{E} + \alpha \sigma}$ & 7.0536e-4  & 5.9096e-5 & 7.4629e-4\\ \hline
MLMC Scalarization (20 iter, Bootstrap) & $\mathbb{\mathbb{E} + \alpha \sigma}$ & 8.0911e-4  & \textbf{1.8353e-5} & 8.5709e-4\\ \hline
\end{tabular}
\caption{Quantitative comparison of the distance of final designs found to the Monte Carlo reference solution. Each row shows a different approach. Each column represent a different metric, with the second column showing the target of the estimator.}
\label{tbl:problem18_results}
\end{table}

To summarize this section, we showed the effectiveness of the newly developed estimators for sampling and optimization under uncertainty. We should adapt the MLMC estimator to the given formulation of the sampling or optimization problem. We discussed the different algorithmic choices of adapting the resource allocation iteratively and using a numerical optimization to further improve the approximation quality. We saw the crucial choice of approximating the covariance for the scalarization case, which motivated our development of these different estimators. In the next section, we move to a more challenging optimization problem, the constrained Rosenbrock function, where we design a three-level test case.

\subsection{Rosenbrock}
\label{ssec:rosenbrock}
We employ the constrained 2-D Rosenbrock optimization problem as presented in~\cite{Rosenbrock} as our second case since it is a common benchmark in the optimization community since it challenges many solvers. In its deterministic form it is given by 
\begin{equation}
\begin{split}
\min_{x_1, x_2} \,&f(x_1, x_2) = 100 (x_2 - x_1^2)^2 + (x_1 - 1)^2, \\
\text{s.t. } &c_1: (x_1 - 1)^3 + 1 - x_2 \leq 0, \\
& c_2: (x_1 + x_2) - 2 \leq 0.
\end{split}
\end{equation}
We visualize the problem for $x_1 \in [-1.5, 1.5]$ and $x_2 \in [-0.5, 2.5]$ in Fig.~\ref{fig:Rosenbrock}. While the unconstrained problem has a single global minimum at $(1, 1)$, this constrained problem has a local minimum at $(0, 0)$, while it has a global minimum at $(1, 1)$. Due to small gradients, it is challenging for optimization algorithms to find the global minimum.
\begin{figure}[h]
\centering
\includegraphics[width=0.8\textwidth]{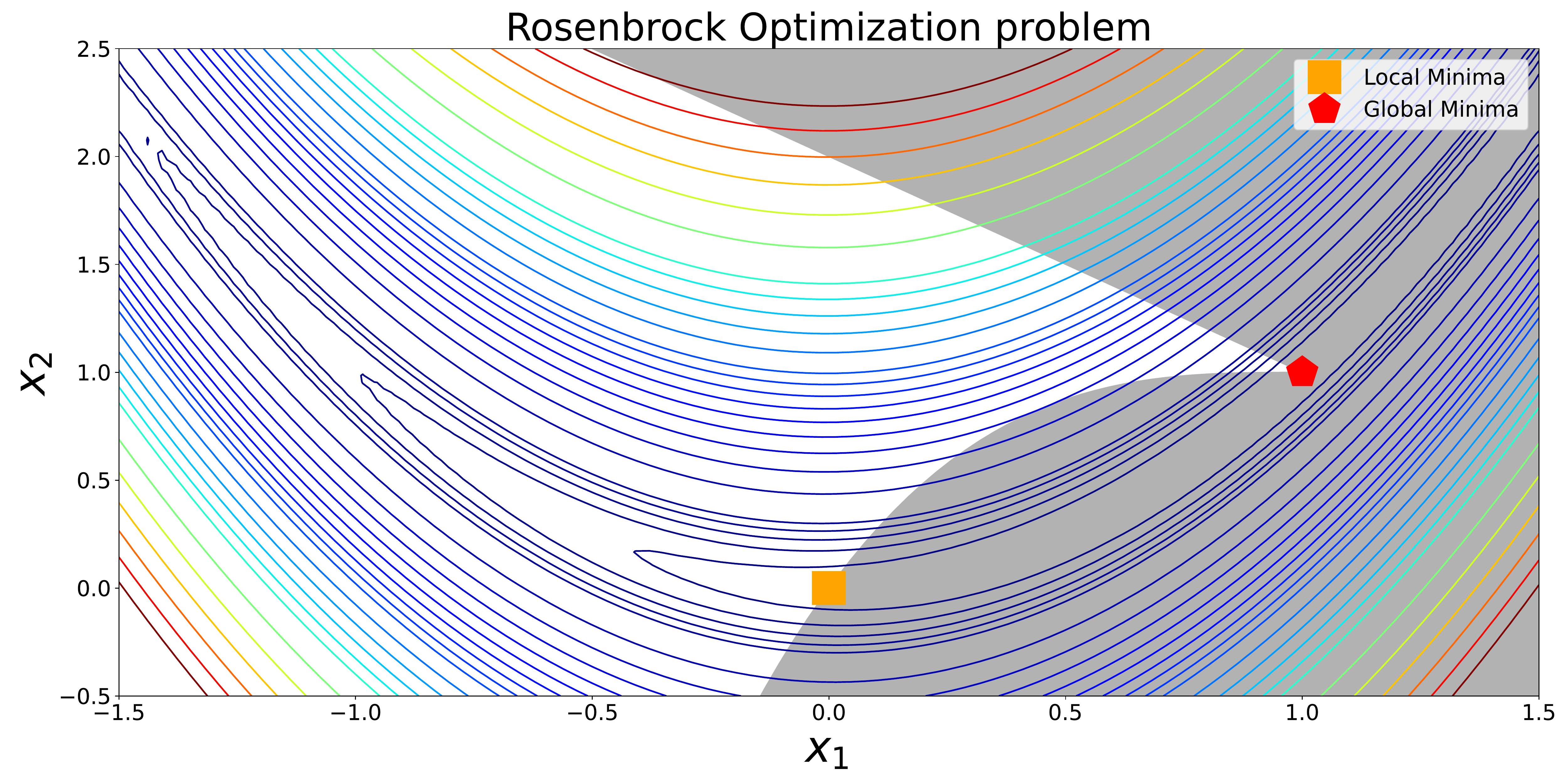}
\caption{Visualization of the optimization problem for the Rosenbrock function. The objective contour lines are plotted plus the two constraints in grey. The infeasible region is marked in grey. The local and global optima are visualized as orange square and red pentagon respectively.}
\label{fig:Rosenbrock}
\end{figure}

To obtain a stochastic problem with multiple levels, we create three levels using the Ishigami function which we adapt from \cite{Qian2018}. The three functions and their corresponding mean and sigma are given in Table~\ref{tbl:ishigamilevels} where $\{z_i\}_{i=1}^3 \sim \mathcal{U}(-\pi, \pi)$ follow a uniform distribution and $a=5$ and $b=0.1$. 
\begin{table}[h]
\centering
\begin{tabular}{c || c | c}
Level function & $\mone^{(i)}$ & $\sigma^{(i)}$ \\ \hline\hline
$I^{(3)}(z_1, z_2, z_3) = \sin(z_1) + a \sin(z_2)^2 + b z_3^4$ & 2.5 & 3.2931 \\ \hline
$I^{(2)}(z_1, z_2, z_3) =  \sin(z_1) + 0.85 a \sin(z_2)^2 + b z_3^4 \sin(z_1)$ & 2.125 & 3.1595 \\ \hline
$I^{(1)}(z_1, z_2, z_3) = \sin(z_1) + 0.6 a \sin(z_2)^2  + 9 b z_3^2 \sin(z_1)$ & 1.5 & 3.5308 \\ \hline
\end{tabular}
\caption{Three levels of the Ishigami function $\{I^{(i)}\}_{i=1}^3$ and their corresponding mean $\mone^{(i)}$ and standard deviation $\sigma^{(i)}$.}
\label{tbl:ishigamilevels}
\end{table}

We combine the Ishigami function with the Rosenbrock objective function $f(x_1, x_2)$ to obtain three levels $\{f(x_1, x_2) + \beta I^{(i)}(z_1, z_2, z_3)\}_{i=1}^3$. Additionally, we use a scaling factor $\beta = \sqrt{0.0001}$ to normalize the stochastic effect of the Ishigami function to the deterministic Rosenbrock function. We again assume a cost ratio between the different levels of $\frac{C_i}{C_{i-1}} = 10, i=2,3$, with $C_3=1$, such that $C_1 < C_2 < C_3$. 

In the end this results in an optimization problem similar to the previous test problem where we look at the reliability formulation $\Rscal{f}$ \footnote{We omit the formulation $\Rmu{f}$ in this case since it is not our contribution}:
\begin{equation}
\begin{split}
\min_{x_1, x_2} &\Rscal{f}[f + I^{(3)}] - \mone^{(3)} - 3 \sigma_3= \mathbb{E}[f + I^{(3)}] + 3 \sigma[f + I^{(3)}] - \mone^{(3)} - 3 \sigma^{(3)}, \\
\text{s.t. } &c_1: (x_1 - 1)^3 + 1 - x_2 \leq 0, \\
& c_2: (x_1 + x_2) - 2 \leq 0.
\end{split}
\end{equation}
Note here that we subtract the mean value $\mone^{(3)}$ and also the standard deviation $\sigma^{(3)}$ to make the solution comparable to the deterministic case. Hence, the local and global optima are the same as in the deterministic setting.

We proceed with a similar study as for the previous problem. Given the initial starting point $(x_1, x_2) = (0.25, 1.5)$ we run 25 independent optimization runs for each of the different approaches computing the MLMC estimator. We again compare to an optimization using single level Monte Carlo estimates computed with 1000 samples. The targets $\epsilon_{\mathbb{E}}^2 \approx 1.0849e-05 $ and $\epsilon_{\mathbb{\mathbb{E} + \alpha \sigma}}^2 \approx 8.8951e-05$ are therefore given by the reference variance of a Monte Carlo estimator. For a better comparison, we restrict all of the optimization runs to 250 iterations each. We again plot the optimization result, i.e., the final design found for all the different runs. Additionally, we compare the average computational cost for a single iteration. To reduce the number of results, we make the algorithmic choice to use 20 iterations and numerical optimization for finding the resource allocation. The choice is based on the previous results showing the best performance. We again compare the different covariance approximation strategies.

We show the optimization results of 25 independent runs in Fig.~\ref{fig:rosenbrockresults}. In Fig.~\ref{fig:rosenbrockresults_meansigma_Pearson}-~\ref{fig:rosenbrockresults_meansigma_Bootstrap}, we again compare the three different approaches in approximating the covariance function. The final optimal designs of the reference Monte Carlo approach are plotted as red circles. The optimal designs using the standard MLMC approach targeting the mean are shown as blue crosses. The optimal designs using our new contribution, the MLMC estimator targeting the scalarization, are shown as green triangles. Qualitatively, we see a close match between all our new approaches and the reference solution. Also the standard MLMC approach performs better compared to the results shown in the previous sections. Nevertheless, we already see qualitatively that the set of final designs shows a larger variance. When we look at the cost on the right, we see the reason. The standard MLMC estimator is again under-resolving the estimators, which results in a larger variance for the estimator and more noise in SNOWPAC. The magnitude of the noise is an important factor for the convergence of SNOWPAC, as we stated in Section~\ref{ssec:snowpac}. Regarding the cost for the covariance approximation, we see a similar image as in the previous example: the Pearson approximation is too conservative, which results in unnecessary computational cost. The correlation lift and bootstrap approximation result in similar, lower cost, although we neglect the additional computational cost of bootstrapping; hence, the correlation lift approximation seems to be the preferable choice.
\begin{figure}
\centering
\begin{subfigure}{0.49\textwidth}
	\includegraphics[width=\textwidth]{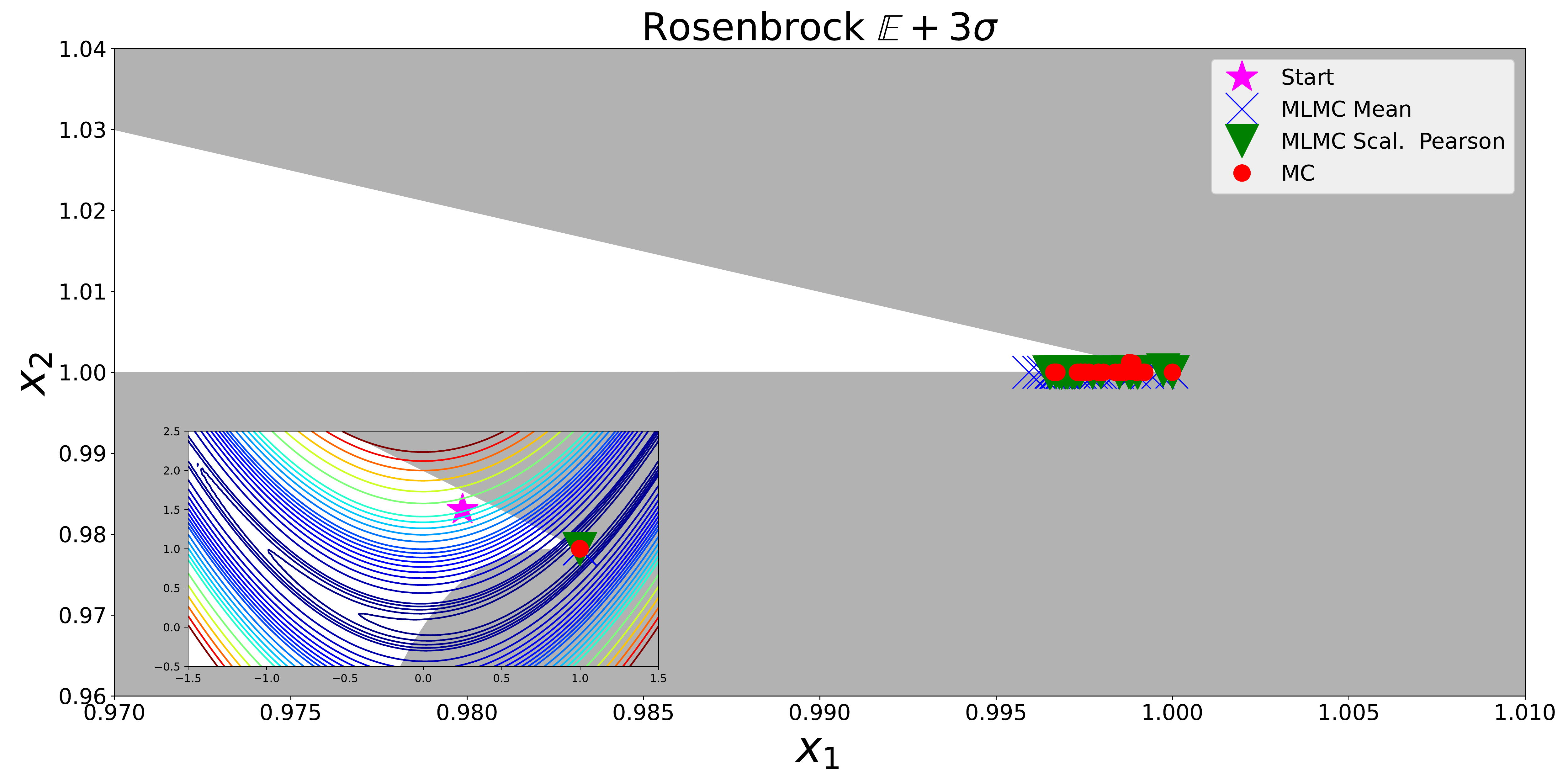}
    \caption{We use the Pearson upper bound as described in Section~\ref{sssec:covarianceupperbound} to bound the covariance term of Eq.~\eqref{eq:mlmcvarianceofscalarization}.}
    \label{fig:rosenbrockresults_meansigma_Pearson}
\end{subfigure}
\hfill
\begin{subfigure}{0.49\textwidth}
	\includegraphics[width=\textwidth]{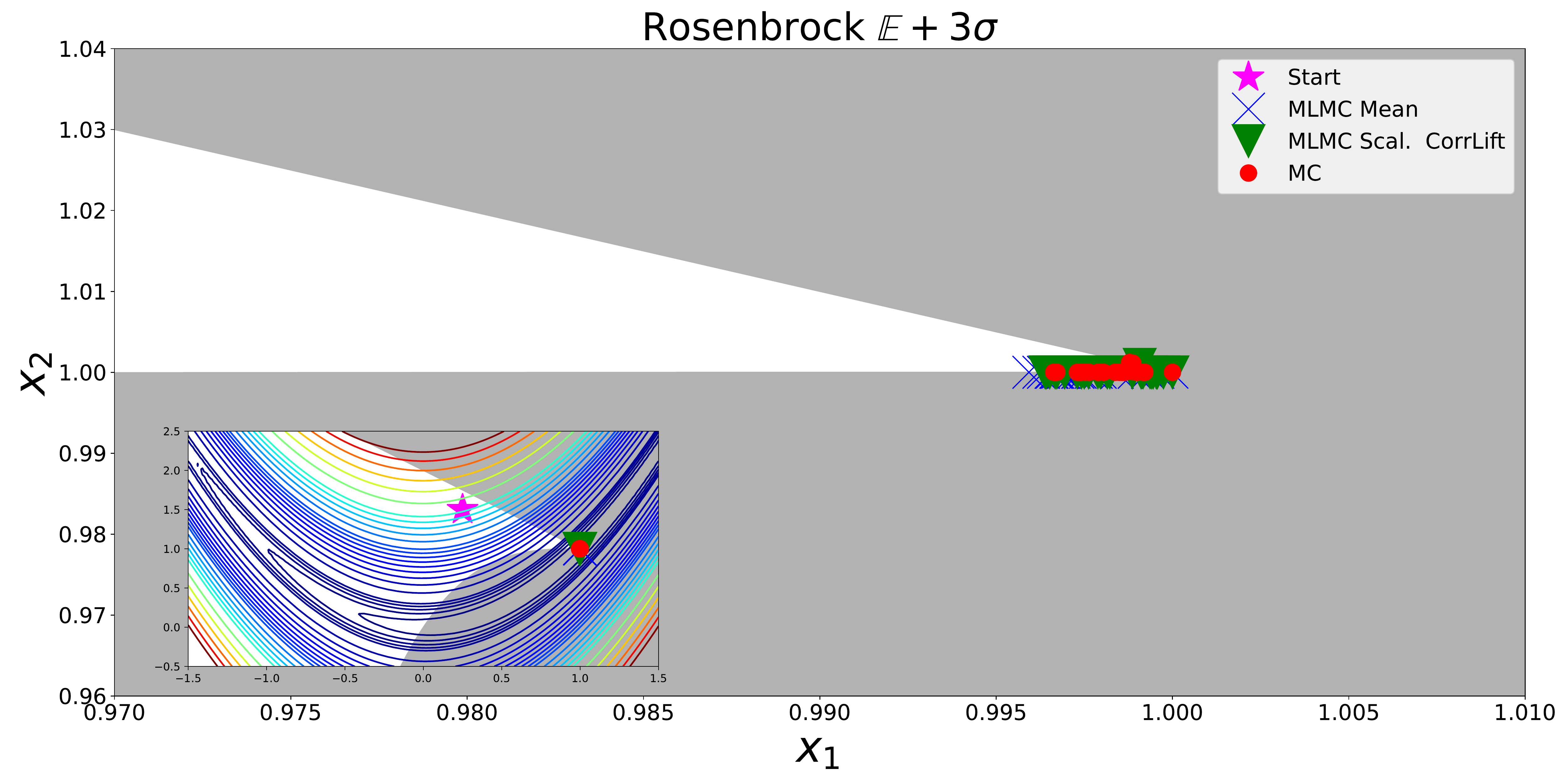}
    \caption{We use the Correlation lift approximation as described in Section~\ref{sssec:covariancecorrlift} to approximate the covariance term of Eq.~\eqref{eq:mlmcvarianceofscalarization}.}
    \label{fig:rosenbrockresults_meansigma_CorrLift}
\end{subfigure}
\hfill
\begin{subfigure}{0.49\textwidth}
\includegraphics[width=\textwidth]{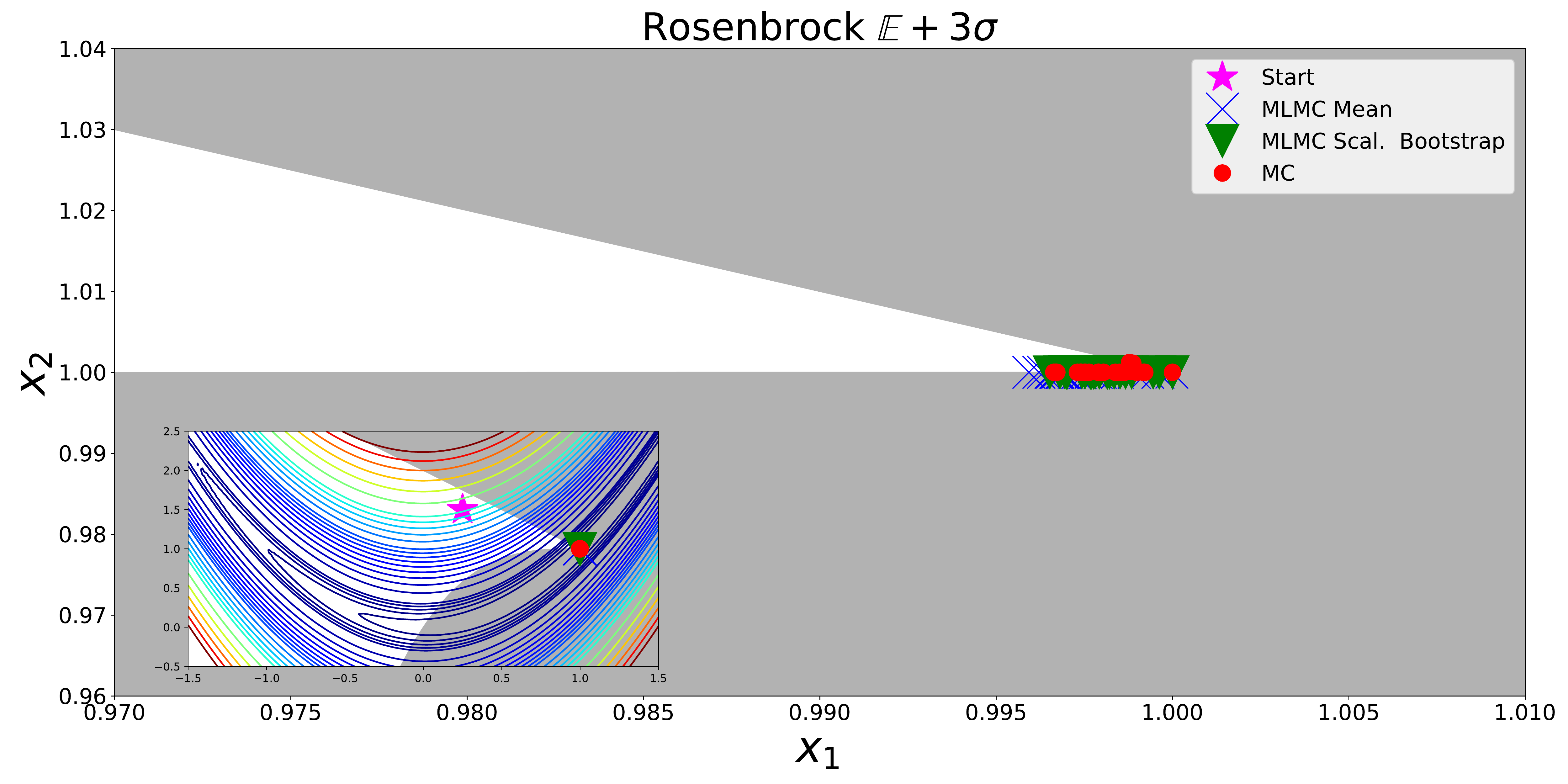}
    \caption{We use the Bootstrap approximation as described in Section~\ref{sssec:covariancebootstrap} to approximate the covariance term of Eq.~\eqref{eq:mlmcvarianceofscalarization}.}
    \label{fig:rosenbrockresults_meansigma_Bootstrap}
\end{subfigure}
\begin{subfigure}{0.49\textwidth}
	\includegraphics[width=\textwidth]{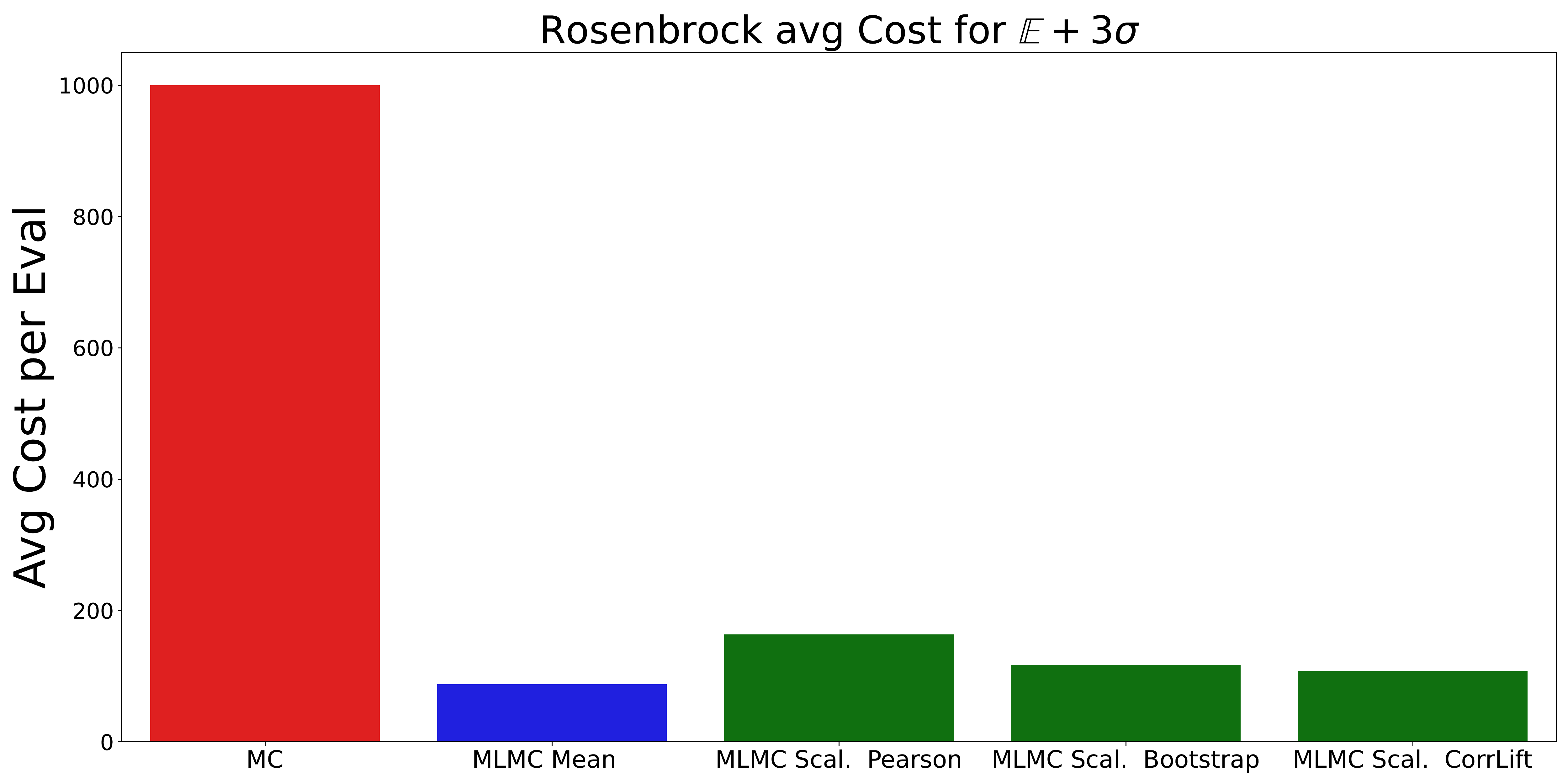}
    \caption{Cost average for a single evaluation for the different approaches. The cost are averaged over 25 optimization runs with 250 iterations each.}
    \label{fig:rosenbrockresults_meansigma_Cost}
\end{subfigure}
\caption{Optimization results for 25 individual runs after 250 iterations. In the three first figures, we see an enlarged area around the optimal design. The small figure in the bottom left shows the full domain. Each marker corresponds to the final design found by the individual run. We show the results when using a standard Monte Carlo estimator using $1000$ samples as red dot and compare to the final design found using a MLMC estimator targeting the mean as blue x and a MLMC estimator targeting the scalarization as green triangle. In the first three figures, we compare the different approximation strategies for the covariance term and the fourth figure shows a cost comparison for all three approaches averaged over 25 runs and 250 iterations.}
\label{fig:rosenbrockresults}
\end{figure}

For a more quantitative comparison of the results, we show again different metrics to compute the distance of the optimal designs for the MLMC approaches to the MC reference solution in Table~\ref{tbl:rosenbrock_results}. Section~\ref{ssec:distancemetrics} explains the metrics. In all metrics, we see that our newly developed approach is closer to the reference solution for all covariance approximations. The lowest (and thus best) value in each column is marked in bold.
\begin{table}[h]
\centering
\begin{tabular}{c | c | c | c | c }
Method & $\mathbb{X}$ & $\text{Dist}_{\text{C}}$  & $\text{Dist}_\sigma$ & $\text{Dist}^{\mathbb{X}}_{\text{RMSdev}}$  \\ \hline
MLMC Mean (20 iter) & $\mathbb{E}$ & 1.3519e-3 & [1.9332e-4, 3.1944e-4] & 1.3106e-3\\ \hline
MLMC Scalarization (20 iter, Pearson) & $\mathbb{\mathbb{E} + \alpha \sigma}$ &  1.2922e-3 & [3.1977e-4, 2.7096e-4] & 1.2249e-3\\ \hline
MLMC Scalarization (20 iter, CorrLift) & $\mathbb{\mathbb{E} + \alpha \sigma}$ & {1.2686e-3} & [3.1931e-4, \textbf{1.3793e-4}] & 1.1886e-3\\ \hline
MLMC Scalarization (20 iter, Bootstrap) & $\mathbb{\mathbb{E} + \alpha \sigma}$ & \textbf{1.1108e-3} & [\textbf{9.8059e-5}, 3.2164e-4]  & \textbf{9.7472e-4}\\ \hline
\end{tabular}
\caption{Quantitative comparison of the distance of final designs found to the Monte Carlo reference solution. Each row shows a different approach. Each column represent a different metric, with the second column showing the target of the estimator.}
\label{tbl:rosenbrock_results}
\end{table}

Also for this test case, we see improved performance in the optimization when using our newly developed estimators for this more challenging test case. Again, it is important to synchronize the MLMC allocation target with the respective formulation of the optimization problem, and then tailoring the estimation algorithm to robustly and efficiently obtain the desired accuracy.

\section{Conclusions}
\label{sec:conclusion}
In this work, we presented new multilevel Monte Carlo estimators for the statistics of variance, standard deviation, and the linear combination of mean and standard deviation, called scalarization. This required the derivation of variances for these estimators as one of the main contributions of this work. These statistics are especially relevant in optimization under uncertainty, where we not only optimize for the mean but also often include the standard deviation in the optimization problem to find a robust or reliable solution. The standard multilevel Monte Carlo estimator, optimized to provide a target precision for the mean, is in general inadequate for these statistics and the multilevel resource allocation problem needs to be modified to target these alternate statistical goals. During the optimization process, we build and evaluate the estimators repeatedly in each optimization step, which amplifies the need for an accurate estimator of the relevant statistic. 

We presented results on two benchmark problems: a one-dimensional constrained problem, called Problem 18, with one uncertain variable and four levels, and the two-dimensional constrained Rosenbrock function with three uncertain variables and three levels. We presented the sampling results for Problem 18, where we compare our new estimators with the standard multilevel Monte Carlo estimator for the mean to show their performance. We used a single-level Monte Carlo estimator as a reference and built the multilevel resource allocation to match the accuracy of the Monte Carlo estimator. We showed that our estimators more directly synchronize with the statistical goals of interest, while the estimator for the mean does not offer control beyond the mean estimator's variance. Moving to the optimization benchmark, we illustrated the impact of using these estimators. Employing our new estimators, we are able to control for any of the expanded set of statistical goals that we have focused on here. Similarly, for the constrained Rosenbrock function, we demonstrated a close match between our new estimators, while the standard multilevel Monte Carlo estimator for the mean was not reaching the targeted precision, which reflected also in a non-optimal solution. Regarding algorithmic choices, we saw improvements in the approximation quality when combining the approach presented in~\cite{Krumscheid2020} with numerical optimization. Based on the findings, we propose using an iterative approach to compute a more robust resource allocation. Finally, regarding the covariance term in the scalarization case, we presented three different approximations and argued for the use of the approximation called \textit{correlation lift} due to a good balance between approximation quality and computational cost.

We motivated the use of these new estimators for matching the corresponding OUU goals. The estimators are implemented in the Dakota software and can also be combined with the optimization method SNOWPAC as presented in this work. For future directions, we intend to extend these estimators for other formulations that are relevant for robust and reliable optimization problems. These problems include, e.g., the conditional value at risk or quantile estimation. Both formulations can be defined in the form of a sampling estimator, which should make an extension straightforward. The first work in this direction has been done by \cite{Ganesh2022} and shows promising results.

%% The Acknowledgements part is started with the command \acknowledgements;
%% acknowledgements are then done as normal sections before appendix
\acknowledgements
The authors were partially supported by the DOE SciDAC FASTMath institute. Sandia National Laboratories is a multi-mission laboratory managed and operated by National Technology and Engineering Solutions of Sandia, LLC., a wholly owned subsidiary of Honeywell International, Inc., for the U.S. Department of Energy's National Nuclear Security Administration under contract DE-NA-0003525. The views expressed in the article do not necessarily represent the views of the U.S. Department of Energy or the United States Government.

\bibliographystyle{unsrtnat}
\bibliography{References}

\end{document}